\documentclass[pra,aps,reprint,a4paper,superscriptaddress,floatfix,longbibliography]{revtex4-2}
\usepackage{times}
\usepackage{graphicx}
\usepackage{epsfig}
\usepackage{amsfonts}
\usepackage{amsmath}
\usepackage{amssymb}
\usepackage{dsfont}
\usepackage{amsthm}
\usepackage{color}
\usepackage{comment}
\usepackage{braket}
\usepackage{physics}
\usepackage{multirow}
\usepackage{diagbox}
\usepackage{multirow}
\usepackage{makecell}
\usepackage{enumerate}
\usepackage{breakurl}

\newtheorem{theorem}{Theorem}
\newtheorem{proposition}{Proposition}

\newtheorem{corollary}{Corollary}
\newtheorem{definition}{Definition}

\newtheorem{remark}{Remark}

\newcommand{\id}{\mathds{1}}

\usepackage{float}
\usepackage{rotating}
\usepackage[short,c3]{optidef}
\usepackage[linesnumbered,ruled]{algorithm2e}
\usepackage{scalerel}[2016/12/29]

\usepackage[colorlinks=true,linkcolor=blue,citecolor=blue,urlcolor=blue]{hyperref}

\begin{document}
\renewcommand{\arraystretch}{1.3}


\title{Optimal conversion of Kochen-Specker sets into bipartite perfect quantum strategies}


\author{Stefan Trandafir}
\affiliation{Departamento de F\'{\i}sica Aplicada II, Universidad de Sevilla, E-41012 Sevilla, Spain}

\author{Ad\'an~Cabello}
\email{adan@us.es}
\affiliation{Departamento de F\'{\i}sica Aplicada II, Universidad de Sevilla, E-41012 Sevilla,
Spain}
\affiliation{Instituto Carlos~I de F\'{\i}sica Te\'orica y Computacional, Universidad de
Sevilla, E-41012 Sevilla, Spain}


\begin{abstract}
Bipartite perfect quantum strategies (BPQSs) allow two players isolated from each other to win every trial of a nonlocal game. BPQSs have crucial roles in recent developments in quantum information and quantum computation. However, only few BPQSs with a small number of inputs are known and only one of them has been experimentally tested. It has recently been shown that every BPQS has an associated Kochen-Specker (KS) set. Here, we first prove that any BPQS of minimum input cardinality that can be obtained from a generalized KS set can also be obtained from a KS set of pure states. Then, we address the problem of finding BPQSs of small input cardinality starting from KS sets. We introduce an algorithm that identifies the BPQS with the minimum number of settings for any given KS set. We apply it to many well-known KS sets of small cardinality in dimensions $3$, $4$, $5$, $6$, $7$, and $8$. In each dimension, the algorithm either recovers the best BPQS known or find one with fewer inputs.
\end{abstract}


\maketitle


\section{Introduction}


\subsection{Nonlocal games}


Any quantum violation of a bipartite Bell inequality \cite{Bell:1964PHY} provides quantum advantage in a nonlocal game. That is, two players, Alice and Bob, who are not allowed to communicate with each other, can use a quantum correlation to win a game more often than using the best classical strategy. In every trial of the game, the referee gives Alice an input $x \in X$, and gives Bob an input $y  \in Y$. Inputs are distributed  with a probability $\pi(x,y)$ known by the players. For example, in the Clauser-Horne-Shimony-Holt (CHSH) game \cite{Clauser:1969PRL}, $x,y \in \{0,1\}$ and
$\pi(x,y)=\frac{1}{4}\ \forall (x,y)$.
In return, Alice gives the referee an output $a \in A$ and Bob gives the referee an output $b \in B$. The players win the game if the inputs and outputs satisfy a certain condition $W(x,y,a,b) \in \{0,1\}$. For example, in the CHSH game, $a,b \in \{0,1\}$ and
\begin{equation}
    W(x,y,a,b) = \left\{ \begin{array}{ll} 1 & \text{if } a \oplus b = xy, \\ 0 & \text{otherwise,} \end{array} \right.
\end{equation}
where $\oplus$ denotes the sum modulo two. Then, if the correlation $p(a,b|x,y)$ gives the probabilities of all possible combinations of inputs and outputs, the probability of winning the game is 
\begin{equation}
    \omega = \sum_{x,y,a,b} \pi(x,y)\,p(a,b|x,y)\,W(x,y,a,b).
\end{equation}
If a quantum correlation $p(a,b|x,y)$ violates a Bell inequality, then there is quantum advantage. For example, in the CHSH game, the maximum classical probability of winning is 
\begin{equation}
    \omega_C = \frac{3}{4} = 0.75,
\end{equation}
while the maximum quantum probability of winning is \cite{Tsirelson:1980LMP}
\begin{equation}
    \omega_Q = \frac{2+\sqrt{2}}{4} \approx 0.853.
\end{equation}
Alice's input cardinality is $|X|$ and Bob's input cardinality is $|Y|$. Hereafter, by minimizing the input cardinality of a bipartite nonlocal game, we mean minimizing the product $|X||Y|$.


\subsection{Bipartite perfect quantum strategies}
\label{sec:perf-quantum-strategies}


A bipartite nonlocal game has a bipartite perfect quantum strategy (BPQS) if
\begin{equation}
    \omega_C < \omega_Q = 1.
\end{equation}
Sometimes the term ``quantum pseudotelepathy'' \cite{GBT05} is used instead of BPQS. 
It is known that there are no BPQSs using qubit-qubit correlations; the simplest bipartite quantum system that allows for a BPQS is a qutrit-qutrit system \cite{Renner2004b}. It is also known that there are no BPQSs if $|X|=2$ \cite{Gisin:2007IJQI} or if $|A|=|B|=2$ \cite{CHTW04}. More recently, it has been proven that there are no BPQSs if $|X|=|Y|=3$, $|A|=4$, $|B|=2$ \cite{Liu:2023PRR}. The simplest example known of a BPQS is the magic square game \cite{Cabello:2001PRLa,Cabello:2001PRLb,Aravind:2004AJP} and requires $|X|=|Y|=3$, $|A|=|B|=4$. To define this game, it is convenient to take $x,y \in \{0,1,2\}$ and $a=({\cal A}, \alpha), b=({\cal B}, \beta) \in \{(1,1),(1,-1),(-1,1),(-1,-1)\}$.
The nine possible combinations of inputs are uniformly distributed. That is, $\pi(x,y)=\frac{1}{9}\ \forall (x,y)$.
The winning condition is 
\begin{equation}
    W = \left\{ \begin{array}{ll} 1 & \text{if }(x,y,{\cal A},\alpha,{\cal B},\beta)\text{ satisfies Table \ref{tab:wc-magicsquare},} \\ 0 & \text{otherwise.} \end{array} \right.
\end{equation}


\begin{table}[]
    \centering
    \begin{tabular}{ l l l l l l l }
     \hline \hline
     & \quad \quad \quad  & $y=0$ & \quad \quad \quad  & $y=1$ & \quad \quad \quad  & $y=2$ \\ \hline 
    $x=0$ & & ${\cal A}={\cal B}$ & & $\alpha={\cal B}$ & & ${\cal A} \alpha = {\cal B}$ \\ 
    $x=1$ & & ${\cal A}=\beta$ & & $\alpha=\beta$ & & ${\cal A} \alpha = \beta$ \\
    $x=2$ & & ${\cal A}={\cal B} \beta$ & & $ \alpha = {\cal B} \beta$ & & ${\cal A} \alpha = -{\cal B} \beta$ \\
    \hline \hline
  \end{tabular}
  \caption{The winning conditions of the magic square game.}
    \label{tab:wc-magicsquare}
\end{table}

For example, if $x=y=0$, the players win if and only if ${\cal A}={\cal B}$.
For this game \cite{Cabello:2001PRLb},
\begin{align}
\begin{split}
& \omega_C=\frac{8}{9} \approx 0.888, \\
& \omega_Q=1.
\end{split}
\end{align}

There are compelling arguments to conjecture that the magic square game is the bipartite nonlocal game with the minimum input cardinality that allows for a quantum BPQS \cite{Cabello:2023XXX}. As far as we know, the magic square correlation is the only BPQS that has been experimentally tested \cite{CinelliPRL2005,YangPRL2005,BarbieriOS2007,Aolita:2012PRA,Xu:2022PRL}. 


\subsection{Why are perfect quantum strategies important?} \label{sec:recent-results}


BPQSs have become important in the light of several recent results: 

\begin{enumerate}
    \item First, a quantum correlation $p(a,b|x,y)$ has local fraction $1$ \cite{Elitzur:1992PLA} or, equivalently, is ``fully nonlocal' \cite{Aolita:2012PRA} if and only if $p(a,b|x,y)$ is a BPQS \cite{Liu:2023PRR}.

    \item Second, a quantum correlation $p(a,b|x,y)$ is in a face of the nonsignaling polytope that does not contain local points \cite{Goh:2018PRA} if and only if $p(a,b|x,y)$ is a BPQS \cite{Liu:2023PRR}.

    \item Third, a quantum correlation $p(a,b|x,y)$ allows for a Greenberger-Horne-Zeilinger-like proof \cite{GHZ89} if and only if $p(a,b|x,y)$ is a BPQS \cite{Liu:2023PRR}.

    \item Fourth, a quantum correlation $p(a,b|x,y)$ violates Bell-like inequalities for arbitrary joint relaxations of the assumptions of measurement independence and parameter independence \cite{Shimony:1993} if and only if $p(a,b|x,y)$ is a BPQS or $\epsilon$-close to a BPQS \cite{Vieira;2024XXX}.

    \item Fifth, not all infinite-dimensional commuting correlations can be approximated by sequences of finite-dimensional tensor product correlations \cite{Ji:2021CACM}. The proof actually states that there is an infinite-dimensional commuting BPQS for a nonlocal game whose maximum probability of winning with tensor product correlation is $< 1/2$ \cite{Ji:2021CACM}. This result has deep consequences for physics \cite{cabello2023logicalpossibilitiesphysicsmipre}.

    \item Sixth, BPQSs help us to understand the origin of the bounds of quantum correlations. For example, the correlation providing the maximum quantum violation of the CHSH inequality can be ``extended'' to a BPQS in such a way that the latter dictates the limit (the Tsirelson bound \cite{Tsirelson:1980LMP}) of the former \cite{Cabello:2015PRL, Cabello:2017}. More generally, any Bell or Kochen-Specker (KS) contextuality scenario can be extended to a KS contextuality scenario in such a way that the fact that the correlations in the extended scenario attain $\omega_Q=1$ dictate the bounds in the initial scenario \cite{Cabello:2019PRA}.

    \item Seventh, the discovery that there are problems that can be solved by a quantum algorithm with certainty and in constant time for any input size, but require a time logarithmic in the length of the input for any classical circuit that solves them with a sufficiently high probability \cite{Bravyi:2018SCI} was based on perfect quantum strategies.

    \item Eighth, some quantum information protocols require BPQSs \cite{Cubitt:2010PRL,Horodecki:2010XXX,Vidick:2017XXX,Jain:2020IEE,Zhen:2023PRL,Bharti:2023XXX}.
\end{enumerate}


\subsection{Kochen-Specker sets}
\label{subsec:KSsets}


Kochen-Specker (KS) sets are sets of observables represented in quantum theory by rank-one projectors introduced in Ref. \cite{Kochen:1967JMM} for proving the impossibility of noncontextual hidden-variable interpretations of quantum theory. They are defined as follows.

\begin{definition}[KS set \cite{Kochen:1967JMM,Renner2004b,Pavicic:2005JPA}]
\label{def:gksset} 
A KS set in finite dimension $d \ge 3$ is a finite set of rank-one projectors $\mathcal{V}$ in a Hilbert space of dimension $d$ which does not admit an assignment $f: \mathcal{V} \rightarrow \{0,1\}$ satisfying: 

\indent I. $f(u) + f(v) \leq 1$ for each pair of orthogonal projectors $u, v \in \mathcal{V}$, 

\indent II. $\sum_{u \in b} f(u) = 1$ for every set $b \subset \mathcal{V}$ of mutually orthogonal projectors whose sum is the identity.
\end{definition}

Rank-one projectors can be identified with the one-dimensional subspace they project onto (which is a line). We can thus identify rank-one projectors with vectors, and sets of mutually orthogonal rank-one projectors whose sum is the identity with orthogonal bases. 
In this paper, it will be convenient to view a KS set as a pair $(\mathcal{V}, \mathcal{B})$, where $\mathcal{B}$ is the set of orthogonal bases (or, equivalently, the set of sets of mutually orthogonal projectors summing up to the identity) formed by the elements of $\mathcal{V}$. For simplicity, sometimes, we refer to the set of bases as a KS set.

\begin{definition}[Generalized KS set \cite{Xu:2022PRL}]
    The definition is identical to the definition of KS set, but removing the requirement that the projectors are of rank-one.
    \end{definition}

Generalized KS sets are called ``projective'' KS sets in Ref. \cite{MancinskaScarpaSeverini:2013} (and the term generalized KS set is used with a different meaning in Ref. \cite{MancinskaScarpaSeverini:2013}). For a finite-dimensional Hilbert space of dimension $d$, we let $\id_d$ denote the identity operator.

KS sets have been used to experimentally test state-independent contextuality \cite{Cabello:2008PRL,Badziag:2009PRL,Kirchmair:2009NAT,Amselem:2009PRL,D'Ambrosio:2013PRX}, and, some of them, can be certified with any full rank state \cite{xu2023stateindependent}, and used to self-test specific multipartite high-dimensional quantum states \cite{Saha:2025XXX}.


\subsection{Connections between BPQSs and KS sets}
\label{subsec:connections}


The first result that connects BPQSs and KS sets is due to Stairs \cite{Stairs:1983PS,Brown:1990FPH} and Heywood and Redhead \cite{HR83,Redhead:1987}, and has been refined by others \cite{Elby:1992PLA,Renner2004b,CHTW04,Aolita:2012PRA}. It states that a BPQS can be produced by measuring a KS set of dimension $d \ge 3$ on a qudit-qudit maximally entangled state. There are two versions of the result: Method I: each of the two players measures all the elements of $\mathcal{B}$ \cite{HR83}. Method II: one of the players measures all the elements of $\mathcal{B}$ and the other player measures all the elements of $\mathcal{V}$ \cite{Elby:1992PLA}. Therefore, in Method~I the resulting BPQS has input cardinality $|\mathcal{B}|$ and output cardinality $d$ for both players. In Method~II, one player has input cardinality $|\mathcal{B}|$ and output cardinality $d$, and the other player has input cardinality $|\mathcal{V}|$ and output cardinality $2$.

A more recent result establishes a stronger connection between BPQSs and KS sets \cite{Cabello:2023XXX}. 

\begin{theorem} \label{mt}
    A correlation $p(a,b|x,y)$ allows for a BPQS if and only if there is a bipartite-KS (B-KS) set associated with it.
\end{theorem}

\begin{definition}[Bipartite KS set]
\label{def:bks} 
 For a generalized KS set $(\mathcal{V}, \mathcal{B})$, and subsets $S_A, S_B \subseteq \mathcal{B}$, define 
\begin{equation}
   \Omega_{A, B} := \{(v, b) : b \in S_A \cup S_B, v \in b\}.
\end{equation}
We say that $S =(S_A, S_B)$ is a \emph{bipartite KS} set if there is no assignment $f : \Omega_{A,B} \rightarrow \{0,1\}$ satisfying:

 I'. $f(u,b) + f(u',b') \leq 1$ for each $b \in S_A, b' \in S_B, u \in b, u' \in b'$ with $u,u'$ orthogonal,
 
 II'. $\sum_{u \in b} f(u,b) = 1$ for each $b \in S_A$, and $\sum_{u' \in b'} f(u',b') = 1$ for each $b' \in S_B$.
\end{definition} 

In the case that $(\mathcal{V}, \mathcal{B})$ is a KS set, $S_A,S_B$ are simply sets of orthogonal bases of $\mathcal{B}$. Moreover, every B-KS set $(S_A, S_B)$ defines a generalized KS set where $\mathcal{B} = S_A \cup S_B$, and $\mathcal{V}$ is the set of projectors appearing in some basis of $\mathcal{B}$. Note however that it is possible to obtain a B-KS set $(S_A, S_B)$ from some KS set $K$, such that $S_A \cup S_B$ defines a \emph{different} KS set $K'$. 

On the other hand, given a KS set $(\mathcal{V}, \mathcal{B})$, there are many ways to obtain B-KS sets from $(\mathcal{V}, \mathcal{B})$.

Given a B-KS set $(S_A, S_B)$ in dimension $d$, it can be proven that there is a BPQS associated with it by using $(S_A, S_B)$ to construct a nonlocal bipartite game with a BPQS. The construction is as follows: With probability $\pi(x,y)$, Alice is given a basis $x \in X=S_A$ and Bob is given a basis $y \in Y=S_B$. Alice outputs a vector $a \in x$ and Bob outputs a vector $b \in y$. They win if and only if $a$ and $b$ are orthogonal. By definition of B-KS set, this game does not admit a perfect classical strategy. However, the following strategy is a BPQS: Alice and Bob share a qudit-qudit maximally entangled state in dimension $d$. For each $x$ ($y$), Alice (Bob) measures the orthogonal projectors onto each of the vectors in $x$ ($y$) and outputs the vector associated with the result that she (he) has obtained.

A sketch of the proof that every BPQS defines a B-KS set is the following. Suppose that $p(a,b|x,y)$ yields a BPQS. Every quantum correlation $p(a,b|x,y)$ in a finite-dimensional Hilbert space can be expressed as
\begin{equation}
    p(a,b|x,y) = \langle \psi | \Pi_{a|x} \otimes \Pi_{b|y} | \psi \rangle,
\end{equation}
where $\ket{\psi} \in \mathcal{H} = \mathcal{H}_A \otimes \mathcal{H}_B$ is a pure state, $\Pi_{a|x}$ are projective measurements for Alice, and $\Pi_{b|y}$ are projective measurements for Bob \cite{Brunner:2014RMP,Paddock:2022XXX}. 
Therefore, we can define $S_A := \{\{\ket{\psi_{a|x}} : a \in A\} : x \in X\}$, with
\begin{equation}
    \ket{\psi_{a|x}} :=\frac{(\Pi_{a|x} \otimes \id_B) |\psi\rangle}{\sqrt{\langle \psi | (\Pi_{a|x} \otimes \id_B) | \psi \rangle}},
\end{equation}
where $\id_B$ is the identity in $\mathcal{H}_B$,
and $S_B := \{\{\ket{\psi_{b|y}} : b \in B\} : y \in Y\}$, with
\begin{equation}
    \ket{\psi_{b|y}} :=\frac{(\id_A \otimes \Pi_{b|y}) |\psi\rangle}{\sqrt{\langle \psi | (\id_A \otimes \Pi_{b|y}) | \psi \rangle}},
\end{equation}
where $\id_A$ is the identity in $\mathcal{H}_A$.
If $p(a,b|x,y)$ is a BPQS, then no local hidden-variable model can explain the $0$'s in $p(a,b|x,y)$ (see Ref. \cite{Cabello:2023XXX} for details), which is equivalent to that $(S_A,S_B)$ is a B-KS set (see Ref. \cite{Cabello:2023XXX} for details).


\subsection{Aim and structure}


Only a few BPQSs with a small number of inputs are known \cite{Cabello:2001PRLb,Mancinska:2007}. 
Among them, only the magic square correlation has a small number of inputs and has been experimentally tested \cite{CinelliPRL2005,YangPRL2005,BarbieriOS2007,Aolita:2012PRA,Xu:2022PRL}. One reason why we do not have more examples is that the Methods I and II mentioned before for, given a KS set, producing a BPQS lead to BPQSs with many inputs. 
There is also a method for constructing BPQSs by ``parallelization'' \cite{Coladangelo2016arxiv,Coudron2016arxiv,Araujo:2020Quantum} in order to generate BPQSs in higher dimensions. This construction relies on pre-existing BPQSs and the dimension increases exponentially in the number of parallel copies taken so that BPQSs are produced only for certain dimensions (those of the form $d^N$, where $N$ is the number of parallel copies of a BPQS of dimension $d$).

Theorem~\ref{mt} says that every BPQS with input sets $X, Y$ has a B-KS set $(S_A, S_B)$ associated (in particular, of the same dimension, and satisfying $|S_A| = |X|, |S_B| = |Y|$). In Sec.~\ref{Sec:min-KS}, we prove that for each such B-KS set there exists some B-KS set $(S'_A, S'_B)$, of the same dimension, satisfying $|S_A| = |S'_A|, |S_B| = |S'_B|$, but consisting of rank-one projectors  (Theorem \ref{theo:min-KS}). Therefore, the study of allowable input cardinalities for BPQSs may be achieved through the study of those B-KS sets that arise from KS sets.

Our main goal is to identify BPQSs of small input cardinality for each dimension $d$, with the ultimate goal being to find the minimum for each $d$. We thus exploit the connection provided by Theorem \ref{theo:min-KS} and compute optimal B-KS sets for a variety of KS sets in small dimensions. This approach allows us to leverage more than 50 years of research in KS sets in order to find BPQSs of small input cardinality in each such dimension.

Thus, we take KS sets $(\mathcal{V}, \mathcal{B})$, and then compute the associated B-KS sets $(S_A, S_B)$ with $S_A, S_B \subseteq \mathcal{B}$ of minimum cardinality $|S_A||S_B|$. 

Finding such a B-KS set for a given KS set is not a straightforward task: there are (roughly) 
\begin{equation*}
    \binom{2^{|\mathcal{B}|}}{2}
\end{equation*}
pairs $(S_A, S_B)$ that are potential candidates. Moreover, it is nontrivial to check whether a given pair is indeed a B-KS set because this amounts to checking whether a polytope contains integer points (a problem which, in general, is NP complete \cite{Schrijver:2003}). 
Therefore, an important question is how to obtain the BPQS of minimum input cardinality associated with a given KS set. In Sec.~\ref{Sec:Algorithm}, we introduce an algorithm for this purpose.

Given that there are infinitely many KS sets, it is sensible to ask which ones should be tested.
Equation \eqref{eq:lowerbound} suggests that the BPQSs with the smallest input cardinalities allowed by quantum theory in each dimension will be associated with KS sets of relatively small cardinality. Additionally, optimal BPQSs may be rapidly computed in these scenarios. 

For these reasons, in Sec.~\ref{Sec:Correlations}, we apply the algorithm described in Sec.~\ref{Sec:Algorithm} to several famous KS sets of small input cardinality in dimensions 3 to 8. In all the cases, the algorithm produces BPQSs with smaller input cardinalities than with Methods I and II described in Sec.~\ref{subsec:connections}. Moreover, in each dimension, the algorithm either produces the best BPQS known or finds one with fewer inputs. This is discussed in Sec.~\ref{Sec:Discussion}, together with some open questions. 


\section{Kochen-Specker sets and bipartite perfect quantum strategies} \label{Sec:min-KS}


Here, we extend the connection between BPQSs and B-KS sets given by Theorem~\ref{mt} and use it to investigate the following problem: Given some finite dimension $d \geq 3$, what is the minimum input cardinality for which there exists a BPQS with local systems of dimension $d$.

\begin{definition}[Optimal B-KS set]
    A B-KS set $(S_A, S_B)$ 
    is \emph{optimal} if the product $|S_A||S_B|$ is as small as possible.
\end{definition}

Note that ``optimal'' here is in reference to the KS set $(\mathcal{V}, \mathcal{B})$ that the B-KS is obtained from, and not to the dimension $d$. The minimum input cardinality is thus the smallest size $|S_A||S_B|$ over all optimal B-KS sets $(S_A, S_B)$ in dimension $d$.


\subsection{KS sets are sufficient for minimum input cardinality BPQS}


We now show that the minimum input cardinality B-KS set for a given dimension $d$ will always be obtained from a KS set. If each of the projectors in a B-KS set $(S_A, S_B)$ are of rank $1$, we say that $(S_A, S_B)$ is a \emph{rank-one B-KS set}. For such a B-KS set, the set of projectors in $S_A \cup S_B$ form a KS set.

\begin{theorem} \label{theo:min-KS}
    Let $(S_A, S_B)$ be a B-KS set. Then, there exists a rank-one B-KS $(S'_A, S'_B)$ of the same dimension, and with $|S'_A| = |S_A|, |S'_B| = |S_B|.$
\end{theorem}

\begin{proof}
Reference \cite[Appendix A]{MancinskaScarpaSeverini:2013} it is proven that one can always construct a KS set from any generalized KS set $(\mathcal{V}, \mathcal{B})$ (in fact one may, in principle, create infinitely many KS sets, but for our purposes a single one suffices). The construction is as follows: for each rank-$k$ projector $P \in \mathcal{V}$, let $Q$ be the $k$-dimensional subspace that $P$ projects onto, and let $\beta(P) := \{v_1, \dots, v_k\}$ be an orthonormal basis of $Q$. Then $\mathcal{V}' := \bigcup_{P \in \mathcal{V}} \beta(P)$ is a KS set. In words, we simply replace each projector $P$ by an orthonormal basis spanning the subspace projected on by $P$. 

Clearly, the generalized KS set and the KS set are of the same dimension. While there may be more sets of orthogonal projectors summing to the identity (these are orthogonal bases since all projectors are rank-one) that can be formed by the elements of $\mathcal{V}'$ than of $\mathcal{V}$, it is sufficient to consider only the set of orthogonal bases $\mathcal{B}'$ obtained by replacing projectors in the elements $b \in \mathcal{B}$. In other words, setting $\mathcal{B}' := \{\bigcup_{v \in b} \beta(v) : b \in \mathcal{B}\}$, we have that $(\mathcal{V}', \mathcal{B}')$ is a KS set.
This gives a 1-1 correspondence between sets of orthogonal projectors $\mathcal{B}$ summing to $\mathds{1}_d$ of the generalized KS set $(\mathcal{V}, \mathcal{B})$ and the orthogonal bases $\mathcal{B}'$ of the KS set $(\mathcal{V}', \mathcal{B}')$.

Let us consider a B-KS set $(S_A, S_B)$ obtained from $(\mathcal{V}, \mathcal{B})$. Construct the sets $(S'_A, S'_B) \subseteq \mathcal{B}'$, via
\begin{subequations}
\begin{align}
    S'_A &= \{\bigcup_{P \in b} \beta(P) : b \in S_A\},\\
    S'_B &= \{\bigcup_{P \in b} \beta(P) : b \in S_B\},
\end{align}
\end{subequations}
so that $S'_A$ ($S'_B$) is the set of bases in $\mathcal{B}'$ corresponding to the orthogonal bases in $S_A$ ($S_B$). We show that $(S'_A, S'_B)$ is a (rank-one) B-KS set.

Assume towards a contradiction that $(S'_A, S'_B)$ is not a B-KS set. Then there is some function $f' : \Omega_{A', B'}$ satisfying (I') and (II') of Definition \ref{def:bks}. Construct the function $f : \Omega_{A,B} \to \{0,1\}$ via 
\begin{equation}
    f(u,b) = \sum_{u' \in \beta(u)} f'(u',b') 
\end{equation}
where $b'$ is the orthogonal basis corresponding to orthogonal basis $b$. One may check that $f$ does indeed satisfy conditions (I') and (II') of Definition \ref{def:bks}. But this is a contradiction since we assumed that $(S_A, S_B)$ is a B-KS. Therefore, it follows that $(S'_A, S'_B)$ must be a B-KS if $(S_A, S_B)$ is. 
\end{proof}


\subsection{Resulting bounds}


The previous result allows us to obtain some simple bounds on
the input cardinalities $|X|, |Y|$ for BPQSs in some dimension $d.$
Whenever $(S_A, S_B)$ is a rank-one B-KS, it must be that 
\begin{equation}
    \{v : v \in b \text{ for some } b \in S_A \cup S_B\}
\end{equation}
is a KS set $\mathcal{V}$. Therefore, $|S_A \cup S_B| \geq |\mathcal{V}|/d$, and so 
\begin{equation}
    |S_A| + |S_B| \geq \frac{|\mathcal{V}|}{d}. \label{eq:lowerbound}
\end{equation}

\begin{corollary}
    Let $v_d$ be the smallest size of a KS set in dimension $d$. Then, for any B-KS $(S_A,S_B)$ of dimension $d$, we must have 
    \begin{equation}
        |S_A| + |S_B| \geq \frac{v_d}{d}.
    \end{equation}
    Equivalently, for any BPQS in dimension $d$ with input sets $X,Y$, we must have
    \begin{equation}
        |X| + |Y| \geq \frac{v_d}{d}.
    \end{equation}
\end{corollary}

As an example, in dimension $3$ it is known that the smallest KS set must have at least $24$ vectors \cite{Kirchweger:2023,Li:2023XXX}, so in this case we must have $|X| + |Y| \geq 24/3 = 8$.

Additionally, whenever we have some rank-one B-KS $(S_A, S_B)$ in a given dimension, Eq. \eqref{eq:lowerbound} provides an upper bound on the sizes of KS sets that can yield smaller B-KS set (see Sec. \ref{sec:KP-40} for a concrete example).

\begin{remark}
Henceforth, all B-KS sets in this article are rank-one B-KS sets. For simplicity, we will refer to them as B-KS sets.
\end{remark}


\section{Algorithm}
\label{Sec:Algorithm}


Here, we describe an algorithm that takes as input a KS set $(\mathcal{V}, \mathcal{B})$ and outputs a rank-one B-KS set $(S_A, S_B)$ with $S_A, S_B \subseteq \mathcal{B}$ minimizing the product $|S_A||S_B|$. The key ingredient of the algorithm is a filter which identifies and removes subsets of $\mathcal{B}$ that cannot take part in any B-KS set. This allows one to dramatically reduce the number of pairs that must be checked.


\subsection{Checking whether $(S_A, S_B)$ is a B-KS set}


To check whether a given $(S_A, S_B)$ is a B-KS set, we recast the problem as the following integer linear program (ILP):
\begin{flushleft}%
\begin{mini*}{}{1} {\label{lp:canonical}}{}
    \addConstraint{\sum_{v \in b} w_{v,b} = 1}{\quad \forall b \in S_A}
    \addConstraint{\sum_{v' \in b'} w_{v',b'} = 1}{\quad \forall b' \in S_B}
        \addConstraint{w_{v,b} + w_{v',b'} \leq 1}{\quad \text{ whenever } <v,v'> = 0}
        \addConstraint{w_{v,b}},
        w_{v',b'}{\in \{0,1\}},{}
\end{mini*} 
\end{flushleft}
where $<\cdot, \cdot>$ denotes the standard inner product. The condition $\min 1$ is put artificially: the only goal of the ILP is to determine whether there is a point satisfying all of the constraints. In principle, we could replace $1$ with any linear function in the variables $w_{v,b}$ (henceforth just $w$). In essence, we are solving a feasibility problem, but just casting it as an ILP in order to use existing machinery (to be precise we use the GLPK solver in \emph{Sagemath} \cite{sagemath}).
Then, $(S_A, S_B)$ constitutes a B-KS set when the problem described above is infeasible (i.e., there is no $w$ satisfying the constraints). From a feasible solution, one obtains a function $f : \Omega_{A,B} \to \{0,1\}$ satisfying I' and II' of Definition \ref{def:bks} by setting $f(v,b) := w_{v,b}$ for each pair $(v,b) \in \Omega_{A,B}$.
At this point, one could, in principle, iterate over the 
\begin{equation*}
    \binom{2^{|\mathcal{B}|}}{2}
\end{equation*}
pairs $(S_A, S_B)$ in increasing order of cardinality $|S_A||S_B|$ outputting an optimal B-KS once it is clear that no B-KS sets of equal or higher cardinality remain to be checked. However, as we shall see, it is possible to significantly reduce the number of pairs that must be iterated through (and thus also the number of calls to the ILP).

We make the following observations. First, if $(S_A, S_B)$ is a B-KS set and $S_A \subseteq S'_A \subseteq \mathcal{B}, S_B \subseteq S'_B \subseteq \mathcal{B}$, then $(S'_A, S'_B)$ is also a B-KS set. Second, if $(\mathcal{V}, \mathcal{B})$ defines a KS set, we see (trivially) that $(\mathcal{B}, \mathcal{B})$ is a B-KS set. Third, if $(S_A, S_B)$ is a B-KS set, then so is $(S_B, S_A)$. Therefore, we always assume that $|S_A| \leq |S_B|$ without loss of generality.

Together, these observations allow us to conclude that, for a given subset $S_A \subseteq \mathcal{B}$, if $(S_A, \mathcal{B})$ is not a B-KS set, then $(S_A, S_B)$ (and also $(S_B, S_A)$) will not be a B-KS set for any choice of $S_B$. This leads to the following definitions:

\begin{definition}[B-KS capable subset of bases]
    A subset $S_A \subseteq \mathcal{B}$ is said to be \emph{B-KS capable} if $(S_A, \mathcal{B})$ is a B-KS set. 
\end{definition}

\begin{definition}[Essential subset of bases]
A B-KS capable subset $S_A \subseteq \mathcal{B}$ is said to be essential if there is no $S'_A \subseteq S_A$ which is also B-KS capable.
\end{definition}


\subsection{Computing optimal B-KS sets}


Algorithm \ref{alg:bks} takes as input a KS set and outputs an optimal B-KS set.
\begin{algorithm}
    \SetKwInOut{Input}{Input}
    \SetKwInOut{Output}{Output}
    \SetKw{Or}{or}
    \Input{A KS set $(\mathcal{V}, \mathcal{B})$}
    \Output{Optimal B-KS set $(S_A, S_B)$}
    
    $C \gets [\ ], k \gets 3, k_{\min} \gets |\mathcal{B}|, \omega \gets |\mathcal{B}|^2, b \gets \text{False}$\;
    
    \For{$S_A \subseteq \mathcal{B}$ \text{ with } $|S_A|=k$}{
        \If{$S_A$ is B-KS capable}{
            \If{$b = \text{False}$}{
                $k_{\min} \gets |S_A|$, $b \gets \text{True}$\;
            }
            \For{$S_B \in C$}{
                \If{$|S_A||S_B| > \omega$}{break\;} 
                \If{$(S_A, S_B)$ is B-KS}{
                    $(S_A^*, S_B^*) \gets (S_A, S_B)$\;
                    $\omega \gets |S_A^*||S_B^*|$\;
                    break\;
                }
            }
            $C \gets C + [S_A]$\tcp*{append $S_A$ to $C$}
        }
    }
    $k \gets k+1$\;
    \eIf{$k \cdot k_{\min} \geq \omega$}{
        \Return{$(S_A^*, S_B^*)$}\;
    }{
        \text{go to 2}\;
    }
    \caption{Compute an optimal B-KS set.}
    \label{alg:bks}
\end{algorithm}

We iterate over all subsets of $\mathcal{B}$ in order of increasing size in order to identify and compute the B-KS capable sets. We begin with size $k=3$, since there can be no B-KS capable sets of smaller size. The first time we encounter a B-KS capable set, we store its size (as $k_{\min}$), and mark that we have encountered a B-KS capable set (by setting the boolean $b$ to True). Each time we encounter a B-KS capable set $S_A$, we store it (in the list $C$) and test it against previously encountered B-KS capable sets $S_B$ in order to check whether they form a B-KS set. By design, the B-KS capable sets found are stored in order of increasing size. Therefore, we stop looping over $S_B \in C$ in two different scenarios: (i) if $|S_A||S_B|$ is at least as large as the size ($\omega$) of the best B-KS capable set found thus far, or (ii) if $(S_A,S_B)$ forms a B-KS set. In the case of (i), we continue to the following choice of $S_A$, and in the case of (ii) we mark the B-KS set $(S_A,S_B)$ as being optimal (by saving it as $(S_A^*, S_B^*)$ and setting $\omega$ to $|S_A^*||S_B^*|$) and continue to the next choice of $S_A$. Whenever the next choice of $S_A$ has strictly larger size than the previous, we check whether we have already encountered the optimum. This is done by verifying that $|S'_A||S'_B| \geq \omega$ for all the choices of $(S'_A,S'_B)$ that we would encounter. Since $|S'_A| \geq k$, and $|S_B| \geq k_{\min}$ this amounts to checking whether or not $k\cdot k_{\min} \geq \omega$ (which we perform in Line 21).

Let $k_{\max}$ denote the largest $k$ checked by the algorithm, so that $k_{\min}k_{\max}$ is at least as large as the size of an optimal B-KS set.
In the worst case, the algorithm is exponential in $|\mathcal{B}|$. The algorithm checks at most $\binom{\sum_{k=3}^{k_{\max}} \binom{|B|}{k}}{2}$ pairs $(S_A,S_B)$, and there are at most $\sum_{k=3}^{k_{\max}} \binom{|B|}{k}$ calls to the ILP \emph{``is BKS-capable''}, and at most $\binom{\sum_{k=3}^{k_{\max}} \binom{|B|}{k}}{2}$ calls to the ILP \emph{``is B-KS''}. In practice, the filter ``is BKS-capable'' removes a significant number of the pairs visited (and thus also to the calls to ``is BKS-capable''). As an example, one of the KS sets we study in Sec. \ref{Sec:Correlations}, CK-33, has $|\mathcal{B}| = 20$ orthogonal bases, $k_{\min} = 7$ and $k_{\max} = 13.$ Therefore, the bound we just gave would indicate that there are at most $\binom{\sum_{k=7}^{13} \binom{20}{k}}{2} = 4.3 \times 10^{11}$ calls to ``is BKS''. However, after filtering via ``is BKS-capable'', we find that there are at most $9.18 \times 10^{8}$ calls necessary --- three orders of magnitude less.

For the function ``is B-KS capable'' in Algorithm \ref{alg:bks}, one could test whether $(S_A, \mathcal{B})$ is a B-KS set via the B-KS ILP formulation, taking $S_B = \mathcal{B}$. However, we can simplify this formulation, as we shall see in the following section.


\subsection{Checking whether $S_A$ is B-KS capable}


For some $T \subseteq \mathcal{B}$, we define $\Gamma(T) := \{v : v \in b \text{ for some } b \in T\}$ and call this the \emph{groundset of $T$}. 
To check whether a given $(S_A, S_B)$ is B-KS capable, we recast the problem as the following ILP:


\begin{flushleft}%
\begin{mini*}{}{1}{}{} {\label{lp:canonical}}{}
	\addConstraint{\sum_{v \in b} w_{v} = 1}{\quad \forall b \in S_A}
    \addConstraint{\sum_{v' \in b'} w_{v'} \geq 1}{\quad \forall b' \in \mathcal{B} \setminus S_A}
        \addConstraint{w_{v} + w_{v'} \leq 1}{\quad \forall v \in \Gamma(A), v' \in \Gamma(\mathcal{B} \setminus S_A) \text{ with } <v,v'> = 0}
        \addConstraint{w_v \in \{0,1\} \text{ for each } v \in \mathcal{V}}.{}
\end{mini*}
\end{flushleft}


\subsection{Equivalence between the two formulations}


That the two formulations are equivalent (i.e., checking that $(S_A, \mathcal{B}$) is a B-KS set using the first ILP or checking that $S_A$ is a B-KS capable set using the second ILP) can be proven by showing that a feasible solution of one is a feasible solution of the other, and vice versa. 

\begin{proposition}
    Let $(\mathcal{V}, \mathcal{B})$ be a KS set, and let $S_A \subseteq \mathcal{B}$. Then, the B-KS ILP with $(S_A, \mathcal{B})$ has a feasible solution if and only if the B-KS capable ILP has a feasible solution with $S_A$.
\end{proposition}

\begin{proof}
Let $w := (w_{v,b})_{v \in b, b \in \mathcal{B}}$ be a feasible solution of the B-KS ILP formulation. We construct a solution $w'$ for the B-KS capable ILP formulation. For each of the vectors $v \in \Gamma(S_A)$, we must have that
\begin{equation}
    w_{v,b_1} = w_{v,b_2}
\end{equation}
for any choice of $b_1, b_2 \in \mathcal{B}$ with $v \in b_1 \cap b_2$. Therefore for each $v \in \Gamma(S_A)$, we may set the value of $w'_v$ to be this common value. Now consider some $v' \in \Gamma(\mathcal{B} \setminus S_A)$. If $w_{v',b_1} = w_{v',b_2}$ for each pair $b_1,b_2 \in \mathcal{B} \setminus S_A$, then set $w'_{v'}$ to be this common value. Otherwise, $w_{v',b} \neq w_{v',b'}$ for some pair $b,b' \in \mathcal{B} \setminus S_A$ with $v' \in b \cap b'$. By the orthogonality condition of the B-KS ILP formulation, $v'$ cannot be orthogonal to any $v \in \Gamma(S_A)$ with $w'_v = 1$. Thus setting $w'_{v'} = 1$ we satisfy each of the constraints of the B-KS capable set ILP formulation: each basis in $S_A$ has exactly one vector marked with a $1$, the bases in $\mathcal{B} \setminus S_A$ each have at least one vector marked with $1$, and no pair of vectors $v \in \Gamma(S_A), v' \notin \Gamma(\mathcal{B} \setminus S_A)$ marked with $1$ are orthogonal.

Now let us consider the reverse direction --- that is, we have a feasible solution $w' := (w_v)_{v \in \mathcal{V}}$ of the B-KS capable ILP. Consider $w''$ obtained by setting $w''_{v,b} := w'_{v}$ for each $v \in \mathcal{V}$. The first condition of the B-KS ILP formulation is satisfied by $w''$ as is the third. On the other hand the second condition may be violated in the case that $b' \in \mathcal{B} \setminus S_A$. To satisfy this condition, we can modify the infeasible solution $w''$ ``basis-by-basis'' to obtain a feasible solution $w$ by greedily setting $w_{v,b'} := 0$ for each of the vectors $v \in b'$ except for one. Finally, we complete the description of $w$ by setting $w_{v,b} := w_v$ for each $b \in S_A$. 
\end{proof}

Therefore, the two formulations are equivalent for checking B-KS capability: one ILP is infeasible if and only if the other ILP is. The advantage of the second formulation is that the variables now depend only on the vectors, and so this formulation reduces the number of variables for the ILP (taking account also the number of slack variables appearing for each inequality constraint).


\section{Bipartite perfect quantum strategies with minimum input cardinality for some KS sets}
\label{Sec:Correlations}


In this section, we apply the algorithm on some KS sets. In many cases, we are also able to compute all B-KS capable sets. The chosen KS sets are those of minimum (or small) cardinality known and extensions of them. The reason for considering extension is that, as we shall see, they sometimes produce BPQSs of smaller input cardinality. Additionally, it is easy to see that the BPQSs that the extensions produce must have input cardinality at least as small as the KS sets they contain.

In the cases when we compute all B-KS capable sets of some KS set $\mathcal{V}$, we also present the number of B-KS capable sets for each size up to isomorphism (see Table~\ref{tab:BKS-capable-iso}). Here, we view B-KS capable sets $S_A$ and $S_A'$ to be \emph{isomorphic} if $S_A'$ can be obtained from $S_A$ by a permutation of the vectors $\mathcal{V}$ that preserves orthogonality.


\subsection{CEG-18} \label{sec:CEG18}


The 18-vector KS set in $d=4$ \cite{Cabello:1996PLA}, hereafter called CEG-$18$, is the KS with the smallest cardinality in any $d$ allowed by quantum theory \cite{Xu:2022PRL}. It has $9$ orthogonal bases and is shown in the upper part of Table~\ref{tab:P24+kids}. As shown in Table~\ref{tab:P24+kids}, CEG-18 is a subset of the KS set P-24 (see \ref{sec:P24}). A KS set $K$ is \emph{critical} if no subset of $K$ is a KS set. Along with K-20 (see \ref{sec:K20}), CEG-18 is one of 6 critical KS sets contained in P-24 \cite{PavicicMegillMerlet2010}.

Using our algorithm, we have found that the associated optimal B-KS sets have $|S_A| = 5$ and $|S_B| = 6$. An example of such an optimal B-KS set is shown in Table~\ref{tab:P24+kids}.


\begin{table*}[h]
    \centering
    \begin{tabular}{ccrrrrrrrrrrrrrrrrrrr}
        \hline \hline
       & Size & 4 & 5 & 6 & 7 & 8 & 9 & 10 & 11 & 12 & 13 & 14 & 15 & 16 & 17 & 18 & 19 & 20 & 21 & 22 \\
     \hline
     CEG-18 & \# B-KS capable & & 1 & 2 & 2 & 1 & 1 & \\
     \ & \# Essential & & 1 & 1 & & & & & & & & & \\
     CK-37 & \# B-KS capable & & & & 2 & 44 & 410 & 1843 & 4904 & 8552 & 10347 & 9079 & 5960 & 2983 & 1150 & 350 & 87 & 18 & 4 & 1 \\
     \ & \# Essential & & & & 2 & 30 & 43 & 27 & 8 \\
     CK-31 & \# B-KS capable & & & & & 4 & 41 & 168 & 345 & 403 & 286 & 127 & 36 & 7 & 1 \\
     \ & \# Essential & & & & & 4 & 14 & 12 & & & & & \\
     CK-33 & \# B-KS capable & & & & 2 & 23 & 156 & 560 & 1227 & 1825 & 1953 & 1563 & 939 & 420 & 138 & 33 & 6 & 1 \\
     \ & \# Essential & & & & 2 & 11 & 16 & 1 & 2 & 3 & 2 & 2 & \\
     KP-36 & \# B-KS capable & & & 1 & 9 & 11 & 7 & 3 & 1 & & & & \\
     \ & \# Essential & & & 1 & 7 & 8 & & & & & & & \\
     P-33 & \# B-KS capable & & & & 1 & 5 & 19 & 38 & 49 & 48 & 31 & 13 & 4 & 1 \\
     \ & \# Essential & & & & 1 & 3 & 6 & 2 & & & & & \\
     K-20 & \# B-KS capable & 1 & 3  & 11 & 21 & 18 & 10 & 3 & 1 & & & & \\
     \ & \# Essential & 1 & 0 & 4 & 1 & & & & & & & & \\
     ZP-28 & \# B-KS capable & & & 2 & 10 & 26 & 33 & 33 & 18 & 9 & 3 & 1 & \\
     \ & \# Essential & & & 2 & 4 & 2 & & & & & & & \\
     S-29 & \# B-KS capable & & 1 & 18 & 141 & 595 & 1399 & 1872 & 1534 & 802 & 277 & 64 & 9 & 1 \\
     \ & \# Essential & & 1 & 7 & 28 & 57 & 6 & & & & & & \\
     S-7 & \# B-KS capable & & 1 & 32 & 387 & 2482 & 9641 & 24384 &  42897 & 54809 & 52398 & 38088 & 21175 & 8974 & 2861 & 669 & 111 & 12 & 1 \\
     \ & \# Essential & & 1 & 16 & 58 & 124 & 30 &\\
     \hline \hline
    \end{tabular}
    \caption{Number of B-KS capable and essential B-KS capable sets up to isomorphism of each size for each of the KS sets in Sec.~\ref{Sec:Correlations} with at most $22$ orthogonal bases.}
    \label{tab:BKS-capable-iso}
\end{table*}


\begin{table*}[h]
    \centering
    \begin{tabular}{crrrrrrrrrrrrrrrrrrrrrrrr}
    \hline \hline
     & $v_{1}$ & $v_{2}$ & $v_{3}$ & $v_{4}$ & $v_{5}$ & $v_{6}$ & $v_{7}$ & $v_{8}$ & $v_{9}$ & $v_{10}$ & $v_{11}$ & $v_{12}$ & $v_{13}$ & $v_{14}$ & $v_{15}$ & $v_{16}$ & $v_{17}$ & $v_{18}$ & $v_{19}$ & $v_{20}$ & $v_{21}$ & $v_{22}$ & $v_{23}$ & $v_{24}$ \\ 
    \hline
    $v_{i1}$ & $1$ & $1$ & $0$ & $0$ & $1$ & $0$ & $1$ & $0$ & $1$ & $1$ & $0$ & $0$ & $1$ & $0$ & $0$ & $0$ & $1$ & $1$ & $1$ & $1$ & $1$ & $1$ & $1$ & $-1$ \\
    $v_{i2}$ & $1$ & $-1$ & $0$ & $0$ & $0$ & $1$ & $0$ & $1$ & $0$ & $0$ & $1$ & $1$ & $0$ & $1$ & $0$ & $0$ & $1$ & $-1$ & $1$ & $-1$ & $1$ & $1$ & $-1$ & $1$ \\
    $v_{i3}$ & $0$ & $0$ & $1$ & $1$ & $1$ & $0$ & $-1$ & $0$ & $0$ & $0$ & $1$ & $-1$ & $0$ & $0$ & $1$ & $0$ & $1$ & $1$ & $-1$ & $-1$ & $1$ & $-1$ & $1$ & $1$ \\
    $v_{i4}$ & $0$ & $0$ & $1$ & $-1$ & $0$ & $1$ & $0$ & $-1$ & $1$ & $-1$ & $0$ & $0$ & $0$ & $0$ & $0$ & $1$ & $1$ & $-1$ & $-1$ & $1$ & $-1$ & $1$ & $1$ & $1$ \\
    \hline
     & $\dot{1}$ & $\dot{1}$ & $\dot{1}$ & $\dot{1}$ & ${16}$ & ${16}$ & ${12}$ & ${12}$ & ${23}$ & ${20}$ & ${20}$ & ${23}$ & $\mathring{9}$ & $\mathring{9}$ & $\mathring{9}$ & $\mathring{9}$ & ${12}$ & ${12}$ & ${16}$ & ${16}$ & ${23}$ & ${20}$ & ${20}$ & ${23}$ \\ 
     & ${14}$ & ${2}$ & ${14}$ & ${2}$ & ${24}$ & ${19}$ & ${19}$ & ${24}$ & ${7}$ & ${7}$ & ${11}$ & ${11}$ & ${11}$ & ${7}$ & ${7}$ & ${11}$ & ${2}$ & ${14}$ & ${2}$ & ${14}$ & ${19}$ & ${24}$ & ${19}$ & ${24}$ \\ 
    $\mathcal{B}_{24}$ & ${8}$ & ${8}$ & ${3}$ & ${3}$ & $\dot{4}$ & $\dot{4}$ & $\dot{4}$ & $\dot{4}$ & $\dot{5}$ & $\dot{5}$ & $\dot{5}$ & $\dot{5}$ & ${3}$ & ${3}$ & ${8}$ & ${8}$ & $\mathring{15}$ & $\mathring{15}$ & $\mathring{15}$ & $\mathring{15}$ & $\mathring{22}$ & $\mathring{22}$ & $\mathring{22}$ & $\mathring{22}$ \\ 
     & ${21}$ & ${18}$ & ${18}$ & ${21}$ & ${10}$ & ${6}$ & ${10}$ & ${6}$ & ${13}$ & ${17}$ & ${13}$ & ${17}$ & ${6}$ & ${10}$ & ${6}$ & ${10}$ & ${17}$ & ${13}$ & ${13}$ & ${17}$ & ${18}$ & ${18}$ & ${21}$ & ${21}$ \\ 
    \hline 
    &  &  &  &  & $\dot{16}$ & $\dot{16}$ & & &  & $\mathring{20}$ & $\mathring{20}$ & ${23}$ & $\dot{9}$ & $\dot{9}$ & $\dot{9}$ & $\dot{9}$ & & & $\dot{16}$ & $\dot{16}$ & & $\mathring{20}$ & $\mathring{20}$ & \\ 
     & & $\mathring{2}$ & & $\mathring{2}$ & $\mathring{24}$ & & & $\mathring{24}$ & & & $\mathring{11}$ & $\mathring{11}$ & $\mathring{11}$ & & & $\mathring{11}$ & $\mathring{2}$ & & $\mathring{2}$ & & & $\mathring{24}$ & & $\mathring{24}$ \\ 
    $\mathcal{B}_{20}$ &  & & $\mathring{3}$ & $\mathring{3}$ & & & & & &  &  &  & $\mathring{3}$ & $\mathring{3}$ & & & & & & & $\dot{22}$ & $\dot{22}$ & $\dot{22}$ & $\dot{22}$ \\ 
     &  & $\mathring{18}$ & $\mathring{18}$ & & & $\mathring{6}$ & & $\mathring{6}$ & & $\dot{17}$ & & $\dot{17}$ & $\mathring{6}$ & & $\mathring{6}$ & & $\dot{17}$ & & & $\dot{17}$ & $\mathring{18}$ & $\mathring{18}$ & & \\ 
    \hline 
    & & &  &  & & & $\mathring{12}$ & $\mathring{12}$ & $\ddot{23}$ & & & $\ddot{23}$ &  & & & & $\mathring{12}$ & $\mathring{12}$ & & & $\ddot{23}$ & & & $\ddot{23}$ \\ 
     & & $\ddot{2}$ & & $\ddot{2}$ & $\mathring{24}$ & & & $\mathring{24}$ & & & $\mathring{11}$ & $\mathring{11}$ & $\mathring{11}$ & & & $\mathring{11}$ & $\ddot{2}$ & & $\ddot{2}$ & & & $\mathring{24}$ & & $\mathring{24}$ \\ 
    $\mathcal{B}_{18}$ & & & $\mathring{3}$ & $\mathring{3}$ & & & & && & & & $\mathring{3}$ & $\mathring{3}$ && & & & & & & & & \\ 
     && $\dot{18}$ & $\dot{18}$ & & $\dot{10}$ & & $\dot{10}$ & & $\dot{13}$ & & $\dot{13}$ & & & $\dot{10}$ & & $\dot{10}$ & & $\dot{13}$ & $\dot{13}$ & & $\dot{18}$ & $\dot{18}$ & & \\ 
    \hline \hline
    \end{tabular}
    \caption{P-24, K-20, CEG-18, and their optimal B-KSs. The upper part of the table displays the $24$ vectors $v_i$ of P-24. Each column corresponds to one vector. Each row gives one component. The following three sections of the table show optimal B-KS sets for P-24, K-20, and CEG-18, respectively. P-24 has $|\mathcal{B}|=24$ bases, the bases are numbered from $1$ to $24$. Each vector of P-24 belongs to four orthogonal bases of P-24: those indicated in the second part of table in the column corresponding to $v_i$. The notation is the following: $\dot{1}$ indicates that basis $1$ is only in $S_A$. $\mathring{9}$ indicates that basis $9$ is only in $S_B$. $2$ indicates that basis $2$ is in neither $S_A$ nor $S_B$.
    Therefore, an optimal B-KS set for P-24 is $S_A = \{\dot{1}, \dot{4}, \dot{5}\}$, $S_B = \{\mathring{9}, \mathring{15}, \mathring{22}\}$. K-20 and CEG-18 are both subsets of P-24, so they each contain a subset of the vectors and bases of P-24. The third and fourth parts of the table display optimal B-KS sets for K-20 and CEG-18, respectively.In addition to the notation previously explained, $\ddot{2}$ indicates that the basis $2$ is both in $S_A$ and $S_B$ (for CEG-18). Therefore, an optimal B-KS set for K-20 is $S_A = \{\dot{9},\dot{16},\dot{17},\dot{22}\}, S_B = \{\mathring{2},\mathring{3},\mathring{6},\mathring{11},\mathring{18},\mathring{20},\mathring{24}\}$ and an optimal B-KS set for CEG-18 is $S_A = \{\ddot{2},\dot{10},\dot{13},\dot{18},\ddot{23}\}, S_B = \{\ddot{2},\mathring{3},\mathring{11},\mathring{12},\ddot{23},\mathring{24}\}$.}
    \label{tab:P24+kids}
\end{table*}


In total (i.e., before considering isomorphisms), there are 96 B-KS capable sets of which only the following $15$ are essential (the bases are defined in Table~\ref{tab:P24+kids}):
\[
\begin{array}{lll}
    C_1 = \{ 1, 2, 4, 7, 8\}, & 
    C_2 = \{ 1, 2, 5, 6, 9\}, \\
    C_3 = \{ 1, 3, 4, 5, 7\}, &
    C_4 = \{ 1, 3, 6, 8, 9\}, \\ 
    C_5 = \{ 1, 4, 6, 7, 9\}, &
    C_6 = \{ 2, 3, 4, 6, 8\}, \\ 
    C_7 = \{ 2, 3, 5, 7, 9\}, & 
    C_8 = \{ 2, 4, 5, 8, 9\}, \\
    C_9 = \{ 3, 5, 6, 7, 8\}, &
    C_{10} = \{ 1, 2, 4, 5, 7, 9\}, \\
    C_{11} = \{ 1, 2, 4, 6, 8, 9\}, & 
    C_{12} = \{ 1, 3, 4, 6, 7, 8\}, \\ 
    C_{13} = \{ 1, 3, 5, 6, 7, 9\}, &
    C_{14} = \{ 2, 3, 4, 5, 7, 8\}, \\
    C_{15} = \{ 2, 3, 5, 6, 8, 9\}.
\end{array}
\]
Table~\ref{tab:BKS-capable-iso} shows the number of B-KS capable sets of each size up to isomorphism. From the 15 essential B-KS capable sets $C_1, \dots, C_{15}$, there are only 2 isomorphism classes: one of size~5 (containing $C_1$ through $C_9$) and one of size~6 (containing $C_{10}$ through $C_{15}$).


\subsection{P-24} \label{sec:P24}


Peres' 24-vector KS set in $d=4$ \cite{Peres:1991JPA}, hereafter called P-24, has $24$ orthogonal bases. The KS set, as well as one of its optimal B-KS sets (with $|S_A| = |S_B| = 3$) are shown in Table~\ref{tab:P24+kids}.

This B-KS set has $|S_A||S_B|$ minimum over all B-KS sets of any dimension $d$, and therefore corresponds to the BPQS of minimum input cardinality \cite{Cabello:2023XXX}.


\subsection{K-20} \label{sec:K20}


Kernaghan's 20 vector KS set in $d=4$ \cite{Kernaghan:1994JPA}, hereafter called K-20, has $11$ orthogonal bases. Similarly to CEG-18, K-20 is a critical KS set that is a subset of P-24. K-20 and one of its optimal B-KS sets (with $|S_A| = 4, |S_B| = 7$) are shown in Table~\ref{tab:P24+kids}.

In total, K-20 has $465$ B-KS capable sets of which $36$ are essential. 
The number of B-KS capable sets of each size are shown, up to isomorphism, in Table~\ref{tab:BKS-capable-iso}. We remark that the sequence of essential B-KS has a ``gap'': the number of essential B-KS sets of sizes 4 and 6 are both nonzero, but the number of essential B-KS sets of size 5 is zero. This is the only case we have encountered thus far that exhibits this behavior.


\subsection{Pen-40}


Penrose's 40-vector KS set in $d=4$ \cite{Penrose:2000}, hereafter called Pen-40, is a complex KS set. 20 of its vectors (referred to as \emph{explicit} vectors) are described by taking the Majorana map (see Ref. \cite{Zimba:1993SHPS} for more details) of the directions corresponding to vertices of a regular dodecahedron inscribed in the unit sphere. The other 20 vectors (referred to as \emph{implicit vectors}) are obtained by ``completing'' mutually orthogonal sets of size three to orthogonal bases of $\mathbb{C}^4$. In total Pen-40 has 40 orthogonal bases.

As in Ref. \cite{Cabello:1996}, we describe these vectors purely by their orthogonal relationships (see Table~\ref{tab:P40-orthogonalities}) - referring to the explicit vectors as $A, B, \dots, J, A^*, B^*, \dots, J^*$ and the implicit vectors as $a, b, \dots, j, a^*, b^*, \dots, j^*$. Pen-40 is illustrated in Table~\ref{tab:Pen40+soleheir}. 


\begin{table*}[h]
\centering
\begin{tabular}{ll}
\hline \hline
\text{Vector} & \text{Orthogonal vectors} \\
\hline
\texttt{A}   & \{\texttt{A*}, \texttt{C}, \texttt{D}, \texttt{G}, \texttt{H*}, \texttt{I*}, \texttt{J}, \texttt{a}, \texttt{a*}, \texttt{b}, \texttt{e}, \texttt{f}\} \\
\texttt{A*}  & \{\texttt{C*}, \texttt{D*}, \texttt{G*}, \texttt{H}, \texttt{I}, \texttt{J*}, \texttt{a}, \texttt{a*}, \texttt{b*}, \texttt{e*}, \texttt{f*}\} \\
\texttt{B}   & \{\texttt{B*}, \texttt{D}, \texttt{E}, \texttt{F}, \texttt{H}, \texttt{I*}, \texttt{J*}, \texttt{a}, \texttt{b}, \texttt{b*}, \texttt{c}, \texttt{g}\} \\
\texttt{B*}  & \{\texttt{D*}, \texttt{E*}, \texttt{F*}, \texttt{H*}, \texttt{I}, \texttt{J}, \texttt{a*}, \texttt{b}, \texttt{b*}, \texttt{c*}, \texttt{g*}\} \\
\texttt{C}   & \{\texttt{C*}, \texttt{E}, \texttt{F*}, \texttt{G}, \texttt{I}, \texttt{J*}, \texttt{b}, \texttt{c}, \texttt{c*}, \texttt{d}, \texttt{h}\} \\
\texttt{C*}  & \{\texttt{E*}, \texttt{F}, \texttt{G*}, \texttt{I*}, \texttt{J}, \texttt{b*}, \texttt{c}, \texttt{c*}, \texttt{d*}, \texttt{h*}\} \\
\texttt{D}   & \{\texttt{D*}, \texttt{F*}, \texttt{G*}, \texttt{H}, \texttt{J}, \texttt{c}, \texttt{d}, \texttt{d*}, \texttt{e}, \texttt{i}\} \\
\texttt{D*}  & \{\texttt{F}, \texttt{G}, \texttt{H*}, \texttt{J*}, \texttt{c*}, \texttt{d}, \texttt{d*}, \texttt{e*}, \texttt{i*}\} \\
\texttt{E}   & \{\texttt{E*}, \texttt{F}, \texttt{G*}, \texttt{H*}, \texttt{I}, \texttt{a}, \texttt{d}, \texttt{e}, \texttt{e*}, \texttt{j}\} \\
\texttt{E*}  & \{\texttt{F*}, \texttt{G}, \texttt{H}, \texttt{I*}, \texttt{a*}, \texttt{d*}, \texttt{e}, \texttt{e*}, \texttt{j*}\} \\
\texttt{F}   & \{\texttt{F*}, \texttt{G}, \texttt{J}, \texttt{a}, \texttt{f}, \texttt{f*}, \texttt{h*}, \texttt{i*}\} \\
\texttt{F*}  & \{\texttt{G*}, \texttt{J*}, \texttt{a*}, \texttt{f}, \texttt{f*}, \texttt{h}, \texttt{i}\} \\
\texttt{G}   & \{\texttt{G*}, \texttt{H}, \texttt{b}, \texttt{g}, \texttt{g*}, \texttt{i*}, \texttt{j*}\} \\
\texttt{G*}  & \{\texttt{H*}, \texttt{b*}, \texttt{g}, \texttt{g*}, \texttt{i}, \texttt{j}\} \\
\texttt{H}   & \{\texttt{H*}, \texttt{I}, \texttt{c}, \texttt{f*}, \texttt{h}, \texttt{h*}, \texttt{j*}\} \\
\texttt{H*}  & \{\texttt{I*}, \texttt{c*}, \texttt{f}, \texttt{h}, \texttt{h*}, \texttt{j}\} \\
\texttt{I}   & \{\texttt{I*}, \texttt{J}, \texttt{d}, \texttt{f*}, \texttt{g*}, \texttt{i}, \texttt{i*}\} \\
\texttt{I*}  & \{\texttt{J*}, \texttt{d*}, \texttt{f}, \texttt{g}, \texttt{i}, \texttt{i*}\} \\
\texttt{J}   & \{\texttt{J*}, \texttt{e}, \texttt{g*}, \texttt{h*}, \texttt{j}, \texttt{j*}\} \\
\texttt{J*}  & \{\texttt{e*}, \texttt{g}, \texttt{h}, \texttt{j}, \texttt{j*}\} \\
\texttt{a}   & \{\texttt{a*}, \texttt{c*}, \texttt{d*}, \texttt{g*}, \texttt{h}, \texttt{i}, \texttt{j*}\} \\
\texttt{a*}  & \{\texttt{c}, \texttt{d}, \texttt{g}, \texttt{h*}, \texttt{i*}, \texttt{j}\} \\
\texttt{b}   & \{\texttt{b*}, \texttt{d*}, \texttt{e*}, \texttt{f*}, \texttt{h*}, \texttt{i}, \texttt{j}\} \\
\texttt{b*}  & \{\texttt{d}, \texttt{e}, \texttt{f}, \texttt{h}, \texttt{i*}, \texttt{j*}\} \\
\texttt{c}   & \{\texttt{c*}, \texttt{e*}, \texttt{f}, \texttt{g*}, \texttt{i*}, \texttt{j}\} \\
\texttt{c*}  & \{\texttt{e}, \texttt{f*}, \texttt{g}, \texttt{i}, \texttt{j*}\} \\
\texttt{d}   & \{\texttt{d*}, \texttt{f}, \texttt{g}, \texttt{h*}, \texttt{j*}\} \\
\texttt{d*}  & \{\texttt{f*}, \texttt{g*}, \texttt{h}, \texttt{j}\} \\
\texttt{e}   & \{\texttt{e*}, \texttt{f*}, \texttt{g}, \texttt{h}, \texttt{i*}\} \\
\texttt{e*}  & \{\texttt{f}, \texttt{g*}, \texttt{h*}, \texttt{i}\} \\
\texttt{f}   & \{\texttt{f*}, \texttt{g*}, \texttt{j*}\} \\
\texttt{f*}  & \{\texttt{g}, \texttt{j}\} \\
\texttt{g}   & \{\texttt{g*}, \texttt{h*}\} \\
\texttt{g*}  & \{\texttt{h}\} \\
\texttt{h}   & \{\texttt{h*},\texttt{i*}\} \\
\texttt{h*}   & \{\texttt{i}\} \\
\texttt{i}   & \{\texttt{i*,j*}\} \\
\texttt{i*}   & \{\texttt{j}\} \\
\texttt{j}   & \{\texttt{j*}\} \\
\texttt{j*}   &  \{\ \}\\
\hline \hline
\end{tabular}
\caption{Orthogonalities of P-40. For each vector, we only list the orthogonal vectors appearing below it in the table in order to avoid redundancy.}
\label{tab:P40-orthogonalities}
\end{table*}


\begin{table*}[h]
    \centering
    {\scalebox{0.8}{\begin{tabular}{crrrrrrrrrrrrrrrrrrrrrrrrrrrrrrrrrrrrrrrr}
    \hline \hline
    \ & $A$ & $B$ & $C$ & $D$ & $E$ & $F$ & $G$ & $H$ & $I$ & $J$ & $A*$ & $B*$ & $C*$ & $D*$ & $E*$ & $F*$ & $G*$ & $H*$ & $I*$ & $J*$ & $a$ & $b$ & $c$ & $d$ & $e$ & $f$ & $g$ & $h$ & $i$ & $j$ & $a*$ & $b*$ & $c*$ & $d*$ & $e*$ & $f*$ & $g*$ & $h*$ & $i*$ & $j*$ \\
    \hline
     & ${25}$ & $\ddot{1}$ & ${25}$ & $\ddot{1}$ &  & ${34}$ & ${25}$ & $\ddot{1}$ &  & ${37}$ &  & $\mathring{17}$ & ${21}$ & $\mathring{17}$ & ${21}$ & ${34}$ &  & $\mathring{17}$ & ${21}$ & ${37}$ &  & ${25}$ & $\ddot{1}$ & ${40}$ & $\mathring{33}$ & ${34}$ & ${40}$ & $\mathring{33}$ &  & ${37}$ & ${40}$ & $\mathring{33}$ & $\mathring{17}$ & ${21}$ &  & ${34}$ &  & ${40}$ & $\mathring{33}$ & ${37}$ \\ 
     & ${30}$ &  & $\ddot{2}$ & ${30}$ & $\ddot{2}$ &  & $\mathring{14}$ & ${16}$ & $\ddot{2}$ & ${30}$ & ${8}$ & ${38}$ &  & ${8}$ & ${38}$ & ${38}$ & $\mathring{14}$ & ${16}$ &  & ${8}$ & ${24}$ & $\mathring{28}$ &  & $\ddot{2}$ & ${30}$ &  & $\mathring{14}$ & ${16}$ & ${24}$ & $\mathring{28}$ & ${38}$ &  & ${24}$ & $\mathring{28}$ & ${8}$ & $\mathring{28}$ & $\mathring{14}$ & ${16}$ &  & ${24}$ \\ 
    $\mathcal{B}_{40}$ & $\dot{9}$ & ${26}$ &  &  & ${15}$ & $\ddot{3}$ &  &  & ${20}$ & $\ddot{3}$ & $\dot{9}$ & ${26}$ & $\ddot{3}$ &  &  &  & ${15}$ & ${15}$ & ${20}$ &  & $\dot{9}$ & ${26}$ & ${35}$ &  &  & ${35}$ &  &  & ${20}$ & ${15}$ & $\dot{9}$ & ${26}$ &  &  & ${35}$ &  & ${35}$ & $\ddot{3}$ & ${20}$ &  \\ 
     &  & ${19}$ & ${29}$ & ${13}$ & ${31}$ & ${4}$ & ${4}$ & ${6}$ & ${6}$ &  & ${6}$ &  & ${29}$ & ${4}$ & ${31}$ & ${13}$ & ${13}$ &  & ${19}$ & ${19}$ & ${23}$ &  & ${29}$ & ${36}$ & ${31}$ & ${36}$ & ${19}$ & ${23}$ & ${13}$ &  &  & ${36}$ & ${29}$ & ${23}$ & ${31}$ & ${6}$ & ${23}$ &  & ${4}$ & ${36}$ \\ 
     & ${18}$ & $\mathring{22}$ & ${12}$ & ${11}$ & $\mathring{22}$ & $\mathring{22}$ & ${5}$ & ${5}$ & ${10}$ & ${10}$ & ${7}$ & ${10}$ & ${7}$ & ${11}$ & ${5}$ & ${12}$ & ${7}$ & ${18}$ & ${18}$ & ${12}$ & $\mathring{22}$ & ${27}$ & ${39}$ & ${11}$ & ${32}$ & ${18}$ & ${32}$ & ${12}$ & ${27}$ & ${39}$ & ${39}$ & ${7}$ & ${32}$ & ${11}$ & ${27}$ & ${32}$ & ${10}$ & ${27}$ & ${39}$ & ${5}$ \\
     \hline
     &  & $\dot{1}$ &  & $\dot{1}$ &  &  &  & $\dot{1}$ &  & $\ddot{37}$ &  & ${17}$ &  & ${17}$ &  &  &  & ${17}$ &  & $\ddot{37}$ &  &  & $\dot{1}$ &  & & & & &  & $\ddot{37}$ &  & & ${17}$ & &  & &  & &  & $\ddot{37}$ \\ 
     &  &  &  &  &  &  &  &  &  &  &  & $\mathring{38}$ &  &  & $\mathring{38}$ & $\mathring{38}$ &  &  &  &  & $\mathring{24}$ & &  & & &  & & & $\mathring{24}$ &  & $\mathring{38}$ &  & $\mathring{24}$ & & & & & &  & $\mathring{24}$ \\ 
    $\mathcal{B}_{28}$ & $\ddot{9}$ &  &  &  & $\mathring{15}$ &  &  &  & $\ddot{20}$ &  & $\ddot{9}$ &  &  &  &  &  & $\mathring{15}$ & $\mathring{15}$ & $\ddot{20}$ &  & $\ddot{9}$ &  &  &  &  & &  &  & $\ddot{20}$ & $\mathring{15}$ & $\ddot{9}$ & &  &  & &  & & & $\ddot{20}$ &  \\ 
     &  &  & $\dot{29}$ & ${13}$ &  & $\dot{4}$ & $\dot{4}$ &  &  &  &  &  & $\dot{29}$ & $\dot{4}$ &  & ${13}$ & ${13}$ &  &  &  &  &  & $\dot{29}$ & & & & & & ${13}$ &  &  & & $\dot{29}$ & & & & &  & $\dot{4}$ &  \\ 
     &  & $\mathring{22}$ &  &  & $\mathring{22}$ & $\mathring{22}$ & $\mathring{5}$ & $\mathring{5}$ &  &  &  &  &  &  & $\mathring{5}$ &  &  &  &  &  & $\mathring{22}$ & & ${39}$ & & & & & &  & ${39}$ & ${39}$ & &  & & & & & & ${39}$ & $\mathring{5}$ \\ 
    \hline \hline
    \end{tabular}
    }}
\caption{Pen-40, ZP-28 and optimal B-KS of ZP-28. The top row indicates the names of each of the (implicitly defined) vectors of Pen-40. See \ref{tab:P40-orthogonalities} for their orthogonalities. The following two sections illustrate optimal B-KS sets for Pen-40 and ZP-28, respectively. Pen-40 has $40$, labeled from 1 to 40. Each vector is in exactly 4 orthogonal bases. The notation is the following: $\dot{9}$ indicates that basis $1$ is in $S_A$, but not $S_B$. $\mathring{14}$ indicates that basis $14$ is in $S_B$, but not $S_A$. $25$ indicates that basis $4$ is in neither $S_A$ nor $S_B$, and $\ddot{1}$ indicates that $1$ is in both $S_A$ and $S_B$.
ZP-28 is a subset of Pen-40, and so its vectors and orthogonal bases are subsets of the vectors and orthogonal bases of Pen-40. Therefore, the optimal B-KS set described is $S_A = \{\ddot{1}, \ddot{2}, \ddot{3}, \dot{9}\}, S_B = \{\ddot{1}, \ddot{2}, \ddot{3}, 33, 14, 17, 22, 28\}$.
The third (bottom) part of the table indicates the orthogonal bases of ZP-28 and an optimal B-KS $S_A = \{\dot{1}, \dot{4}, \ddot{9}, \ddot{20}, \dot{29}, \ddot{37}\}, S_B = \{\mathring{5}, \ddot{9}, \mathring{15}, \ddot{20}, \mathring{22}, \mathring{24}, \ddot{37}, \mathring{38}\}.$
$\dot{1}$ indicates that basis 1 is only in $S_A$. $\mathring{15}$ indicates that basis 15 is only $S_B$. $\ddot{9}$ indicates that basis 9 is in both $S_A$ and $S_B$. $17$ indicates that basis 17 is in neither $S_A$ nor $S_B$.}
    \label{tab:Pen40+soleheir}
\end{table*}

Optimal B-KS sets $(S_A,S_B)$ produced by Pen-40 have $|S_A| = 4, |S_B| = 8$. In this case (as for MP-57), Algorithm 1 was insufficient to compute the optimum. Instead, we first identified that four was the minimum value of $|S_A|$, then computed that there are four such B-KS capable sets up to isomorphism ($\{1,2,3,8\}, \{1,2,3,9\}, \{1,2,3,24\}, \{1,2,14,33\}$). We then checked against each of these four sets and found that there no B-KS sets with $|S_B| = 7$, but there are B-KS sets with $|S_B|= 8$. We applied a similar technique to show that there are no solutions with $|S_A| = |S_B| = 5.$


\subsection{ZP-28} \label{sec:ZP28}


Zimba-Penrose's 28-vector KS set in $d=4$ \cite{Zimba:1993SHPS} is a critical KS set formed by removing 12 implicit vectors of Pen-40 such that the remaining 8 implicit vectors correspond to one of the five cubes (see Ref. \cite{Cabello:1996} for details) inscribed in the dodecahedron. One possible choice of implicit vectors to keep is $\{a, c, i, j, a^*, c^*, i^*, j^*\}$. 
There are 14 orthogonal bases remaining from Pen-40, and optimal B-KS sets for ZP-28 have $|S_A| = 6, |S_B| = 8$. ZP-28 along with one of its optimal B-KS set are illustrated in Table~\ref{tab:Pen40+soleheir}.

We have been able to enumerate all $2020$ B-KS capable sets, of which $96$ are essential. Table~\ref{tab:BKS-capable-iso} illustrates the number of B-KS capable sets of ZP-28 by size up to isomorphism. 


\subsection{CK-31}


Conway and Kochen's 31-vector KS set in $d=3$ \cite{Peres:1993} is a critical KS set and is currently the smallest (in terms of number of vectors) known KS set in dimension $3$. It has $17$ orthogonal bases. The optimal B-KS sets produced in this case have $|S_A| = 8, |S_B| = 9$. CK-31 along with an optimal B-KS set are illustrated in Table~\ref{tab:CK37+kids}.

We were able to compute each of the $5252$ B-KS capable sets of which $88$ are essential. The number of such sets of each size is given, up to isomorphism, in
Table~\ref{tab:BKS-capable-iso}.


\subsection{CK-37} \label{sec:CK-37}


Conway and Kochen's 37-vector set in $d=3$ \cite{Peres:1993} has $22$ orthogonal bases. It is an extension of CK-31 (indeed, there are $6$ copies of CK-31 in CK-37). The KS set, as well as an optimal B-KS set (with $|S_A| = 8, |S_B| = 9$ just as in the CK-31 case) are shown in Table~\ref{tab:CK37+kids}.


\begin{table*}[h]
    \centering
    {\scalebox{0.85}{
    \begin{tabular}{crrrrrrrrrrrrrrrrrrrrrrrrrrrrrrrrrrrrr}
    \hline \hline
     & $v_{1}$ & $v_{2}$ & $v_{3}$ & $v_{4}$ & $v_{5}$ & $v_{6}$ & $v_{7}$ & $v_{8}$ & $v_{9}$ & $v_{10}$ & $v_{11}$ & $v_{12}$ & $v_{13}$ & $v_{14}$ & $v_{15}$ & $v_{16}$ & $v_{17}$ & $v_{18}$ & $v_{19}$ & $v_{20}$ & $v_{21}$ & $v_{22}$ & $v_{23}$ & $v_{24}$ & $v_{25}$ & $v_{26}$ & $v_{27}$ & $v_{28}$ & $v_{29}$ & $v_{30}$ & $v_{31}$ & $v_{32}$ & $v_{33}$ & $v_{34}$ & $v_{35}$ & $v_{36}$ & $v_{37}$ \\ 
    \hline
    $v_{i1}$ & $0$ & $0$ & $1$ & $1$ & $1$ & $1$ & $0$ & $0$ & $0$ & $0$ & $1$ & $2$ & $2$ & $1$ & $1$ & $2$ & $1$ & $1$ & $-2$ & $1$ & $1$ & $2$ & $-1$ & $1$ & $1$ & $1$ & $-1$ & $1$ & $1$ & $0$ & $0$ & $2$ & $2$ & $2$ & $-1$ & $1$ & $1$ \\
    $v_{i2}$ & $1$ & $1$ & $0$ & $0$ & $-1$ & $1$ & $2$ & $1$ & $2$ & $1$ & $0$ & $0$ & $0$ & $0$ & $2$ & $-1$ & $1$ & $1$ & $1$ & $1$ & $1$ & $-1$ & $1$ & $-1$ & $2$ & $-1$ & $1$ & $2$ & $0$ & $1$ & $0$ & $1$ & $1$ & $1$ & $2$ & $-2$ & $-2$ \\
    $v_{i3}$ & $-1$ & $1$ & $-1$ & $1$ & $0$ & $0$ & $1$ & $-2$ & $-1$ & $2$ & $2$ & $-1$ & $1$ & $-2$ & $0$ & $0$ & $-2$ & $1$ & $1$ & $2$ & $-1$ & $1$ & $2$ & $1$ & $1$ & $2$ & $1$ & $-1$ & $0$ & $0$ & $1$ & $1$ & $0$ & $-1$ & $1$ & $0$ & $1$ \\
    \hline 
    & $\dot{16}$ & $\dot{16}$ & $\dot{1}$ & $\dot{1}$ & $\mathring{8}$ & $\mathring{14}$ &  &  &  &  &  &  &  &  & $\dot{3}$ & $\dot{3}$ &  &  &  & $\mathring{8}$ & $\mathring{8}$ &  &  &  &  & $\mathring{14}$ & $\mathring{14}$ &  & $\dot{16}$ & $\dot{1}$ & $\dot{3}$ &  &  &  &  &  &  \\ 
     & $\mathring{6}$ & $\mathring{9}$ & $\mathring{10}$ & $\mathring{13}$ & $\dot{2}$ & $\dot{2}$ & $\dot{17}$ & $\dot{17}$ &  &  & $\dot{20}$ & $\dot{20}$ &  &  &  &  &  & $\mathring{6}$ & $\mathring{6}$ &  & $\mathring{9}$ & $\mathring{9}$ &  & $\mathring{10}$ & $\mathring{10}$ &  & $\mathring{13}$ & $\mathring{13}$ & $\dot{17}$ & $\dot{20}$ & $\dot{2}$ &  &  &  &  &  &  \\ 
    $\mathcal{B}_{37}$ &  & ${12}$ & ${4}$ &  &  &  &  &  & $\dot{18}$ & $\dot{18}$ &  &  & $\dot{21}$ & $\dot{21}$ &  &  &  & ${4}$ &  &  &  &  &  & ${12}$ &  &  &  &  & $\dot{18}$ & $\dot{21}$ & ${22}$ &  & ${22}$ & ${12}$ &  & ${22}$ & ${4}$ \\ 
     & ${15}$ &  &  & ${7}$ & $\mathring{5}$ & $\mathring{11}$ &  &  &  &  &  &  &  &  &  &  & $\mathring{5}$ & $\mathring{5}$ &  &  & ${7}$ &  & $\mathring{11}$ & $\mathring{11}$ &  &  & ${15}$ &  & $\mathring{19}$ & $\mathring{19}$ & $\mathring{19}$ & ${15}$ &  &  & ${7}$ &  & \\
     \hline  
     & $\dot{16}$ & $\dot{16}$ & $\dot{1}$ & $\dot{1}$ & $\mathring{8}$ & $\mathring{14}$ &  &  &  &  &  &  &  &  & $\dot{3}$ & $\dot{3}$ &  &  &  & $\mathring{8}$ & $\mathring{8}$ &  &  &  &  & $\mathring{14}$ & $\mathring{14}$ &  & $\dot{16}$ & $\dot{1}$ & $\dot{3}$ &  &  &  &  &  &  \\ 
     & $\mathring{6}$ & $\mathring{9}$ & $\mathring{10}$ & $\mathring{13}$ & $\dot{2}$ & $\dot{2}$ & &  &  &  & $\dot{20}$ & $\dot{20}$ &  &  &  &  &  & $\mathring{6}$ & $\mathring{6}$ &  & $\mathring{9}$ & $\mathring{9}$ &  & $\mathring{10}$ & $\mathring{10}$ &  & $\mathring{13}$ & $\mathring{13}$ &  & $\dot{20}$ & $\dot{2}$ &  &  &  &  &  &  \\ 
    $\mathcal{B}_{33}$ &  & $\mathring{12}$ & $\mathring{4}$ &  &  &  &  &  &  &  &  &  & $\dot{21}$ & $\dot{21}$ &  &  &  & $\mathring{4}$ &  &  &  &  &  & $\mathring{12}$ &  &  &  &  &  & $\dot{21}$ & $\dot{22}$ &  & $\dot{22}$ & $\mathring{12}$ &  & $\dot{22}$ & $\mathring{4}$ \\ 
     & $\mathring{15}$ &  &  & $\mathring{7}$ & $\mathring{5}$ & $\mathring{11}$ &  &  &  &  &  &  &  &  &  &  & $\mathring{5}$ & $\mathring{5}$ &  &  & $\mathring{7}$ &  & $\mathring{11}$ & $\mathring{11}$ &  &  & $\mathring{15}$ &  & $\mathring{19}$ & $\mathring{19}$ & $\mathring{19}$ & $\mathring{15}$ &  &  & $\mathring{7}$ &  &  \\
     \hline 
     & $\dot{16}$ & $\dot{16}$ & $\dot{1}$ & $\dot{1}$ & $\mathring{8}$ & $\mathring{14}$ &  &  &  &  &  &  &  &  & $\dot{3}$ & $\dot{3}$ &  &  &  & $\mathring{8}$ & $\mathring{8}$ &  &  &  &  & $\mathring{14}$ & $\mathring{14}$ &  & $\dot{16}$ & $\dot{1}$ & $\dot{3}$ &  &  &  &  &  &  \\ 
     & $\mathring{6}$ & $\mathring{9}$ & $\mathring{10}$ & $\mathring{13}$ & $\dot{2}$ & $\dot{2}$ & $\dot{17}$ & $\dot{17}$ &  &  & $\dot{20}$ & $\dot{20}$ &  &  &  &  &  & $\mathring{6}$ & $\mathring{6}$ &  & $\mathring{9}$ & $\mathring{9}$ &  & $\mathring{10}$ & $\mathring{10}$ &  & $\mathring{13}$ & $\mathring{13}$ & $\dot{17}$ & $\dot{20}$ & $\dot{2}$ &  &  &  &  &  &  \\ 
    $\mathcal{B}_{31}$ &  &  &  &  &  &  &  &  & $\dot{18}$ & $\dot{18}$ &  &  & $\dot{21}$ & $\dot{21}$ &  &  &  &  &  &  &  &  &  &  &  &  &  &  & $\dot{18}$ & $\dot{21}$ &  &  & & &  &  &  \\ 
     &  &  &  &  & $\mathring{5}$ & $\mathring{11}$ &  &  &  &  &  &  &  &  &  &  & $\mathring{5}$ & $\mathring{5}$ &  &  &  &  & $\mathring{11}$ & $\mathring{11}$ &  &  &  &  & $\mathring{19}$ & $\mathring{19}$ & $\mathring{19}$ &  &  &  &  &  &  \\
    \hline \hline
    \end{tabular}
    }}
\caption{CK-37, CK-33, CK-31, and their optimal B-KSs. The upper part of the table displays the $37$ vectors of CK-37, with each column corresponding to a vector and each row to a component. The second part of the table presents an optimal B-KS. CK-37 has $|\mathcal{B}| = 22$ orthogonal bases numbered from $1$ to $37$. The bases containing each vector are indicated in the second part of table in the column corresponding to that vector. $\dot{16}$ indicates that basis $16$ is only in $S_A$. $\mathring{8}$ indicates that basis $8$ is only in $S_B$. $22$ indicates that basis $22$ is in neither $S_A$ nor $S_B$. 
The optimal B-KS set illustrated is thus $S_A = \{\dot{1}, \dot{2}, \dot{3}, \dot{16}, \dot{17}, \dot{18}, \dot{20}, \dot{21}\}, S_B = \{\mathring{5}, \mathring{6}, \mathring{8}, \mathring{9}, \mathring{10}, \mathring{11}, \mathring{13}, \mathring{14}, \mathring{19}\}.$ The KS set CK-33 is obtained from CK-37 by removing vectors $v_7, \dots, v_{10}$ (see the third part of the table), and the KS set CK-31 is obtained from CK-37 by removing vectors $v_{32}, \dots, v_{37}$ (see the bottom part of the table). In the case of CK-31, we obtain the same optimal B-KS set as for CK-37. However, in the case of CK-33, we obtain the optimal B-KS set $S_A = \{\dot{1}, \dot{2}, \dot{3}, \dot{16}, \dot{20}, \dot{21}, \dot{22}\}, S_B = \{\mathring{4}, \mathring{5}, \mathring{6}, \mathring{7}, \mathring{8}, \mathring{9}, \mathring{10}, \mathring{11}, \mathring{12}, \mathring{13}, \mathring{14}, \mathring{15}, \mathring{19}\}$. }
    \label{tab:CK37+kids}
\end{table*}


We were able to compute the set of all $1060326$ B-KS capable sets of which $2127$ are essential.
The number of such sets of each size is given up to isomorphism in Table~\ref{tab:BKS-capable-iso}. We note that the optimal B-KS sets do not use an essential B-KS capable set $S_A$ of minimum size ($7$) - indeed, the smallest B-KS set using such an $S_A$ has $|S_B| = 13$.


\subsection{CK-33}


Conway and Kochen's 33-vector KS set in $d=3$ \cite{Peres:1993}, hereafter called CK-33 (or Sch\"{u}tte-33 \cite{Bub:1996FP}), is also a  critical KS set that it is a subset of CK-37 (there are 3 copies of CK-33 within CK-37). 
It has $20$ orthogonal bases, and its optimal B-KS sets have $|S_A| = 7, |S_B| = 13$, larger than that of CK-37 and CK-31. Table \ref{tab:CK37+kids} illustrates CK-33, as well as one of its optimal B-KS sets. 

We were able to compute all $65298$ B-KS capable sets of which $185$ are essential. We illustrate the number of sets of each size up to isomorphism in Table~\ref{tab:BKS-capable-iso}.


\subsection{P-33}


Peres' 33-vector set in $d=3$ \cite{Peres:1991JPA}, hereafter called P-33 is a critical KS set. It has $16$ orthogonal bases. The KS set, as well as one of its optimal B-KS sets (with $|S_A| = 7, |S_B| = 9$) are illustrated in Table~\ref{tab:MP57+soleheir}. P-33 is a member of a family of KS sets sharing the same orthogonality graph \cite{Penrose:2000,gould2010isomorphism,bengtsson2012gleason}.

It has been conjectured \cite[Conjecture 2]{Cabello:2023XXX} that this yields the qutrit-qutrit BPQS with the smallest $|S_A| |S_B|$. 

We have computed each of the 2008 B-KS capable sets, of which 76 are essential. 
Table~\ref{tab:BKS-capable-iso} lists the number of B-KS capable sets of each size up to isomorphism.


\subsection{MP-57}


In Ref. \cite{Mancinska:2007}, Man\v{c}inska constructed the ``completion'' of P-33, adding vectors so that each pair of orthogonal vectors is in some orthogonal basis.
The resulting 57-vector KS set (henceforth MP-57) has 40 orthogonal bases. 

The optimal B-KS sets of MP-57 are the same size as for P-33 ($|S_A| = 7, |S_B| = 9$). We illustrate the KS set and one of its optimal B-KS sets in Table~\ref{tab:MP57+soleheir}. 

To compute the optimal B-KS sets of MP-57, we took a similar approach to the Pen-40 case, applying a modification of Algorithm~\ref{alg:bks} which takes advantage of symmetry (isomorphisms). More explicitly, when computing B-KS capable sets, we only store one representative from each isomorphism class. In this way, we are able to significantly reduce the number of times that we need to call the function ``is B-KS.''


\begin{table*}[h]
    \centering
    {\scalebox{0.46}{
    \begin{tabular}{cccccccccccccccccccccccccccccccccccccccccccccccccccccccccc}
    \hline \hline
     & $v_{1}$ & $v_{2}$ & $v_{3}$ & $v_{4}$ & $v_{5}$ & $v_{6}$ & $v_{7}$ & $v_{8}$ & $v_{9}$ & $v_{10}$ & $v_{11}$ & $v_{12}$ & $v_{13}$ & $v_{14}$ & $v_{15}$ & $v_{16}$ & $v_{17}$ & $v_{18}$ & $v_{19}$ & $v_{20}$ & $v_{21}$ & $v_{22}$ & $v_{23}$ & $v_{24}$ & $v_{25}$ & $v_{26}$ & $v_{27}$ & $v_{28}$ & $v_{29}$ & $v_{30}$ & $v_{31}$ & $v_{32}$ & $v_{33}$ & $v_{34}$ & $v_{35}$ & $v_{36}$ & $v_{37}$ & $v_{38}$ & $v_{39}$ & $v_{40}$ & $v_{41}$ & $v_{42}$ & $v_{43}$ & $v_{44}$ & $v_{45}$ & $v_{46}$ & $v_{47}$ & $v_{48}$ & $v_{49}$ & $v_{50}$ & $v_{51}$ & $v_{52}$ & $v_{53}$ & $v_{54}$ & $v_{55}$ & $v_{56}$ & $v_{57}$ \\ 
    \hline
    $v_{i1}$ & $1$ & $0$ & $0$ & $0$ & $1$ & $1$ & $0$ & $1$ & $1$ & $0$ & $0$ & $0$ & $0$ & $1$ & $t$ & $1$ & $t$ & $1$ & $t$ & $1$ & $t$ & $1$ & $-1$ & $1$ & $1$ & $1$ & $-1$ & $1$ & $1$ & $t$ & $t$ & $t$ & $-t$ & $1$ & $1$ & $1$ & $1$ & $1$ & $1$ & $1$ & $1$ & $1$ & $1$ & $1$ & $1$ & $1$ & $1$ & $1$ & $1$ & $1$ & $1$ & $1$ & $1$ & $1$ & $1$ & $1$ & $1$ \\
    $v_{i2}$ & $0$ & $1$ & $1$ & $1$ & $0$ & $0$ & $0$ & $1$ & $-1$ & $1$ & $t$ & $1$ & $t$ & $0$ & $0$ & $0$ & $0$ & $t$ & $-1$ & $-t$ & $1$ & $-1$ & $1$ & $1$ & $1$ & $t$ & $t$ & $t$ & $-t$ & $1$ & $-1$ & $1$ & $1$ & $-t/3$ & $t/3$ & $1/3$ & $-1/3$ & $t/3$ & $-t/3$ & $1/3$ & $-1/3$ & $-3t/2$ & $t/2$ & $3t/2$ & $-t/2$ & $-3t/2$ & $t/2$ & $3t/2$ & $-t/2$ & $3$ & $-3$ & $3$ & $-3$ & $t$ & $-t$ & $t$ & $-t$ \\
    $v_{i3}$ & $0$ & $1$ & $-1$ & $0$ & $1$ & $-1$ & $1$ & $0$ & $0$ & $t$ & $-1$ & $-t$ & $1$ & $t$ & $-1$ & $-t$ & $1$ & $0$ & $0$ & $0$ & $0$ & $t$ & $t$ & $t$ & $-t$ & $-1$ & $1$ & $1$ & $1$ & $-1$ & $1$ & $1$ & $1$ & $1/3$ & $-1/3$ & $t/3$ & $-t/3$ & $1/3$ & $-1/3$ & $-t/3$ & $t/3$ & $-t/2$ & $3t/2$ & $t/2$ & $-3t/2$ & $t/2$ & $-3t/2$ & $-t/2$ & $3t/2$ & $t$ & $t$ & $-t$ & $-t$ & $3$ & $3$ & $-3$ & $-3$ \\
    \hline
     & $\mathring{1}$ & $\mathring{1}$ & $\mathring{1}$ & $\mathring{21}$ &  & $\dot{5}$ & $\mathring{24}$ &  &  & ${29}$ & ${11}$ & ${14}$ & ${17}$ & ${31}$ & ${19}$ & $\mathring{21}$ & $\mathring{21}$ & $\mathring{24}$ & $\mathring{24}$ & ${37}$ &  & ${19}$ & ${11}$ &  & ${17}$ & ${29}$ & ${14}$ & $\dot{5}$ & $\dot{5}$ & ${31}$ &  & ${37}$ &  & ${29}$ &  & ${11}$ &  & ${14}$ &  &  & ${17}$ & ${31}$ &  &  &  &  & ${37}$ &  &  & ${19}$ &  &  &  &  &  &  &  \\ 
     & $\dot{2}$ &  &  & $\dot{2}$ &  &  & $\dot{2}$ &  &  &  & ${12}$ &  &  & ${39}$ &  &  & ${23}$ & ${35}$ &  &  & ${28}$ &  &  & ${12}$ & ${23}$ &  &  &  & ${28}$ &  & ${35}$ &  & ${39}$ &  &  &  & ${12}$ &  &  &  &  &  &  &  & ${35}$ &  &  & ${39}$ &  &  &  &  & ${23}$ &  &  &  & ${28}$ \\ 
    $\mathcal{B}_{57}$ & $\mathring{10}$ &  & $\dot{38}$ & $\mathring{3}$ & $\mathring{3}$ & $\mathring{3}$ & $\mathring{6}$ & $\mathring{6}$ & $\mathring{6}$ & $\mathring{10}$ & $\mathring{10}$ & ${15}$ & ${16}$ &  & ${20}$ & ${34}$ & ${22}$ &  & ${30}$ & ${32}$ & ${27}$ & ${16}$ & ${22}$ & ${20}$ &  & ${30}$ & ${27}$ & ${15}$ &  & ${32}$ & ${34}$ & $\dot{38}$ & $\dot{38}$ &  &  &  &  &  & ${15}$ & ${16}$ &  &  & ${32}$ & ${34}$ &  &  &  &  &  &  & ${20}$ & ${22}$ &  & ${30}$ & ${27}$ &  &  \\ 
     & $\mathring{13}$ & $\dot{33}$ &  & $\mathring{18}$ & $\dot{4}$ &  & $\mathring{26}$ & $\dot{7}$ & $\dot{8}$ & ${9}$ &  & $\mathring{13}$ & $\mathring{13}$ & $\mathring{18}$ & $\mathring{18}$ & ${36}$ &  & ${40}$ & ${25}$ & $\mathring{26}$ & $\mathring{26}$ & $\dot{7}$ & $\dot{7}$ & $\dot{8}$ & $\dot{8}$ & $\dot{4}$ & $\dot{4}$ & ${25}$ & ${9}$ & $\dot{33}$ & $\dot{33}$ & ${36}$ & ${40}$ &  & ${9}$ &  &  &  &  &  &  &  &  &  &  & ${36}$ &  &  & ${40}$ &  &  &  &  &  &  & ${25}$ &  \\ 
     \hline
     & $\mathring{1}$ & $\mathring{1}$ & $\mathring{1}$ & $\mathring{21}$ &  & $\dot{5}$ & $\mathring{24}$ &  &  & ${29}$ & ${11}$ & ${14}$ & ${17}$ & ${31}$ & ${19}$ & $\mathring{21}$ & $\mathring{21}$ & $\mathring{24}$ & $\mathring{24}$ & ${37}$ &  & ${19}$ & ${11}$ &  & ${17}$ & ${29}$ & ${14}$ & $\dot{5}$ & $\dot{5}$ & ${31}$ &  & ${37}$ &  &  &  & &  &&  &  & & &  &  &  &  & &  &  & &  &  &  &  &  &  &  \\ 
     & $\dot{2}$ &  &  & $\dot{2}$ &  &  & $\dot{2}$ &  &  &  & ${12}$ &  &  & ${39}$ &  &  & ${23}$ & ${35}$ &  &  & ${28}$ &  &  & ${12}$ & ${23}$ &  &  &  & ${28}$ &  & ${35}$ &  & ${39}$ &  &  &  & &  &  &  &  &  &  &  & &  &  & &  &  &  &  & &  &  &  &  \\ 
    $\mathcal{B}_{33}$ & $\mathring{10}$ &  & $\dot{38}$ & $\mathring{3}$ & $\mathring{3}$ & $\mathring{3}$ & $\mathring{6}$ & $\mathring{6}$ & $\mathring{6}$ & $\mathring{10}$ & $\mathring{10}$ & ${15}$ & ${16}$ &  & ${20}$ & ${34}$ & ${22}$ &  & ${30}$ & ${32}$ & ${27}$ & ${16}$ & ${22}$ & ${20}$ &  & ${30}$ & ${27}$ & ${15}$ &  & ${32}$ & ${34}$ & $\dot{38}$ & $\dot{38}$ &  &  &  &  &  & & &  &  &  & &  &  &  &  &  &  & &  &  &  & &  &  \\ 
     & $\mathring{13}$ & $\dot{33}$ &  & $\mathring{18}$ & $\dot{4}$ &  & $\mathring{26}$ & $\dot{7}$ & $\dot{8}$ & ${9}$ &  & $\mathring{13}$ & $\mathring{13}$ & $\mathring{18}$ & $\mathring{18}$ & ${36}$ &  & ${40}$ & ${25}$ & $\mathring{26}$ & $\mathring{26}$ & $\dot{7}$ & $\dot{7}$ & $\dot{8}$ & $\dot{8}$ & $\dot{4}$ & $\dot{4}$ & ${25}$ & ${9}$ & $\dot{33}$ & $\dot{33}$ & ${36}$ & ${40}$ &  &  &  &  &  &  &  &  &  &  &  &  &  &  &  &  &  &  &  &  &  &  & &  \\ 
    \hline \hline
    \end{tabular}
    }}
\caption{MP-57, P-33, and their optimal B-KS. The upper part of the table displays the $57$ vectors of MP-57. Each column corresponds to one vector. Each row gives one component. Here $t$ denotes $\sqrt{2}$. MP-57 has $|\mathcal{B}| = 40$ orthogonal bases, numbered from 1 to 40.
The second (middle) part of the table presents an optimal B-KS set for MP-57. The notation is the following: $\dot{2}$ indicates that basis 2 is only in $S_A$. $\mathring{1}$ indicates that basis 1 is only in $S_B$. 11 indicates that basis 11 is in neither $S_A$ nor $S_B$. The optimal B-KS set given is thus $S_A = \{\dot{2}, \dot{4}, \dot{5}, \dot{7}, \dot{8}, \dot{33}, \dot{38}\}, S_B = \{\mathring{1}, \mathring{3}, \mathring{6}, \mathring{10}, \mathring{13}, \mathring{18}, \mathring{21}, \mathring{24}, \mathring{26}\}.$
P-33 is a subset of MP-57, and so its vectors and orthogonal bases are subsets of MP-57. In the third (bottom) part of the table, we present an optimal B-KS set for P-33. It is the same as for MP-57.}
    \label{tab:MP57+soleheir}
\end{table*}


\subsection{KP-36}


Kernaghan and Peres' 36-vector KS set in $d=8$ \cite{Kernaghan:1995PLA} hereafter called KP-36, is a critical KS set. It has $11$ orthogonal bases, and optimal B-KS sets have size $|S_A| = 6, |S_B| = 8$. The KS set and one of its optimal B-KS sets are shown in Table~\ref{tab:KP40+soleheir}.

We have been able to compute all $242$ B-KS capable sets, of which $66$ are essential. The B-KS capable sets up to isomorphism are listed by size in Table~\ref{tab:BKS-capable-iso}.

We briefly remark that in the original article there is a typo --- removing the four vectors mentioned on page two of Ref. \cite{Kernaghan:1995PLA} does not yield a KS set. 


\subsection{KP-40} \label{sec:KP-40}


Kernaghan and Peres' 40-vector KS set in $d=8$ \cite{Kernaghan:1995PLA}, hereafter called KP-40, has $25$ orthogonal bases and is an extension of KP-36. The KS set, as well as an optimal B-KS set (with $|S_A| = 3, |S_B| = 4$) are illustrated in Table~\ref{tab:KP40+soleheir}. 
Any B-KS set yielding a smaller input cardinality would have to have $|S_A| = |S_B| =3$ (recalling the discussion at the beginning of Sec.~\ref{sec:perf-quantum-strategies}), and so by Eq. \eqref{eq:lowerbound}, the associated KS set would have at most $8(3+3) = 48$ vectors.


\begin{table*}[h]
    \centering
    {\scalebox{0.75}{
    \begin{tabular}{crrrrrrrrrrrrrrrrrrrrrrrrrrrrrrrrrrrrrrrr}
    \hline \hline
     & $v_{1}$ & $v_{2}$ & $v_{3}$ & $v_{4}$ & $v_{5}$ & $v_{6}$ & $v_{7}$ & $v_{8}$ & $v_{9}$ & $v_{10}$ & $v_{11}$ & $v_{12}$ & $v_{13}$ & $v_{14}$ & $v_{15}$ & $v_{16}$ & $v_{17}$ & $v_{18}$ & $v_{19}$ & $v_{20}$ & $v_{21}$ & $v_{22}$ & $v_{23}$ & $v_{24}$ & $v_{25}$ & $v_{26}$ & $v_{27}$ & $v_{28}$ & $v_{29}$ & $v_{30}$ & $v_{31}$ & $v_{32}$ & $v_{33}$ & $v_{34}$ & $v_{35}$ & $v_{36}$ & $v_{37}$ & $v_{38}$ & $v_{39}$ & $v_{40}$ \\ 
    \hline
    $v_{i1}$ & $0$ & $1$ & $1$ & $0$ & $1$ & $0$ & $0$ & $-1$ & $1$ & $0$ & $0$ & $0$ & $0$ & $0$ & $0$ & $0$ & $1$ & $0$ & $1$ & $0$ & $1$ & $0$ & $1$ & $0$ & $0$ & $0$ & $1$ & $1$ & $0$ & $0$ & $1$ & $1$ & $0$ & $0$ & $0$ & $0$ & $1$ & $1$ & $1$ & $1$ \\
    $v_{i2}$ & $1$ & $0$ & $0$ & $1$ & $0$ & $1$ & $-1$ & $0$ & $0$ & $1$ & $0$ & $0$ & $0$ & $0$ & $0$ & $0$ & $0$ & $1$ & $0$ & $1$ & $0$ & $1$ & $0$ & $1$ & $0$ & $0$ & $-1$ & $1$ & $0$ & $0$ & $-1$ & $1$ & $0$ & $0$ & $0$ & $0$ & $-1$ & $1$ & $-1$ & $1$ \\
    $v_{i3}$ & $1$ & $0$ & $0$ & $1$ & $0$ & $-1$ & $1$ & $0$ & $0$ & $0$ & $1$ & $0$ & $0$ & $0$ & $0$ & $0$ & $1$ & $0$ & $-1$ & $0$ & $1$ & $0$ & $-1$ & $0$ & $1$ & $1$ & $0$ & $0$ & $1$ & $1$ & $0$ & $0$ & $0$ & $0$ & $0$ & $0$ & $-1$ & $-1$ & $1$ & $1$ \\
    $v_{i4}$ & $0$ & $1$ & $1$ & $0$ & $-1$ & $0$ & $0$ & $1$ & $0$ & $0$ & $0$ & $1$ & $0$ & $0$ & $0$ & $0$ & $0$ & $1$ & $0$ & $-1$ & $0$ & $1$ & $0$ & $-1$ & $-1$ & $1$ & $0$ & $0$ & $-1$ & $1$ & $0$ & $0$ & $0$ & $0$ & $0$ & $0$ & $1$ & $-1$ & $-1$ & $1$ \\
    $v_{i5}$ & $1$ & $0$ & $0$ & $-1$ & $0$ & $1$ & $1$ & $0$ & $0$ & $0$ & $0$ & $0$ & $1$ & $0$ & $0$ & $0$ & $1$ & $0$ & $1$ & $0$ & $-1$ & $0$ & $-1$ & $0$ & $0$ & $0$ & $-1$ & $-1$ & $0$ & $0$ & $1$ & $1$ & $1$ & $1$ & $1$ & $1$ & $0$ & $0$ & $0$ & $0$ \\
    $v_{i6}$ & $0$ & $1$ & $-1$ & $0$ & $1$ & $0$ & $0$ & $1$ & $0$ & $0$ & $0$ & $0$ & $0$ & $1$ & $0$ & $0$ & $0$ & $1$ & $0$ & $1$ & $0$ & $-1$ & $0$ & $-1$ & $0$ & $0$ & $1$ & $-1$ & $0$ & $0$ & $-1$ & $1$ & $-1$ & $1$ & $-1$ & $1$ & $0$ & $0$ & $0$ & $0$ \\
    $v_{i7}$ & $0$ & $-1$ & $1$ & $0$ & $1$ & $0$ & $0$ & $1$ & $0$ & $0$ & $0$ & $0$ & $0$ & $0$ & $1$ & $0$ & $1$ & $0$ & $-1$ & $0$ & $-1$ & $0$ & $1$ & $0$ & $-1$ & $-1$ & $0$ & $0$ & $1$ & $1$ & $0$ & $0$ & $-1$ & $-1$ & $1$ & $1$ & $0$ & $0$ & $0$ & $0$ \\
    $v_{i8}$ & $-1$ & $0$ & $0$ & $1$ & $0$ & $1$ & $1$ & $0$ & $0$ & $0$ & $0$ & $0$ & $0$ & $0$ & $0$ & $1$ & $0$ & $1$ & $0$ & $-1$ & $0$ & $-1$ & $0$ & $1$ & $1$ & $-1$ & $0$ & $0$ & $-1$ & $1$ & $0$ & $0$ & $1$ & $-1$ & $-1$ & $1$ & $0$ & $0$ & $0$ & $0$ \\
    \hline
     & $\ddot{1}$ & $\ddot{1}$ & $\ddot{1}$ & $\ddot{1}$ & $\ddot{1}$ & $\ddot{1}$ & $\ddot{1}$ & $\ddot{1}$ & $\mathring{6}$ & $\mathring{6}$ & $\mathring{6}$ & $\mathring{6}$ & $\mathring{10}$ & $\mathring{10}$ & $\mathring{10}$ & $\mathring{10}$ & ${18}$ & ${18}$ & ${18}$ & ${18}$ & ${15}$ & ${15}$ & ${15}$ & ${15}$ & ${18}$ & ${18}$ & ${18}$ & ${18}$ & ${15}$ & ${15}$ & ${15}$ & ${15}$ & $\mathring{6}$ & $\mathring{6}$ & $\mathring{6}$ & $\mathring{6}$ & $\mathring{10}$ & $\mathring{10}$ & $\mathring{10}$ & $\mathring{10}$ \\ 
     & ${2}$ & ${2}$ & ${2}$ & ${2}$ & ${3}$ & ${3}$ & ${3}$ & ${3}$ & ${16}$ & ${8}$ & ${16}$ & ${8}$ & ${16}$ & ${8}$ & ${16}$ & ${8}$ & ${8}$ & ${16}$ & ${8}$ & ${16}$ & ${8}$ & ${16}$ & ${8}$ & ${16}$ & $\dot{20}$ & $\dot{20}$ & $\dot{20}$ & $\dot{20}$ & $\dot{20}$ & $\dot{20}$ & $\dot{20}$ & $\dot{20}$ & ${2}$ & ${3}$ & ${3}$ & ${2}$ & ${3}$ & ${2}$ & ${2}$ & ${3}$ \\ 
    $\mathcal{B}_{40}$ & ${5}$ & ${4}$ & ${4}$ & ${5}$ & ${4}$ & ${5}$ & ${5}$ & ${4}$ & ${5}$ & ${4}$ & ${4}$ & ${5}$ & ${4}$ & ${5}$ & ${5}$ & ${4}$ & $\ddot{14}$ & $\ddot{14}$ & $\ddot{14}$ & $\ddot{14}$ & $\ddot{14}$ & $\ddot{14}$ & $\ddot{14}$ & $\ddot{14}$ & ${22}$ & ${24}$ & ${22}$ & ${24}$ & ${22}$ & ${24}$ & ${22}$ & ${24}$ & ${24}$ & ${22}$ & ${24}$ & ${22}$ & ${24}$ & ${22}$ & ${24}$ & ${22}$ \\ 
     & ${19}$ & ${19}$ & ${23}$ & ${23}$ & ${19}$ & ${19}$ & ${23}$ & ${23}$ & ${7}$ & ${7}$ & ${7}$ & ${7}$ & ${7}$ & ${7}$ & ${7}$ & ${7}$ & ${12}$ & ${12}$ & ${17}$ & ${17}$ & ${12}$ & ${12}$ & ${17}$ & ${17}$ & ${19}$ & ${23}$ & ${23}$ & ${19}$ & ${23}$ & ${19}$ & ${19}$ & ${23}$ & ${12}$ & ${12}$ & ${17}$ & ${17}$ & ${12}$ & ${12}$ & ${17}$ & ${17}$ \\ 
     & ${13}$ & ${11}$ & ${13}$ & ${11}$ & ${13}$ & ${11}$ & ${13}$ & ${11}$ & ${21}$ & ${21}$ & ${9}$ & ${9}$ & ${21}$ & ${21}$ & ${9}$ & ${9}$ & ${11}$ & ${13}$ & ${13}$ & ${11}$ & ${13}$ & ${11}$ & ${11}$ & ${13}$ & ${21}$ & ${21}$ & ${9}$ & ${9}$ & ${21}$ & ${21}$ & ${9}$ & ${9}$ & ${25}$ & ${25}$ & ${25}$ & ${25}$ & ${25}$ & ${25}$ & ${25}$ & ${25}$ \\ 
     \hline
     &  & &  &  &  &  &  &  & &  &  &  & $\mathring{10}$ & $\mathring{10}$ & $\mathring{10}$ & $\mathring{10}$ &  & &  &  & $\dot{15}$ & $\dot{15}$ & $\dot{15}$ & $\dot{15}$ & &  &  &  & $\dot{15}$ & $\dot{15}$ & $\dot{15}$ & $\dot{15}$ &  &  &  &  & $\mathring{10}$ & $\mathring{10}$ & $\mathring{10}$ & $\mathring{10}$ \\ 
     &  &  &  &  & $\mathring{3}$ & $\mathring{3}$ & $\mathring{3}$ & $\mathring{3}$ & & $\ddot{8}$ &  & $\ddot{8}$ &  & $\ddot{8}$ &  & $\ddot{8}$ & $\ddot{8}$ &  & $\ddot{8}$ &  & $\ddot{8}$ &  & $\ddot{8}$ &  & &  &  &  &  &  &  &  &  & $\mathring{3}$ & $\mathring{3}$ &  & $\mathring{3}$ &  &  & $\mathring{3}$ \\ 
    $\mathcal{B}_{36}$ & & $\ddot{4}$ & $\ddot{4}$ &  & $\ddot{4}$ &  &  & $\ddot{4}$ & & $\ddot{4}$ & $\ddot{4}$ &  & $\ddot{4}$ &  &  & $\ddot{4}$ &  & &  &  &  &  &  &  & & $\mathring{24}$ &  & $\mathring{24}$ &  & $\mathring{24}$ &  & $\mathring{24}$ & $\mathring{24}$ &  & $\mathring{24}$ &  & $\mathring{24}$ &  & $\mathring{24}$ &  \\ 
     & &  & $\dot{23}$ & $\dot{23}$ &  &  & $\dot{23}$ & $\dot{23}$ & &  &  &  &  &  &  &  &  & & $\mathring{17}$ & $\mathring{17}$ &  &  & $\mathring{17}$ & $\mathring{17}$ & & $\dot{23}$ & $\dot{23}$ &  & $\dot{23}$ &  &  & $\dot{23}$ &  &  & $\mathring{17}$ & $\mathring{17}$ &  &  & $\mathring{17}$ & $\mathring{17}$ \\ 
     & & $\ddot{11}$ &  & $\ddot{11}$ &  & $\ddot{11}$ &  & $\ddot{11}$ & &  & $\dot{9}$ & $\dot{9}$ &  &  & $\dot{9}$ & $\dot{9}$ & $\ddot{11}$ & &  & $\ddot{11}$ &  & $\ddot{11}$ & $\ddot{11}$ &  & &  & $\dot{9}$ & $\dot{9}$ &  &  & $\dot{9}$ & $\dot{9}$ & $\mathring{25}$ & $\mathring{25}$ & $\mathring{25}$ & $\mathring{25}$ & $\mathring{25}$ & $\mathring{25}$ & $\mathring{25}$ & $\mathring{25}$ \\ 
     \hline \hline
    \end{tabular}
    }}
\caption{KP-40, KP-36, and their optimal B-KS. The upper part of the table displays the $40$ vectors of KP-40, with each column corresponding to a vector, and each row to a component. The second part of the table presents an optimal B-KS set for KP-40. KP-40 has $|\mathcal{B}| = 24$ orthogonal bases, numbered from $1$ to $24$. The notation is the following: $\dot{20}$ indicates that basis $20$ is only in $S_A$. $\mathring{6}$ indicates that basis $6$ is only in $S_B$. $\protect\ddot{1}$ indicates that the basis $1$ is in both $S_A$ and $S_B$. $2$ indicates that the basis $2$ is in neither $S_A$ nor $S_B$. An optimal B-KS set of KP-40 is thus
$S_A = \{\ddot{1}, \ddot{14}, \dot{20}\}, S_B= \{\ddot{1}, \mathring{6}, \mathring{10}, \ddot{14}\}.$ KP-36 is a subset of KP-40, and so its vectors and orthogonal bases are subsets of those of KP-40.
The third part of the table illustrates an optimal B-KS set for KP-36 using the same notation as for KP-40. The optimal B-KS set illustrated for KP-36 is $S_A = \{\ddot{4}, \ddot{8}, \dot{9}, \ddot{11}, \dot{15}, \dot{23}\}, S_B =\{\mathring{3}, \ddot{4},\ddot{8},\mathring{10},\ddot{11},\mathring{17},\mathring{24}, \mathring{25}\}$.}
    \label{tab:KP40+soleheir}
\end{table*}


\subsection{S-29} \label{sec:S-29}


The KS set S-29 is built from CEG-18 via a recursive construction introduced in Ref. \cite{Cabello:2005PLA} (therein called S-5). It is in dimension $5$, consists of $29$ vectors and $16$ orthogonal bases, and is critical. The KS set, as well as one of its optimal B-KS sets (with $|S_A| = 6, |S_B| = 9)$ are illustrated in Table~\ref{tab:S-29}.


\begin{table*}[h]
    \centering
    \begin{tabular}{crrrrrrrrrrrrrrrrrrrrrrrrrrrrr}
    \hline \hline
     & $v_{1}$ & $v_{2}$ & $v_{3}$ & $v_{4}$ & $v_{5}$ & $v_{6}$ & $v_{7}$ & $v_{8}$ & $v_{9}$ & $v_{10}$ & $v_{11}$ & $v_{12}$ & $v_{13}$ & $v_{14}$ & $v_{15}$ & $v_{16}$ & $v_{17}$ & $v_{18}$ & $v_{19}$ & $v_{20}$ & $v_{21}$ & $v_{22}$ & $v_{23}$ & $v_{24}$ & $v_{25}$ & $v_{26}$ & $v_{27}$ & $v_{28}$ & $v_{29}$ \\ 
    \hline
    $v_{i1}$ & $1$ & $1$ & $0$ & $0$ & $0$ & $1$ & $-1$ & $0$ & $0$ & $1$ & $1$ & $1$ & $0$ & $0$ & $0$ & $0$ & $0$ & $0$ & $1$ & $1$ & $0$ & $1$ & $0$ & $0$ & $0$ & $0$ & $1$ & $0$ & $0$ \\
    $v_{i2}$ & $1$ & $-1$ & $1$ & $1$ & $0$ & $-1$ & $1$ & $1$ & $1$ & $0$ & $1$ & $1$ & $0$ & $1$ & $1$ & $0$ & $1$ & $0$ & $0$ & $1$ & $0$ & $0$ & $1$ & $0$ & $-1$ & $0$ & $-1$ & $1$ & $1$ \\
    $v_{i3}$ & $1$ & $0$ & $0$ & $-1$ & $1$ & $1$ & $1$ & $1$ & $0$ & $0$ & $0$ & $-1$ & $0$ & $-1$ & $1$ & $1$ & $1$ & $0$ & $0$ & $-1$ & $1$ & $1$ & $-1$ & $1$ & $1$ & $0$ & $-1$ & $0$ & $1$ \\
    $v_{i4}$ & $-1$ & $0$ & $0$ & $0$ & $-1$ & $-1$ & $1$ & $1$ & $-1$ & $1$ & $0$ & $1$ & $1$ & $1$ & $-1$ & $1$ & $-1$ & $1$ & $0$ & $-1$ & $0$ & $0$ & $-1$ & $0$ & $1$ & $0$ & $1$ & $1$ & $0$ \\
    $v_{i5}$ & $0$ & $0$ & $1$ & $0$ & $0$ & $0$ & $0$ & $-1$ & $0$ & $0$ & $0$ & $0$ & $1$ & $-1$ & $-1$ & $0$ & $1$ & $0$ & $0$ & $0$ & $-1$ & $0$ & $1$ & $1$ & $1$ & $1$ & $0$ & $0$ & $0$ \\
    \hline
     &  &  & $\mathring{1}$ &  & $\mathring{1}$ &  &  & $\mathring{1}$ &  &  &  &  &  &  &  &  &  &  & $\mathring{1}$ & $\dot{13}$ &  & $\dot{13}$ &  &  & $\mathring{1}$ & $\dot{13}$ & $\dot{13}$ & $\dot{13}$ &  \\ 
     &  &  &  & $\mathring{2}$ &  & $\mathring{11}$ &  & $\mathring{2}$ &  & $\mathring{11}$ &  &  & $\mathring{2}$ &  &  &  & $\mathring{2}$ &  & $\mathring{2}$ & $\mathring{11}$ &  &  &  &  &  & $\mathring{11}$ &  &  & $\mathring{11}$ \\ 
     & $\mathring{3}$ & $\mathring{3}$ &  &  &  &  &  &  & ${15}$ &  &  & $\mathring{3}$ &  &  &  & $\mathring{3}$ &  &  & ${15}$ &  & ${15}$ &  &  & ${15}$ &  & $\mathring{3}$ &  & ${15}$ &  \\ 
     &  &  & $\dot{6}$ &  &  &  & $\mathring{4}$ &  & $\mathring{4}$ &  &  & $\mathring{4}$ &  & $\dot{6}$ & $\dot{6}$ & $\dot{6}$ &  &  & $\dot{6}$ &  &  & $\mathring{4}$ &  &  &  & $\mathring{4}$ &  &  &  \\ 
    $\mathcal{B}_{29}$ & $\dot{10}$ & $\dot{9}$ &  & $\dot{10}$ &  &  & $\dot{10}$ &  &  & $\dot{10}$ & $\dot{9}$ &  & $\dot{5}$ & $\dot{5}$ &  &  &  & $\dot{9}$ & $\dot{5}$ &  & $\dot{9}$ &  & $\dot{5}$ & $\dot{9}$ &  & $\dot{10}$ &  &  & $\dot{5}$ \\ 
     &  & $\mathring{7}$ &  &  & $\mathring{7}$ &  &  &  &  &  & $\mathring{7}$ &  &  &  &  & $\mathring{7}$ & $\dot{14}$ &  & $\dot{14}$ &  & $\dot{14}$ &  &  &  & $\dot{14}$ & $\mathring{7}$ &  & $\dot{14}$ &  \\ 
     &  &  &  &  &  & $\mathring{8}$ &  &  &  &  & $\mathring{8}$ &  &  &  & $\mathring{16}$ & $\mathring{8}$ &  &  & $\mathring{16}$ &  &  &  & $\mathring{16}$ & $\mathring{16}$ &  & $\mathring{8}$ & $\mathring{8}$ & $\mathring{16}$ &  \\ 
     &  &  &  & $\mathring{12}$ &  &  &  &  &  &  &  &  &  &  &  &  &  & $\mathring{12}$ & $\mathring{12}$ &  &  &  &  &  &  & $\mathring{12}$ &  &  & $\mathring{12}$ \\ 
    \hline \hline
    \end{tabular}
    \caption{S-29 and an optimal B-KS. The upper part of the table displays the $29$ vectors of S-29. Each column corresponds to one vector. Each row gives one component. The lower part of the table presents an optimal B-KS. S-29 has $|\mathcal{B}|=16$ bases, the bases are numbered from 1 to 16. Each vector belongs to the bases indicated the lower part of the column corresponding to that vector. The notation is the following: $\dot{5}$ indicates that basis $5$ is only in $S_A$. $\mathring{1}$ indicates that basis $1$ is only in $S_B$. $15$ indicates that basis $15$ is in neither $S_A$ nor $S_B$.
    The optimal B-KS set illustrated is thus $S_A = \{ \dot{5}, \dot{6}, \dot{9}, \dot{10}, \dot{13}, \dot{14}\}$, $S_B = \{\mathring{1}, \mathring{2}, \mathring{3}, \mathring{4}, \mathring{7}, \mathring{8}, \mathring{11}, \mathring{12}, \mathring{16}\}.$}
    \label{tab:S-29}
\end{table*}


We have been able to compute all $13251$ B-KS capable sets in this case, of which $187$ are essential. Table~\ref{tab:BKS-capable-iso} illustrates the number of B-KS capable sets by size up to isomorphism.


\subsection{S-31} \label{sec:S31}


In the recursive construction of Ref. \cite{Cabello:2005PLA} extra vectors [specifically $(0,1,0,0,0),(0,0,1,0,0)$] are eliminated in order to obtain a critical KS set. We define S-31 to be the noncritical KS set that is obtained before the elimination of these vectors (equivalently S-29 with the additional vectors previously mentioned). This $5$-dimensional KS set has $31$ vectors and $21$ orthogonal bases. Optimal B-KS sets have the same size as for S-29: $|S_A| = 6, |S_B| = 9$.


\subsection{S-34} \label{sec:S34}


The KS set in dimension $7$ with fewest number of vectors known is S-34, which was introduced in Ref. \cite{Cabello:1996PLA} (therein called S-7). It has $28$ orthogonal bases and is critical. The KS set as well as an optimal B-KS set (with $|S_A| = 6, |S_B| = 8$) is illustrated in Table~\ref{tab:S34}.


\begin{table*}[h]
    \centering
    {\scalebox{0.8}{
    \begin{tabular}{ccccccccccccccccccccccccccccccccccc}
    \hline \hline
     & $v_{1}$ & $v_{2}$ & $v_{3}$ & $v_{4}$ & $v_{5}$ & $v_{6}$ & $v_{7}$ & $v_{8}$ & $v_{9}$ & $v_{10}$ & $v_{11}$ & $v_{12}$ & $v_{13}$ & $v_{14}$ & $v_{15}$ & $v_{16}$ & $v_{17}$ & $v_{18}$ & $v_{19}$ & $v_{20}$ & $v_{21}$ & $v_{22}$ & $v_{23}$ & $v_{24}$ & $v_{25}$ & $v_{26}$ & $v_{27}$ & $v_{28}$ & $v_{29}$ & $v_{30}$ & $v_{31}$ & $v_{32}$ & $v_{33}$ & $v_{34}$ \\ 
    \hline
    $v_{i1}$ & $0$ & $1$ & $0$ & $1$ & $1$ & $0$ & $0$ & $1$ & $1$ & $0$ & $0$ & $0$ & $1$ & $0$ & $0$ & $0$ & $0$ & $0$ & $0$ & $0$ & $1$ & $1$ & $0$ & $1$ & $0$ & $-1$ & $0$ & $0$ & $0$ & $0$ & $1$ & $0$ & $0$ & $0$ \\
    $v_{i2}$ & $0$ & $1$ & $0$ & $1$ & $0$ & $0$ & $0$ & $-1$ & $0$ & $0$ & $0$ & $0$ & $1$ & $1$ & $0$ & $0$ & $0$ & $0$ & $0$ & $1$ & $-1$ & $0$ & $0$ & $-1$ & $1$ & $1$ & $0$ & $0$ & $0$ & $0$ & $1$ & $1$ & $0$ & $0$ \\
    $v_{i3}$ & $0$ & $1$ & $0$ & $-1$ & $1$ & $0$ & $0$ & $-1$ & $0$ & $1$ & $0$ & $0$ & $0$ & $0$ & $0$ & $0$ & $0$ & $1$ & $0$ & $1$ & $1$ & $0$ & $0$ & $0$ & $0$ & $1$ & $0$ & $0$ & $0$ & $0$ & $-1$ & $-1$ & $0$ & $0$ \\
    $v_{i4}$ & $0$ & $-1$ & $0$ & $-1$ & $0$ & $0$ & $1$ & $1$ & $1$ & $0$ & $-1$ & $0$ & $0$ & $-1$ & $0$ & $0$ & $1$ & $1$ & $1$ & $0$ & $-1$ & $0$ & $1$ & $0$ & $1$ & $1$ & $1$ & $1$ & $1$ & $1$ & $1$ & $0$ & $1$ & $0$ \\
    $v_{i5}$ & $0$ & $0$ & $1$ & $0$ & $0$ & $1$ & $1$ & $0$ & $0$ & $0$ & $1$ & $1$ & $0$ & $0$ & $0$ & $0$ & $-1$ & $0$ & $1$ & $0$ & $0$ & $0$ & $0$ & $0$ & $0$ & $0$ & $-1$ & $1$ & $0$ & $1$ & $0$ & $0$ & $-1$ & $1$ \\
    $v_{i6}$ & $1$ & $0$ & $0$ & $0$ & $0$ & $0$ & $0$ & $0$ & $0$ & $0$ & $1$ & $-1$ & $0$ & $0$ & $0$ & $1$ & $0$ & $0$ & $1$ & $0$ & $0$ & $0$ & $1$ & $0$ & $0$ & $0$ & $1$ & $-1$ & $0$ & $-1$ & $0$ & $0$ & $-1$ & $1$ \\
    $v_{i7}$ & $1$ & $0$ & $-1$ & $0$ & $0$ & $1$ & $0$ & $0$ & $0$ & $0$ & $1$ & $0$ & $0$ & $0$ & $1$ & $0$ & $0$ & $0$ & $-1$ & $0$ & $0$ & $0$ & $0$ & $0$ & $0$ & $0$ & $-1$ & $-1$ & $1$ & $1$ & $0$ & $0$ & $1$ & $0$ \\
    \hline
     & $\mathring{3}$ &  &  & ${26}$ & ${26}$ & ${18}$ &  & ${26}$ &  & $\mathring{3}$ &  & ${26}$ & $\mathring{3}$ &  & ${26}$ &  & $\mathring{3}$ &  & $\mathring{3}$ & ${18}$ &  & ${18}$ & ${18}$ & $\mathring{3}$ & ${26}$ &  &  & ${18}$ &  & $\mathring{3}$ &  & ${18}$ & ${18}$ & ${26}$ \\ 
     & ${16}$ &  &  &  & ${13}$ & $\mathring{5}$ &  &  &  & $\mathring{5}$ &  & ${13}$ & $\mathring{5}$ & ${13}$ & ${13}$ &  & ${16}$ &  & ${16}$ & ${16}$ &  & ${16}$ & $\mathring{5}$ & $\mathring{5}$ &  & ${13}$ &  & $\mathring{5}$ &  & ${16}$ & ${13}$ & ${16}$ & $\mathring{5}$ & ${13}$ \\ 
     & ${8}$ & $\dot{28}$ & $\dot{25}$ & $\dot{25}$ & $\dot{25}$ & $\dot{25}$ & ${8}$ & $\dot{25}$ & $\dot{28}$ & ${8}$ &  & $\dot{28}$ & ${8}$ &  & $\dot{28}$ & $\dot{25}$ &  &  &  &  &  &  &  & ${8}$ & $\dot{25}$ & $\dot{28}$ & ${8}$ &  &  &  &  & $\dot{28}$ & ${8}$ & $\dot{28}$ \\ 
     &  &  & ${17}$ &  &  &  &  & ${2}$ &  &  & ${17}$ & ${2}$ & ${2}$ &  & ${2}$ &  &  & ${2}$ &  & ${17}$ & ${2}$ & ${17}$ & ${17}$ &  &  &  &  &  &  & ${17}$ &  & ${17}$ &  & ${2}$ \\ 
     &  &  & $\mathring{1}$ &  &  & $\mathring{1}$ &  & $\mathring{1}$ &  & ${11}$ &  & ${11}$ & $\mathring{1}$ & ${11}$ & ${11}$ & $\mathring{1}$ &  & $\mathring{1}$ &  &  & $\mathring{1}$ & ${11}$ &  &  & ${11}$ &  &  &  &  &  &  &  &  & ${11}$ \\ 
     &  &  & $\dot{4}$ &  &  &  & ${20}$ &  &  & $\dot{4}$ & $\dot{4}$ &  & $\dot{4}$ &  & ${20}$ & ${20}$ & ${20}$ &  &  & ${20}$ &  & ${20}$ & $\dot{4}$ & $\dot{4}$ &  &  &  &  &  & $\dot{4}$ &  & ${20}$ &  &  \\ 
    $\mathcal{B}_{34}$ &  &  & ${14}$ & ${14}$ &  & ${14}$ &  &  & ${14}$ & $\dot{6}$ & $\dot{6}$ & $\dot{6}$ & $\dot{6}$ &  &  & ${14}$ &  &  & $\dot{6}$ & ${14}$ & ${14}$ &  &  & $\dot{6}$ &  &  &  &  & $\dot{6}$ &  &  &  &  &  \\ 
     &  &  &  &  &  &  & $\mathring{7}$ &  &  & $\mathring{7}$ & ${19}$ & ${19}$ & $\mathring{7}$ &  & $\mathring{7}$ & $\mathring{7}$ & $\mathring{7}$ &  & ${19}$ & ${19}$ &  & ${19}$ &  & $\mathring{7}$ &  &  &  &  & ${19}$ &  &  & ${19}$ &  &  \\ 
     &  & $\dot{23}$ & $\dot{23}$ &  &  & $\dot{23}$ &  &  &  &  &  &  &  &  &  & $\dot{23}$ &  & $\dot{23}$ &  & $\mathring{22}$ &  & $\mathring{22}$ &  & $\dot{23}$ &  &  & $\mathring{22}$ & $\mathring{22}$ & $\mathring{22}$ &  & $\dot{23}$ & $\mathring{22}$ &  & $\mathring{22}$ \\ 
     &  & ${27}$ & ${27}$ &  &  & ${27}$ &  &  & ${27}$ & ${9}$ &  &  & ${9}$ &  &  & ${27}$ &  &  &  &  &  &  &  & ${9}$ &  & ${27}$ & ${9}$ & ${9}$ & ${9}$ &  &  & ${27}$ &  & ${9}$ \\ 
     &  & ${24}$ & $\mathring{10}$ &  &  & $\mathring{10}$ &  &  &  & $\mathring{10}$ &  & ${24}$ &  & $\mathring{10}$ & ${24}$ & $\mathring{10}$ &  & ${24}$ &  &  &  & $\mathring{10}$ &  & ${24}$ & $\mathring{10}$ &  &  &  &  &  & ${24}$ &  &  & ${24}$ \\ 
     &  &  & $\mathring{12}$ & $\mathring{15}$ & $\mathring{12}$ & $\mathring{12}$ &  &  & $\mathring{15}$ &  &  & $\mathring{15}$ &  & $\mathring{12}$ & $\mathring{15}$ & $\mathring{12}$ &  &  &  & $\mathring{15}$ & $\mathring{15}$ &  &  &  &  & $\mathring{12}$ &  &  &  &  & $\mathring{12}$ &  &  & $\mathring{15}$ \\ 
     & $\dot{21}$ &  &  &  &  &  & $\dot{21}$ &  &  &  &  &  &  &  &  &  &  &  &  & $\dot{21}$ &  & $\dot{21}$ &  &  &  &  & $\dot{21}$ &  &  &  &  & $\dot{21}$ & $\dot{21}$ &  \\ 
    \hline \hline
    \end{tabular}}}
\caption{S-34 and an optimal B-KS. The upper part of the table displays the $34$ vectors of S-34. Each column corresponds to one vector. Each row gives one component. The lower part of the table presents an optimal B-KS. S-34 has $|\mathcal{B}|=28$ bases, the bases are numbered from 1 to 28. Each vector belongs to the bases indicated the lower part of the column corresponding to that vector. The notation is the following: $\dot{4}$ indicates that basis $4$ is only in $S_A$. $\mathring{3}$ indicates that basis $3$ is only in $S_B$. $16$ indicates that basis $16$ is in neither $S_A$ nor $S_B$.
    The optimal B-KS set illustrated is thus $S_A = \{\dot{4}, \dot{6}, \dot{21}, \dot{23}, \dot{25}, \dot{28} \}$, $S_B = \{\mathring{1}, \mathring{3}, \mathring{5}, \mathring{7}, \mathring{10}, \mathring{12}, \mathring{15}, \mathring{22}\}.$}
    \label{tab:S34}
\end{table*}


\subsection{S-35}


Similarly to S-29, the KS set S-35 is obtained by eliminating vectors in the recursive construction of Ref. \cite{Cabello:2005PLA}. In this case, a single vector [$(0,0,0,1,0,0,0)$] is eliminated (so S-35 may be obtained from S-34 by adding back this vector). S-35 has $32$ orthogonal bases, and as for S-34, optimal B-KSs of S-35 have $|S_A| = 6, |S_B| = 8$.


\section{Discussion and future}
\label{Sec:Discussion}


In light of the recent results outlined in Sec.~\ref{sec:recent-results}, BPQSs have become fundamental objects of study in quantum foundations, quantum information, and quantum computation. Therefore, it is important to know for which input cardinalities one may have BPQS. In this article, we have addressed this question by first sharpening a connection between BPQSs and the KS sets (Theorem \ref{theo:min-KS}), and then exploiting it in order to find BPQSs of smaller (or equal) input cardinality than those constructed by existing methods (Methods I and II described in Sec.~\ref{subsec:connections}).



\begin{table}[h]
    \centering
    \begin{tabular}{ccccc}
    \hline \hline
       KS-set & Dimension & Optimal B-KS set & Method I & Method II \\
       \hline 
       P-33 \cite{Cabello:2023XXX} & 3 & 7-9 & 16-16 & 16-33  \\
       CK-31 & 3 & 8-9  & 17-17 & 17-31 \\
       CK-33 & 3 & 8-9 & 20-20 & 20-33 \\
       CK-37 & 3 & 8-9 & 22-22 & 22-37 \\
       P-24 \cite{Cabello:2023XXX} & 4 & 3-3  & 24-24 & 24-24  \\
       K-20 & 4 & 4-7 & 11-11 & 11-20   \\
       CEG-18 & 4 & 5-6 & 9-9 & 9-18  \\
       ZP-28 & 4 & 6-8 & 14-14 & 14-28  \\
       S-29 & 5 & 6-9 & 16-16 & 16-29  \\
       L-21 \cite{trandafir2024perfectquantumstrategiessmall} & 6 & 5-5 & 7-7 & 7-21 \\
       S-7 & 7 & 6-8 & 22-22 & 22-34 \\
       KP-40 & 8 & 3-4 & 25-25 & 25-40 \\
       KP-36 & 8 & 6-8 & 11-11 & 11-36 \\ 
       TC-36 \cite{trandafir2024perfectquantumstrategiessmall} & 8 & 7-7 & 9-9 & 9-36 \\
       \hline \hline 
    \end{tabular}
\caption{Sizes of the optimal B-KS sets and, therefore, of the BPQSs with minimum $|X| |Y|$) for some KS sets of small cardinality (including those with the smallest cardinality known in each dimension). ``Optimal B-KS set'' gives $|S_A|$-$|S_B|$ (i.e., $|X|$-$|Y|$) for the optimal B-KS associated with the KS set. ``Method I'' and  ``Method II'' give $|S_A|$-$|S_B|$ using, respectively, Methods I and II (see Sec.~\ref{subsec:connections}).}
     \label{tab:B-KS-final}
\end{table}


In Tables \ref{tab:BKS-capable-iso} and \ref{tab:B-KS-final}, we have summarized our results. For completeness, we have added some other recent results from Refs. \cite{Cabello:2023XXX,trandafir2024perfectquantumstrategiessmall}. 

As we can see, the algorithm leads to BPQSs matching the existing records in $d=3$ (associated with P-33) \cite{Cabello:2023XXX} and $d=4$ (associated with P-24) \cite{Cabello:2023XXX}, produces new records in $d=5, 7, 8$, and gives results largely improving those obtained using Methods I and II in each dimension $d=3,4,5,6,7,8$. 
Particularly interesting is the observation that, in $d=8$, there is a nontrivial BPQS with $|X|=3$ and $|Y|=4$. If this cardinality is not minimum, then any B-KS with smaller cardinality would have to correspond to a KS set in dimension $8$ with, at most, $48$ vectors.

The one-to-one connection between BPQSs and KS sets suggests that the BPQSs presented in Table~\ref{tab:B-KS-final} might already include those with the minimum $|X| |Y|$ for each of the dimensions. In the worst case scenario, we are not aware of any qudit-qudit BPQSs with smaller input cardinality for $d=3,4,5,6,7,8$ in the literature.

Some directions for future research are the following: 

{\em Potential improvements of Algorithm \ref{alg:bks}.} Algorithm \ref{alg:bks} is sufficiently efficient in order to compute optimal B-KS sets for each of the B-KS sets that we have considered in this article, with the exception of MP-57 and Pen-40. In both cases, we implemented a modified version of the algorithm which takes into account the symmetries of the input KS. This allows one to drastically reduce the number of calls to the ILPs ``is B-KS'' and ``is B-KS capable.'' This implementation could potentially be improved so as not to generate as many potential B-KS sets that are isomorphic to those that have already been seen. 

The computation of all B-KS capable sets is also an interesting problem, and one where taking advantage of symmetries may be especially fruitful. One possible direction would be to apply an isomorph-free generation approach (see Ref. \cite{Mckay1998}). Such an approach has been applied to KS sets with much success recently in order to prove that the smallest set in dimension $3$ has, at least, $24$ vectors (see Refs. \cite{Kirchweger:2023,Li:2023XXX}). 

Finally, we mention that, in Refs. \cite{Kirchweger:2023,Li:2023XXX}, a SAT-solver approach is applied to check whether a given set is KS. Since both of the ILPs ``is B-KS'' and ``is B-KS capable'' are really feasibility problems, replacing the ILPs with SAT-solver formulations may also be a way to speed up Algorithm \ref{alg:bks}.

{\em The sequences of B-KS capable sets.} The sequence of B-KS capable sets by size up to isomorphism illustrated in Table~\ref{tab:BKS-capable-iso} seems to have some interesting properties. In each case that we have examined, it is unimodular (i.e., it monotonically increases and then monotonically decreases). Moreover, each of the terms on the ``left'' are larger than their corresponding terms on the ``right'' (i.e., 
denoting the sequence of nonzero terms by $(a_0, \dots, a_n)$, we find that $a_i \geq a_{n-i}$ for each $i=1,\dots, \lfloor \frac{n}{2} \rfloor$). These properties do not hold universally for the sequence of essential B-KS capable sets: the exceptions being CK-33 and K-20. 

Another question is to obtain bounds for the range of sizes of all B-KS sets, and the essential ones. For each KS set we tested, apart from CK-33, the range of the essential B-KS capable sets is less than half of the range of that of all B-KS capable sets. 

The set of $S_A$ which are \emph{not} B-KS capable form an abstract simplicial complex $\Sigma$ (i.e., a collection of sets that is closed under taking subsets). One may try to exploit this structure in order to better understand the sequences.

{\em Disjoint optimal B-KS sets.} In many cases, optimal B-KS sets $(S_A, S_B)$ have the feature that $S_A$ and $S_B$ are disjoint: this occurs for P-24, K-20, CK-37, CK-31, P-33, S-29 (and its extension), and S-7. A question then is to understand which features of the KS set lead to this behavior. An answer to this question would provide a useful heuristic for Algorithm~\ref{alg:bks}.

There is also a more strict notion of disjointness, insisting not only that $S_A$ and $S_B$ are disjoint, but that the set of bases $\{b : b \in S_A \cup S_B\}$ partitions the vectors of the associated KS set. The only such case we have identified is P-24. Such a set satisfies inequality~\eqref{eq:lowerbound} with equality, and since P-24 is optimal in dimension $4$, it would be fruitful to identify other KS sets satisfying this property.

{\em Extensions.} Certain extensions have smaller B-KS sets than the KS sets they contain. For example, optimal B-KSs of P-24 have $|S_A| = |S_B| = 3$, while CEG-18 (which is contained in P-24) has optimal B-KSs with $|S_A| = 5, |S_B| = 6$. On the other hand, other extensions do not affect the sizes of the optimal B-KSs (as in S-31 vs S-29, S-35 vs S-34). This is especially clear in the CK-37 case, where the optimum is better than in the CK-33 case, but is the same as in the CK-31 case. Let us note that the ``completion'' of a given KS set will always have optimal B-KS sets at least as good as any extension of the KS set (obtained by adding vectors that are involved in some orthogonal basis with the original vectors). However, this process produces KS sets of large size (i.e., with many orthogonal bases) for which the computation of optimal B-KS sets becomes difficult.
This begs the question: is there an efficient (i.e., not adding too many vectors and bases) way to extend a KS set in order to reduce the size of the optimal B-KS sets?


\section*{Acknowledgements}


This work was supported by the EU-funded project \href{10.3030/101070558}{FoQaCiA}. AC is also supported by the \href{10.13039/501100011033}{MCINN/AEI} (Project No.\ PID2020-113738GB-I00).



\begin{thebibliography}{74}%
	\makeatletter
	\providecommand \@ifxundefined [1]{%
		\@ifx{#1\undefined}
	}%
	\providecommand \@ifnum [1]{%
		\ifnum #1\expandafter \@firstoftwo
		\else \expandafter \@secondoftwo
		\fi
	}%
	\providecommand \@ifx [1]{%
		\ifx #1\expandafter \@firstoftwo
		\else \expandafter \@secondoftwo
		\fi
	}%
	\providecommand \natexlab [1]{#1}%
	\providecommand \enquote  [1]{``#1''}%
	\providecommand \bibnamefont  [1]{#1}%
	\providecommand \bibfnamefont [1]{#1}%
	\providecommand \citenamefont [1]{#1}%
	\providecommand \href@noop [0]{\@secondoftwo}%
	\providecommand \href [0]{\begingroup \@sanitize@url \@href}%
	\providecommand \@href[1]{\@@startlink{#1}\@@href}%
	\providecommand \@@href[1]{\endgroup#1\@@endlink}%
	\providecommand \@sanitize@url [0]{\catcode `\\12\catcode `\$12\catcode `\&12\catcode `\#12\catcode `\^12\catcode `\_12\catcode `\%12\relax}%
	\providecommand \@@startlink[1]{}%
	\providecommand \@@endlink[0]{}%
	\providecommand \url  [0]{\begingroup\@sanitize@url \@url }%
	\providecommand \@url [1]{\endgroup\@href {#1}{\urlprefix }}%
	\providecommand \urlprefix  [0]{URL }%
	\providecommand \Eprint [0]{\href }%
	\providecommand \doibase [0]{https://doi.org/}%
	\providecommand \selectlanguage [0]{\@gobble}%
	\providecommand \bibinfo  [0]{\@secondoftwo}%
	\providecommand \bibfield  [0]{\@secondoftwo}%
	\providecommand \translation [1]{[#1]}%
	\providecommand \BibitemOpen [0]{}%
	\providecommand \bibitemStop [0]{}%
	\providecommand \bibitemNoStop [0]{.\EOS\space}%
	\providecommand \EOS [0]{\spacefactor3000\relax}%
	\providecommand \BibitemShut  [1]{\csname bibitem#1\endcsname}%
	\let\auto@bib@innerbib\@empty
	\bibitem [{\citenamefont {Bell}(1964)}]{Bell:1964PHY}%
	\BibitemOpen
	\bibfield  {author} {\bibinfo {author} {\bibfnamefont {J.~S.}\ \bibnamefont {Bell}},\ }\bibfield  {title} {\bibinfo {title} {On the {E}instein {P}odolsky {R}osen paradox},\ }\href {https://doi.org/10.1103/PhysicsPhysiqueFizika.1.195} {\bibfield  {journal} {\bibinfo  {journal} {Physics}\ }\textbf {\bibinfo {volume} {1}},\ \bibinfo {pages} {195} (\bibinfo {year} {1964})}\BibitemShut {NoStop}%
	\bibitem [{\citenamefont {Clauser}\ \emph {et~al.}(1969)\citenamefont {Clauser}, \citenamefont {Horne}, \citenamefont {Shimony},\ and\ \citenamefont {Holt}}]{Clauser:1969PRL}%
	\BibitemOpen
	\bibfield  {author} {\bibinfo {author} {\bibfnamefont {J.~F.}\ \bibnamefont {Clauser}}, \bibinfo {author} {\bibfnamefont {M.~A.}\ \bibnamefont {Horne}}, \bibinfo {author} {\bibfnamefont {A.}~\bibnamefont {Shimony}},\ and\ \bibinfo {author} {\bibfnamefont {R.~A.}\ \bibnamefont {Holt}},\ }\bibfield  {title} {\bibinfo {title} {Proposed experiment to test local hidden-variable theories},\ }\href {https://doi.org/10.1103/PhysRevLett.23.880} {\bibfield  {journal} {\bibinfo  {journal} {Phys. Rev. Lett.}\ }\textbf {\bibinfo {volume} {23}},\ \bibinfo {pages} {880} (\bibinfo {year} {1969})}\BibitemShut {NoStop}%
	\bibitem [{\citenamefont {Cirel'son}(1980)}]{Tsirelson:1980LMP}%
	\BibitemOpen
	\bibfield  {author} {\bibinfo {author} {\bibfnamefont {B.~S.}\ \bibnamefont {Cirel'son}},\ }\bibfield  {title} {\bibinfo {title} {Quantum generalizations of {B}ell's inequality},\ }\href {https://doi.org/10.1007/BF00417500} {\bibfield  {journal} {\bibinfo  {journal} {Lett. Math. Phys.}\ }\textbf {\bibinfo {volume} {4}},\ \bibinfo {pages} {93} (\bibinfo {year} {1980})}\BibitemShut {NoStop}%
	\bibitem [{\citenamefont {Brassard}\ \emph {et~al.}(2005)\citenamefont {Brassard}, \citenamefont {Broadbent},\ and\ \citenamefont {Tapp}}]{GBT05}%
	\BibitemOpen
	\bibfield  {author} {\bibinfo {author} {\bibfnamefont {G.}~\bibnamefont {Brassard}}, \bibinfo {author} {\bibfnamefont {A.}~\bibnamefont {Broadbent}},\ and\ \bibinfo {author} {\bibfnamefont {A.}~\bibnamefont {Tapp}},\ }\bibfield  {title} {\bibinfo {title} {Quantum pseudo-telepathy},\ }\href {https://doi.org/10.1007/s10701-005-7353-4} {\bibfield  {journal} {\bibinfo  {journal} {Found. Phys.}\ }\textbf {\bibinfo {volume} {35}},\ \bibinfo {pages} {1877} (\bibinfo {year} {2005})}\BibitemShut {NoStop}%
	\bibitem [{\citenamefont {Renner}\ and\ \citenamefont {Wolf}(2004)}]{Renner2004b}%
	\BibitemOpen
	\bibfield  {author} {\bibinfo {author} {\bibfnamefont {R.}~\bibnamefont {Renner}}\ and\ \bibinfo {author} {\bibfnamefont {S.}~\bibnamefont {Wolf}},\ }\bibfield  {title} {\bibinfo {title} {Quantum pseudo-telepathy and the {K}ochen-{S}pecker theorem},\ }in\ \href {https://doi.org/10.1109/ISIT.2004.1365359} {\emph {\bibinfo {booktitle} {International Symposium on Information Theory, 2004. ISIT 2004. Proceedings}}}\ (\bibinfo  {publisher} {IEEE},\ \bibinfo {address} {Piscataway, NJ},\ \bibinfo {year} {2004})\ p.\ \bibinfo {pages} {322}\BibitemShut {NoStop}%
	\bibitem [{\citenamefont {Gisin}\ \emph {et~al.}(2007)\citenamefont {Gisin}, \citenamefont {M\'ethot},\ and\ \citenamefont {Scarani}}]{Gisin:2007IJQI}%
	\BibitemOpen
	\bibfield  {author} {\bibinfo {author} {\bibfnamefont {N.}~\bibnamefont {Gisin}}, \bibinfo {author} {\bibfnamefont {A.~A.}\ \bibnamefont {M\'ethot}},\ and\ \bibinfo {author} {\bibfnamefont {V.}~\bibnamefont {Scarani}},\ }\bibfield  {title} {\bibinfo {title} {Pseudo-telepathy: input cardinality and {B}ell-type inequalities},\ }\href {https://doi.org/10.1142/S021974990700289X} {\bibfield  {journal} {\bibinfo  {journal} {Int. J. Quantum Inform.}\ }\textbf {\bibinfo {volume} {5}},\ \bibinfo {pages} {525} (\bibinfo {year} {2007})}\BibitemShut {NoStop}%
	\bibitem [{\citenamefont {Cleve}\ \emph {et~al.}(2004)\citenamefont {Cleve}, \citenamefont {{H{\o}yer}}, \citenamefont {Toner},\ and\ \citenamefont {Watrous}}]{CHTW04}%
	\BibitemOpen
	\bibfield  {author} {\bibinfo {author} {\bibfnamefont {R.}~\bibnamefont {Cleve}}, \bibinfo {author} {\bibfnamefont {P.}~\bibnamefont {{H{\o}yer}}}, \bibinfo {author} {\bibfnamefont {B.}~\bibnamefont {Toner}},\ and\ \bibinfo {author} {\bibfnamefont {J.}~\bibnamefont {Watrous}},\ }\bibfield  {title} {\bibinfo {title} {Consequences and limits of nonlocal strategies},\ }in\ \href {https://doi.org/10.1109/CCC.2004.1313847} {\emph {\bibinfo {booktitle} {Proceedings. 19th IEEE Annual Conference on Computational Complexity}}}\ (\bibinfo {address} {Washington, DC},\ \bibinfo {year} {2004})\ pp.\ \bibinfo {pages} {236--249}\BibitemShut {NoStop}%
	\bibitem [{\citenamefont {Liu}\ \emph {et~al.}(2024)\citenamefont {Liu}, \citenamefont {Chung}, \citenamefont {Cruzeiro}, \citenamefont {Gonzales-Ureta}, \citenamefont {Ramanathan},\ and\ \citenamefont {Cabello}}]{Liu:2023PRR}%
	\BibitemOpen
	\bibfield  {author} {\bibinfo {author} {\bibfnamefont {Y.}~\bibnamefont {Liu}}, \bibinfo {author} {\bibfnamefont {H.~Y.}\ \bibnamefont {Chung}}, \bibinfo {author} {\bibfnamefont {E.~Z.}\ \bibnamefont {Cruzeiro}}, \bibinfo {author} {\bibfnamefont {J.~R.}\ \bibnamefont {Gonzales-Ureta}}, \bibinfo {author} {\bibfnamefont {R.}~\bibnamefont {Ramanathan}},\ and\ \bibinfo {author} {\bibfnamefont {A.}~\bibnamefont {Cabello}},\ }\bibfield  {title} {\bibinfo {title} {Equivalence between face nonsignaling correlations, full nonlocality, all-versus-nothing proofs, and pseudotelepathy},\ }\href {https://doi.org/10.1103/PhysRevResearch.6.L042035} {\bibfield  {journal} {\bibinfo  {journal} {Phys. Rev. Res.}\ }\textbf {\bibinfo {volume} {6}},\ \bibinfo {pages} {L042035} (\bibinfo {year} {2024})}\BibitemShut {NoStop}%
	\bibitem [{\citenamefont {Cabello}(2001{\natexlab{a}})}]{Cabello:2001PRLa}%
	\BibitemOpen
	\bibfield  {author} {\bibinfo {author} {\bibfnamefont {A.}~\bibnamefont {Cabello}},\ }\bibfield  {title} {\bibinfo {title} {Bell's theorem without inequalities and without probabilities for two observers},\ }\href {https://doi.org/10.1103/PhysRevLett.86.1911} {\bibfield  {journal} {\bibinfo  {journal} {Phys. Rev. Lett.}\ }\textbf {\bibinfo {volume} {86}},\ \bibinfo {pages} {1911} (\bibinfo {year} {2001}{\natexlab{a}})}\BibitemShut {NoStop}%
	\bibitem [{\citenamefont {Cabello}(2001{\natexlab{b}})}]{Cabello:2001PRLb}%
	\BibitemOpen
	\bibfield  {author} {\bibinfo {author} {\bibfnamefont {A.}~\bibnamefont {Cabello}},\ }\bibfield  {title} {\bibinfo {title} {``{A}ll versus {N}othing'' {I}nseparability for {T}wo {O}bservers},\ }\href {https://doi.org/10.1103/PhysRevLett.87.010403} {\bibfield  {journal} {\bibinfo  {journal} {Phys. Rev. Lett.}\ }\textbf {\bibinfo {volume} {87}},\ \bibinfo {pages} {010403} (\bibinfo {year} {2001}{\natexlab{b}})}\BibitemShut {NoStop}%
	\bibitem [{\citenamefont {Aravind}(2004)}]{Aravind:2004AJP}%
	\BibitemOpen
	\bibfield  {author} {\bibinfo {author} {\bibfnamefont {P.~K.}\ \bibnamefont {Aravind}},\ }\bibfield  {title} {\bibinfo {title} {Quantum mysteries revisited again},\ }\href {https://doi.org/10.1119/1.1773173} {\bibfield  {journal} {\bibinfo  {journal} {Am. J. Phys.}\ }\textbf {\bibinfo {volume} {72}},\ \bibinfo {pages} {1303} (\bibinfo {year} {2004})}\BibitemShut {NoStop}%
	\bibitem [{\citenamefont {Cabello}(2025)}]{Cabello:2023XXX}%
	\BibitemOpen
	\bibfield  {author} {\bibinfo {author} {\bibfnamefont {A.}~\bibnamefont {Cabello}},\ }\bibfield  {title} {\bibinfo {title} {Simplest bipartite perfect quantum strategies},\ }\href {https://doi.org/10.1103/PhysRevLett.134.010201} {\bibfield  {journal} {\bibinfo  {journal} {Phys. Rev. Lett.}\ }\textbf {\bibinfo {volume} {134}},\ \bibinfo {pages} {010201} (\bibinfo {year} {2025})}\BibitemShut {NoStop}%
	\bibitem [{\citenamefont {Cinelli}\ \emph {et~al.}(2005)\citenamefont {Cinelli}, \citenamefont {Barbieri}, \citenamefont {Perris}, \citenamefont {Mataloni},\ and\ \citenamefont {De~Martini}}]{CinelliPRL2005}%
	\BibitemOpen
	\bibfield  {author} {\bibinfo {author} {\bibfnamefont {C.}~\bibnamefont {Cinelli}}, \bibinfo {author} {\bibfnamefont {M.}~\bibnamefont {Barbieri}}, \bibinfo {author} {\bibfnamefont {R.}~\bibnamefont {Perris}}, \bibinfo {author} {\bibfnamefont {P.}~\bibnamefont {Mataloni}},\ and\ \bibinfo {author} {\bibfnamefont {F.}~\bibnamefont {De~Martini}},\ }\bibfield  {title} {\bibinfo {title} {{A}ll-{V}ersus-{N}othing {N}onlocality {T}est of {Q}uantum {M}echanics by {T}wo-{P}hoton {H}yperentanglement},\ }\href {https://doi.org/10.1103/PhysRevLett.95.240405} {\bibfield  {journal} {\bibinfo  {journal} {Phys. Rev. Lett.}\ }\textbf {\bibinfo {volume} {95}},\ \bibinfo {pages} {240405} (\bibinfo {year} {2005})}\BibitemShut {NoStop}%
	\bibitem [{\citenamefont {Yang}\ \emph {et~al.}(2005)\citenamefont {Yang}, \citenamefont {Zhang}, \citenamefont {Zhang}, \citenamefont {Yin}, \citenamefont {Zhao}, \citenamefont {{\.Z}ukowski}, \citenamefont {Chen},\ and\ \citenamefont {Pan}}]{YangPRL2005}%
	\BibitemOpen
	\bibfield  {author} {\bibinfo {author} {\bibfnamefont {T.}~\bibnamefont {Yang}}, \bibinfo {author} {\bibfnamefont {Q.}~\bibnamefont {Zhang}}, \bibinfo {author} {\bibfnamefont {J.}~\bibnamefont {Zhang}}, \bibinfo {author} {\bibfnamefont {J.}~\bibnamefont {Yin}}, \bibinfo {author} {\bibfnamefont {Z.}~\bibnamefont {Zhao}}, \bibinfo {author} {\bibfnamefont {M.}~\bibnamefont {{\.Z}ukowski}}, \bibinfo {author} {\bibfnamefont {Z.-B.}\ \bibnamefont {Chen}},\ and\ \bibinfo {author} {\bibfnamefont {J.-W.}\ \bibnamefont {Pan}},\ }\bibfield  {title} {\bibinfo {title} {{A}ll-{V}ersus-{N}othing {V}iolation of {L}ocal {R}ealism by {T}wo-{P}hoton, {F}our-{D}imensional {E}ntanglement},\ }\href {https://doi.org/10.1103/PhysRevLett.95.240406} {\bibfield  {journal} {\bibinfo  {journal} {Phys. Rev. Lett.}\ }\textbf {\bibinfo {volume} {95}},\ \bibinfo {pages} {240406} (\bibinfo {year} {2005})}\BibitemShut {NoStop}%
	\bibitem [{\citenamefont {Barbieri}\ \emph {et~al.}(2007)\citenamefont {Barbieri}, \citenamefont {Vallone}, \citenamefont {De~Martini},\ and\ \citenamefont {Mataloni}}]{BarbieriOS2007}%
	\BibitemOpen
	\bibfield  {author} {\bibinfo {author} {\bibfnamefont {M.}~\bibnamefont {Barbieri}}, \bibinfo {author} {\bibfnamefont {G.}~\bibnamefont {Vallone}}, \bibinfo {author} {\bibfnamefont {F.}~\bibnamefont {De~Martini}},\ and\ \bibinfo {author} {\bibfnamefont {P.}~\bibnamefont {Mataloni}},\ }\bibfield  {title} {\bibinfo {title} {Polarization-momentum hyper-entangled two photon states},\ }\href {https://doi.org/10.1134/S0030400X0707020X} {\bibfield  {journal} {\bibinfo  {journal} {Opt. Spectrosc.}\ }\textbf {\bibinfo {volume} {103}},\ \bibinfo {pages} {129} (\bibinfo {year} {2007})}\BibitemShut {NoStop}%
	\bibitem [{\citenamefont {Aolita}\ \emph {et~al.}(2012)\citenamefont {Aolita}, \citenamefont {Gallego}, \citenamefont {Ac{\'\i}n}, \citenamefont {Chiuri}, \citenamefont {Vallone}, \citenamefont {Mataloni},\ and\ \citenamefont {Cabello}}]{Aolita:2012PRA}%
	\BibitemOpen
	\bibfield  {author} {\bibinfo {author} {\bibfnamefont {L.}~\bibnamefont {Aolita}}, \bibinfo {author} {\bibfnamefont {R.}~\bibnamefont {Gallego}}, \bibinfo {author} {\bibfnamefont {A.}~\bibnamefont {Ac{\'\i}n}}, \bibinfo {author} {\bibfnamefont {A.}~\bibnamefont {Chiuri}}, \bibinfo {author} {\bibfnamefont {G.}~\bibnamefont {Vallone}}, \bibinfo {author} {\bibfnamefont {P.}~\bibnamefont {Mataloni}},\ and\ \bibinfo {author} {\bibfnamefont {A.}~\bibnamefont {Cabello}},\ }\bibfield  {title} {\bibinfo {title} {Fully nonlocal quantum correlations},\ }\href {https://doi.org/10.1103/PhysRevA.85.032107} {\bibfield  {journal} {\bibinfo  {journal} {Phys. Rev. A}\ }\textbf {\bibinfo {volume} {85}},\ \bibinfo {pages} {032107} (\bibinfo {year} {2012})}\BibitemShut {NoStop}%
	\bibitem [{\citenamefont {Xu}\ \emph {et~al.}(2022)\citenamefont {Xu}, \citenamefont {Zhen}, \citenamefont {Yang}, \citenamefont {Cheng}, \citenamefont {Ren}, \citenamefont {Chen}, \citenamefont {Wang},\ and\ \citenamefont {Wang}}]{Xu:2022PRL}%
	\BibitemOpen
	\bibfield  {author} {\bibinfo {author} {\bibfnamefont {J.-M.}\ \bibnamefont {Xu}}, \bibinfo {author} {\bibfnamefont {Y.-Z.}\ \bibnamefont {Zhen}}, \bibinfo {author} {\bibfnamefont {Y.-X.}\ \bibnamefont {Yang}}, \bibinfo {author} {\bibfnamefont {Z.-M.}\ \bibnamefont {Cheng}}, \bibinfo {author} {\bibfnamefont {Z.-C.}\ \bibnamefont {Ren}}, \bibinfo {author} {\bibfnamefont {K.}~\bibnamefont {Chen}}, \bibinfo {author} {\bibfnamefont {X.-L.}\ \bibnamefont {Wang}},\ and\ \bibinfo {author} {\bibfnamefont {H.-T.}\ \bibnamefont {Wang}},\ }\bibfield  {title} {\bibinfo {title} {Experimental {D}emonstration of {Q}uantum {P}seudotelepathy},\ }\href {https://doi.org/10.1103/PhysRevLett.129.050402} {\bibfield  {journal} {\bibinfo  {journal} {Phys. Rev. Lett.}\ }\textbf {\bibinfo {volume} {129}},\ \bibinfo {pages} {050402} (\bibinfo {year} {2022})}\BibitemShut {NoStop}%
	\bibitem [{\citenamefont {Elitzur}\ \emph {et~al.}(1992)\citenamefont {Elitzur}, \citenamefont {Popescu},\ and\ \citenamefont {Rohrlich}}]{Elitzur:1992PLA}%
	\BibitemOpen
	\bibfield  {author} {\bibinfo {author} {\bibfnamefont {A.~C.}\ \bibnamefont {Elitzur}}, \bibinfo {author} {\bibfnamefont {S.}~\bibnamefont {Popescu}},\ and\ \bibinfo {author} {\bibfnamefont {D.}~\bibnamefont {Rohrlich}},\ }\bibfield  {title} {\bibinfo {title} {Quantum nonlocality for each pair in an ensemble},\ }\href {https://doi.org/10.1016/0375-9601(92)90952-I} {\bibfield  {journal} {\bibinfo  {journal} {Phys. Lett. A}\ }\textbf {\bibinfo {volume} {162}},\ \bibinfo {pages} {25} (\bibinfo {year} {1992})}\BibitemShut {NoStop}%
	\bibitem [{\citenamefont {Goh}\ \emph {et~al.}(2018)\citenamefont {Goh}, \citenamefont {Kaniewski}, \citenamefont {Wolfe}, \citenamefont {V\'ertesi}, \citenamefont {Wu}, \citenamefont {Cai}, \citenamefont {Liang},\ and\ \citenamefont {Scarani}}]{Goh:2018PRA}%
	\BibitemOpen
	\bibfield  {author} {\bibinfo {author} {\bibfnamefont {K.~T.}\ \bibnamefont {Goh}}, \bibinfo {author} {\bibfnamefont {J.}~\bibnamefont {Kaniewski}}, \bibinfo {author} {\bibfnamefont {E.}~\bibnamefont {Wolfe}}, \bibinfo {author} {\bibfnamefont {T.}~\bibnamefont {V\'ertesi}}, \bibinfo {author} {\bibfnamefont {X.}~\bibnamefont {Wu}}, \bibinfo {author} {\bibfnamefont {Y.}~\bibnamefont {Cai}}, \bibinfo {author} {\bibfnamefont {Y.-C.}\ \bibnamefont {Liang}},\ and\ \bibinfo {author} {\bibfnamefont {V.}~\bibnamefont {Scarani}},\ }\bibfield  {title} {\bibinfo {title} {Geometry of the set of quantum correlations},\ }\href {https://doi.org/10.1103/PhysRevA.97.022104} {\bibfield  {journal} {\bibinfo  {journal} {Phys. Rev. A}\ }\textbf {\bibinfo {volume} {97}},\ \bibinfo {pages} {022104} (\bibinfo {year} {2018})}\BibitemShut {NoStop}%
	\bibitem [{\citenamefont {Greenberger}\ \emph {et~al.}(1989)\citenamefont {Greenberger}, \citenamefont {Horne},\ and\ \citenamefont {Zeilinger}}]{GHZ89}%
	\BibitemOpen
	\bibfield  {author} {\bibinfo {author} {\bibfnamefont {D.~M.}\ \bibnamefont {Greenberger}}, \bibinfo {author} {\bibfnamefont {M.~A.}\ \bibnamefont {Horne}},\ and\ \bibinfo {author} {\bibfnamefont {A.}~\bibnamefont {Zeilinger}},\ }\bibinfo {title} {Going beyond {B}ell's theorem},\ in\ \href {https://doi.org/10.1007/978-94-017-0849-4} {\emph {\bibinfo {booktitle} {Bell's Theorem, Quantum Theory and Conceptions of the Universe}}},\ \bibinfo {series} {Fundamental Theories of Physics}, Vol.~\bibinfo {volume} {37},\ \bibinfo {editor} {edited by\ \bibinfo {editor} {\bibfnamefont {M.}~\bibnamefont {Kafatos}}}\ (\bibinfo  {publisher} {Springer},\ \bibinfo {address} {Dordrecht, The Netherlands},\ \bibinfo {year} {1989})\ pp.\ \bibinfo {pages} {69--72}\BibitemShut {NoStop}%
	\bibitem [{\citenamefont {Shimony}(1993)}]{Shimony:1993}%
	\BibitemOpen
	\bibfield  {author} {\bibinfo {author} {\bibfnamefont {A.}~\bibnamefont {Shimony}},\ }\href {https://doi.org/10.1017/CBO9781139172196} {\emph {\bibinfo {title} {Search for a Naturalistic World View. {V}olume {I}{I}: Natural Science and Methaphysics}}}\ (\bibinfo  {publisher} {Cambridge University Press},\ \bibinfo {address} {Cambridge, UK},\ \bibinfo {year} {1993})\BibitemShut {NoStop}%
	\bibitem [{\citenamefont {Vieira}\ \emph {et~al.}(2024)\citenamefont {Vieira}, \citenamefont {Ramanathan},\ and\ \citenamefont {Cabello}}]{Vieira;2024XXX}%
	\BibitemOpen
	\bibfield  {author} {\bibinfo {author} {\bibfnamefont {C.}~\bibnamefont {Vieira}}, \bibinfo {author} {\bibfnamefont {R.}~\bibnamefont {Ramanathan}},\ and\ \bibinfo {author} {\bibfnamefont {A.}~\bibnamefont {Cabello}},\ }\href@noop {} {\bibinfo {title} {Test of the physical significance of {B}ell nonlocality}} (\bibinfo {year} {2024}),\ \Eprint {https://arxiv.org/abs/2402.00801} {arXiv:2402.00801 [quant-ph]} \BibitemShut {NoStop}%
	\bibitem [{\citenamefont {Ji}\ \emph {et~al.}(2021)\citenamefont {Ji}, \citenamefont {Natarajan}, \citenamefont {Vidick}, \citenamefont {Wright},\ and\ \citenamefont {Yuen}}]{Ji:2021CACM}%
	\BibitemOpen
	\bibfield  {author} {\bibinfo {author} {\bibfnamefont {Z.}~\bibnamefont {Ji}}, \bibinfo {author} {\bibfnamefont {A.}~\bibnamefont {Natarajan}}, \bibinfo {author} {\bibfnamefont {T.}~\bibnamefont {Vidick}}, \bibinfo {author} {\bibfnamefont {J.}~\bibnamefont {Wright}},\ and\ \bibinfo {author} {\bibfnamefont {H.}~\bibnamefont {Yuen}},\ }\bibfield  {title} {\bibinfo {title} {M{I}{P}*={R}{E}},\ }\href {https://doi.org/10.1145/3485628} {\bibfield  {journal} {\bibinfo  {journal} {Comm. ACM}\ }\textbf {\bibinfo {volume} {64}},\ \bibinfo {pages} {131} (\bibinfo {year} {2021})}\BibitemShut {NoStop}%
	\bibitem [{\citenamefont {Cabello}\ \emph {et~al.}(2023)\citenamefont {Cabello}, \citenamefont {Quintino},\ and\ \citenamefont {Kleinmann}}]{cabello2023logicalpossibilitiesphysicsmipre}%
	\BibitemOpen
	\bibfield  {author} {\bibinfo {author} {\bibfnamefont {A.}~\bibnamefont {Cabello}}, \bibinfo {author} {\bibfnamefont {M.~T.}\ \bibnamefont {Quintino}},\ and\ \bibinfo {author} {\bibfnamefont {M.}~\bibnamefont {Kleinmann}},\ }\href@noop {} {\bibinfo {title} {Logical possibilities for physics after {M}{I}{P}*={R}{E}}} (\bibinfo {year} {2023}),\ \Eprint {https://arxiv.org/abs/2307.02920} {arXiv:2307.02920 [quant-ph]} \BibitemShut {NoStop}%
	\bibitem [{\citenamefont {Cabello}(2015)}]{Cabello:2015PRL}%
	\BibitemOpen
	\bibfield  {author} {\bibinfo {author} {\bibfnamefont {A.}~\bibnamefont {Cabello}},\ }\bibfield  {title} {\bibinfo {title} {{S}imple {E}xplanation of the {Q}uantum {L}imits of {G}enuine {$n$}-{B}ody {N}onlocality},\ }\href {https://doi.org/10.1103/PhysRevLett.114.220402} {\bibfield  {journal} {\bibinfo  {journal} {Phys. Rev. Lett.}\ }\textbf {\bibinfo {volume} {114}},\ \bibinfo {pages} {220402} (\bibinfo {year} {2015})}\BibitemShut {NoStop}%
	\bibitem [{\citenamefont {Cabello}(2017)}]{Cabello:2017}%
	\BibitemOpen
	\bibfield  {author} {\bibinfo {author} {\bibfnamefont {A.}~\bibnamefont {Cabello}},\ }\bibinfo {title} {The {U}nspeakable {W}hy},\ in\ \href {https://doi.org/10.1007/978-3-319-38987-5_11} {\emph {\bibinfo {booktitle} {Quantum [Un]Speakables II: Half a Century of {B}ell's Theorem}}},\ \bibinfo {editor} {edited by\ \bibinfo {editor} {\bibfnamefont {R.}~\bibnamefont {Bertlmann}}\ and\ \bibinfo {editor} {\bibfnamefont {A.}~\bibnamefont {Zeilinger}}}\ (\bibinfo  {publisher} {Springer International Publishing},\ \bibinfo {address} {Cham},\ \bibinfo {year} {2017})\ pp.\ \bibinfo {pages} {189--199}\BibitemShut {NoStop}%
	\bibitem [{\citenamefont {Cabello}(2019)}]{Cabello:2019PRA}%
	\BibitemOpen
	\bibfield  {author} {\bibinfo {author} {\bibfnamefont {A.}~\bibnamefont {Cabello}},\ }\bibfield  {title} {\bibinfo {title} {Quantum correlations from simple assumptions},\ }\href {https://doi.org/10.1103/PhysRevA.100.032120} {\bibfield  {journal} {\bibinfo  {journal} {Phys. Rev. A}\ }\textbf {\bibinfo {volume} {100}},\ \bibinfo {pages} {032120} (\bibinfo {year} {2019})}\BibitemShut {NoStop}%
	\bibitem [{\citenamefont {Bravyi}\ \emph {et~al.}(2018)\citenamefont {Bravyi}, \citenamefont {Gosset},\ and\ \citenamefont {K{\"o}nig}}]{Bravyi:2018SCI}%
	\BibitemOpen
	\bibfield  {author} {\bibinfo {author} {\bibfnamefont {S.}~\bibnamefont {Bravyi}}, \bibinfo {author} {\bibfnamefont {D.}~\bibnamefont {Gosset}},\ and\ \bibinfo {author} {\bibfnamefont {R.}~\bibnamefont {K{\"o}nig}},\ }\bibfield  {title} {\bibinfo {title} {Quantum advantage with shallow circuits},\ }\href {https://doi.org/10.1126/science.aar3106} {\bibfield  {journal} {\bibinfo  {journal} {Science}\ }\textbf {\bibinfo {volume} {362}},\ \bibinfo {pages} {308} (\bibinfo {year} {2018})}\BibitemShut {NoStop}%
	\bibitem [{\citenamefont {Cubitt}\ \emph {et~al.}(2010)\citenamefont {Cubitt}, \citenamefont {Leung}, \citenamefont {Matthews},\ and\ \citenamefont {Winter}}]{Cubitt:2010PRL}%
	\BibitemOpen
	\bibfield  {author} {\bibinfo {author} {\bibfnamefont {T.~S.}\ \bibnamefont {Cubitt}}, \bibinfo {author} {\bibfnamefont {D.}~\bibnamefont {Leung}}, \bibinfo {author} {\bibfnamefont {W.}~\bibnamefont {Matthews}},\ and\ \bibinfo {author} {\bibfnamefont {A.}~\bibnamefont {Winter}},\ }\bibfield  {title} {\bibinfo {title} {Improving {Z}ero-{E}rror {C}lassical {C}ommunication with {E}ntanglement},\ }\href {https://doi.org/10.1103/PhysRevLett.104.230503} {\bibfield  {journal} {\bibinfo  {journal} {Phys. Rev. Lett.}\ }\textbf {\bibinfo {volume} {104}},\ \bibinfo {pages} {230503} (\bibinfo {year} {2010})}\BibitemShut {NoStop}%
	\bibitem [{\citenamefont {Horodecki}\ \emph {et~al.}(2010)\citenamefont {Horodecki}, \citenamefont {Horodecki}, \citenamefont {Horodecki}, \citenamefont {Horodecki}, \citenamefont {Pawlowski},\ and\ \citenamefont {Bourennane}}]{Horodecki:2010XXX}%
	\BibitemOpen
	\bibfield  {author} {\bibinfo {author} {\bibfnamefont {K.}~\bibnamefont {Horodecki}}, \bibinfo {author} {\bibfnamefont {M.}~\bibnamefont {Horodecki}}, \bibinfo {author} {\bibfnamefont {P.}~\bibnamefont {Horodecki}}, \bibinfo {author} {\bibfnamefont {R.}~\bibnamefont {Horodecki}}, \bibinfo {author} {\bibfnamefont {M.}~\bibnamefont {Pawlowski}},\ and\ \bibinfo {author} {\bibfnamefont {M.}~\bibnamefont {Bourennane}},\ }\href@noop {} {\bibinfo {title} {Contextuality offers device-independent security}} (\bibinfo {year} {2010}),\ \Eprint {https://arxiv.org/abs/1006.0468} {arXiv:1006.0468 [quant-ph]} \BibitemShut {NoStop}%
	\bibitem [{\citenamefont {Vidick}(2017)}]{Vidick:2017XXX}%
	\BibitemOpen
	\bibfield  {author} {\bibinfo {author} {\bibfnamefont {T.}~\bibnamefont {Vidick}},\ }\href@noop {} {\bibinfo {title} {Parallel {D}{I}{Q}{K}{D} from parallel repetition}} (\bibinfo {year} {2017}),\ \Eprint {https://arxiv.org/abs/1703.08508} {arXiv:1703.08508 [quant-ph]} \BibitemShut {NoStop}%
	\bibitem [{\citenamefont {Jain}\ \emph {et~al.}(2020)\citenamefont {Jain}, \citenamefont {Miller},\ and\ \citenamefont {Shi}}]{Jain:2020IEE}%
	\BibitemOpen
	\bibfield  {author} {\bibinfo {author} {\bibfnamefont {R.}~\bibnamefont {Jain}}, \bibinfo {author} {\bibfnamefont {C.~A.}\ \bibnamefont {Miller}},\ and\ \bibinfo {author} {\bibfnamefont {Y.}~\bibnamefont {Shi}},\ }\bibfield  {title} {\bibinfo {title} {Parallel device-independent quantum key distribution},\ }\href {https://doi.org/10.1109/TIT.2020.2986740} {\bibfield  {journal} {\bibinfo  {journal} {IEEE Trans. Inf. Theory}\ }\textbf {\bibinfo {volume} {66}},\ \bibinfo {pages} {5567} (\bibinfo {year} {2020})}\BibitemShut {NoStop}%
	\bibitem [{\citenamefont {Zhen}\ \emph {et~al.}(2023)\citenamefont {Zhen}, \citenamefont {Mao}, \citenamefont {Zhang}, \citenamefont {Xu},\ and\ \citenamefont {Sanders}}]{Zhen:2023PRL}%
	\BibitemOpen
	\bibfield  {author} {\bibinfo {author} {\bibfnamefont {Y.-Z.}\ \bibnamefont {Zhen}}, \bibinfo {author} {\bibfnamefont {Y.}~\bibnamefont {Mao}}, \bibinfo {author} {\bibfnamefont {Y.-Z.}\ \bibnamefont {Zhang}}, \bibinfo {author} {\bibfnamefont {F.}~\bibnamefont {Xu}},\ and\ \bibinfo {author} {\bibfnamefont {B.~C.}\ \bibnamefont {Sanders}},\ }\bibfield  {title} {\bibinfo {title} {Device-{I}ndependent {Q}uantum {K}ey {D}istribution {B}ased on the {M}ermin-{P}eres {M}agic {S}quare {G}ame},\ }\href {https://doi.org/10.1103/PhysRevLett.131.080801} {\bibfield  {journal} {\bibinfo  {journal} {Phys. Rev. Lett.}\ }\textbf {\bibinfo {volume} {131}},\ \bibinfo {pages} {080801} (\bibinfo {year} {2023})}\BibitemShut {NoStop}%
	\bibitem [{\citenamefont {Bharti}\ and\ \citenamefont {Jain}(2023)}]{Bharti:2023XXX}%
	\BibitemOpen
	\bibfield  {author} {\bibinfo {author} {\bibfnamefont {K.}~\bibnamefont {Bharti}}\ and\ \bibinfo {author} {\bibfnamefont {R.}~\bibnamefont {Jain}},\ }\href@noop {} {\bibinfo {title} {On the power of geometrically-local classical and quantum circuits}} (\bibinfo {year} {2023}),\ \Eprint {https://arxiv.org/abs/2310.01540} {arXiv:2310.01540 [quant-ph]} \BibitemShut {NoStop}%
	\bibitem [{\citenamefont {Kochen}\ and\ \citenamefont {Specker}(1967)}]{Kochen:1967JMM}%
	\BibitemOpen
	\bibfield  {author} {\bibinfo {author} {\bibfnamefont {S.}~\bibnamefont {Kochen}}\ and\ \bibinfo {author} {\bibfnamefont {E.~P.}\ \bibnamefont {Specker}},\ }\bibfield  {title} {\bibinfo {title} {The {P}roblem of {H}idden {V}ariables in {Q}uantum {M}echanics},\ }\href {https://doi.org/10.1512/iumj.1968.17.17004} {\bibfield  {journal} {\bibinfo  {journal} {J. Math. Mech.}\ }\textbf {\bibinfo {volume} {17}},\ \bibinfo {pages} {59} (\bibinfo {year} {1967})}\BibitemShut {NoStop}%
	\bibitem [{\citenamefont {Pavi\v{c}i{\'c}}\ \emph {et~al.}(2005)\citenamefont {Pavi\v{c}i{\'c}}, \citenamefont {Merlet}, \citenamefont {McKay},\ and\ \citenamefont {Megill}}]{Pavicic:2005JPA}%
	\BibitemOpen
	\bibfield  {author} {\bibinfo {author} {\bibfnamefont {M.}~\bibnamefont {Pavi\v{c}i{\'c}}}, \bibinfo {author} {\bibfnamefont {J.-P.}\ \bibnamefont {Merlet}}, \bibinfo {author} {\bibfnamefont {B.~D.}\ \bibnamefont {McKay}},\ and\ \bibinfo {author} {\bibfnamefont {N.~D.}\ \bibnamefont {Megill}},\ }\bibfield  {title} {\bibinfo {title} {{K}ochen-{S}pecker vectors},\ }\href {https://doi.org/10.1088/0305-4470/38/7/013} {\bibfield  {journal} {\bibinfo  {journal} {J. Phys. A: Math. Gen.}\ }\textbf {\bibinfo {volume} {38}},\ \bibinfo {pages} {1577} (\bibinfo {year} {2005})}\BibitemShut {NoStop}%
	\bibitem [{\citenamefont {Mančinska}\ \emph {et~al.}(2013)\citenamefont {Mančinska}, \citenamefont {Scarpa},\ and\ \citenamefont {Severini}}]{MancinskaScarpaSeverini:2013}%
	\BibitemOpen
	\bibfield  {author} {\bibinfo {author} {\bibfnamefont {L.}~\bibnamefont {Mančinska}}, \bibinfo {author} {\bibfnamefont {G.}~\bibnamefont {Scarpa}},\ and\ \bibinfo {author} {\bibfnamefont {S.}~\bibnamefont {Severini}},\ }\bibfield  {title} {\bibinfo {title} {New separations in zero-error channel capacity through projective {K}ochen–{S}pecker sets and quantum coloring},\ }\href {https://doi.org/10.1109/TIT.2013.2248031} {\bibfield  {journal} {\bibinfo  {journal} {IEEE Trans. Inf. Theory.}\ }\textbf {\bibinfo {volume} {59}},\ \bibinfo {pages} {4025} (\bibinfo {year} {2013})}\BibitemShut {NoStop}%
	\bibitem [{\citenamefont {Cabello}(2008)}]{Cabello:2008PRL}%
	\BibitemOpen
	\bibfield  {author} {\bibinfo {author} {\bibfnamefont {A.}~\bibnamefont {Cabello}},\ }\bibfield  {title} {\bibinfo {title} {Experimentally testable state-independent quantum contextuality},\ }\href {https://doi.org/10.1103/PhysRevLett.101.210401} {\bibfield  {journal} {\bibinfo  {journal} {Phys. Rev. Lett.}\ }\textbf {\bibinfo {volume} {101}},\ \bibinfo {pages} {210401} (\bibinfo {year} {2008})}\BibitemShut {NoStop}%
	\bibitem [{\citenamefont {Badzia\c{g}}\ \emph {et~al.}(2009)\citenamefont {Badzia\c{g}}, \citenamefont {Bengtsson}, \citenamefont {Cabello},\ and\ \citenamefont {Pitowsky}}]{Badziag:2009PRL}%
	\BibitemOpen
	\bibfield  {author} {\bibinfo {author} {\bibfnamefont {P.}~\bibnamefont {Badzia\c{g}}}, \bibinfo {author} {\bibfnamefont {I.}~\bibnamefont {Bengtsson}}, \bibinfo {author} {\bibfnamefont {A.}~\bibnamefont {Cabello}},\ and\ \bibinfo {author} {\bibfnamefont {I.}~\bibnamefont {Pitowsky}},\ }\bibfield  {title} {\bibinfo {title} {Universality of state-independent violation of correlation inequalities for noncontextual theories},\ }\href {https://doi.org/10.1103/PhysRevLett.103.050401} {\bibfield  {journal} {\bibinfo  {journal} {Phys. Rev. Lett.}\ }\textbf {\bibinfo {volume} {103}},\ \bibinfo {pages} {050401} (\bibinfo {year} {2009})}\BibitemShut {NoStop}%
	\bibitem [{\citenamefont {Kirchmair}\ \emph {et~al.}(2009)\citenamefont {Kirchmair}, \citenamefont {Z{\"a}hringer}, \citenamefont {Gerritsma}, \citenamefont {Kleinmann}, \citenamefont {G{\"u}hne}, \citenamefont {Cabello}, \citenamefont {Blatt},\ and\ \citenamefont {Roos}}]{Kirchmair:2009NAT}%
	\BibitemOpen
	\bibfield  {author} {\bibinfo {author} {\bibfnamefont {G.}~\bibnamefont {Kirchmair}}, \bibinfo {author} {\bibfnamefont {F.}~\bibnamefont {Z{\"a}hringer}}, \bibinfo {author} {\bibfnamefont {R.}~\bibnamefont {Gerritsma}}, \bibinfo {author} {\bibfnamefont {M.}~\bibnamefont {Kleinmann}}, \bibinfo {author} {\bibfnamefont {O.}~\bibnamefont {G{\"u}hne}}, \bibinfo {author} {\bibfnamefont {A.}~\bibnamefont {Cabello}}, \bibinfo {author} {\bibfnamefont {R.}~\bibnamefont {Blatt}},\ and\ \bibinfo {author} {\bibfnamefont {C.~F.}\ \bibnamefont {Roos}},\ }\bibfield  {title} {\bibinfo {title} {State-independent experimental test of quantum contextuality},\ }\href {https://doi.org/10.1038/nature08172} {\bibfield  {journal} {\bibinfo  {journal} {Nature}\ }\textbf {\bibinfo {volume} {460}},\ \bibinfo {pages} {494} (\bibinfo {year} {2009})}\BibitemShut {NoStop}%
	\bibitem [{\citenamefont {Amselem}\ \emph {et~al.}(2009)\citenamefont {Amselem}, \citenamefont {R{\aa}dmark}, \citenamefont {Bourennane},\ and\ \citenamefont {Cabello}}]{Amselem:2009PRL}%
	\BibitemOpen
	\bibfield  {author} {\bibinfo {author} {\bibfnamefont {E.}~\bibnamefont {Amselem}}, \bibinfo {author} {\bibfnamefont {M.}~\bibnamefont {R{\aa}dmark}}, \bibinfo {author} {\bibfnamefont {M.}~\bibnamefont {Bourennane}},\ and\ \bibinfo {author} {\bibfnamefont {A.}~\bibnamefont {Cabello}},\ }\bibfield  {title} {\bibinfo {title} {State-independent quantum contextuality with single photons},\ }\href {https://doi.org/10.1103/PhysRevLett.103.160405} {\bibfield  {journal} {\bibinfo  {journal} {Phys. Rev. Lett.}\ }\textbf {\bibinfo {volume} {103}},\ \bibinfo {pages} {160405} (\bibinfo {year} {2009})}\BibitemShut {NoStop}%
	\bibitem [{\citenamefont {D'Ambrosio}\ \emph {et~al.}(2013)\citenamefont {D'Ambrosio}, \citenamefont {Herbauts}, \citenamefont {Amselem}, \citenamefont {Nagali}, \citenamefont {Bourennane}, \citenamefont {Sciarrino},\ and\ \citenamefont {Cabello}}]{D'Ambrosio:2013PRX}%
	\BibitemOpen
	\bibfield  {author} {\bibinfo {author} {\bibfnamefont {V.}~\bibnamefont {D'Ambrosio}}, \bibinfo {author} {\bibfnamefont {I.}~\bibnamefont {Herbauts}}, \bibinfo {author} {\bibfnamefont {E.}~\bibnamefont {Amselem}}, \bibinfo {author} {\bibfnamefont {E.}~\bibnamefont {Nagali}}, \bibinfo {author} {\bibfnamefont {M.}~\bibnamefont {Bourennane}}, \bibinfo {author} {\bibfnamefont {F.}~\bibnamefont {Sciarrino}},\ and\ \bibinfo {author} {\bibfnamefont {A.}~\bibnamefont {Cabello}},\ }\bibfield  {title} {\bibinfo {title} {Experimental implementation of a {K}ochen-{S}pecker set of quantum tests},\ }\href {https://doi.org/10.1103/PhysRevX.3.011012} {\bibfield  {journal} {\bibinfo  {journal} {Phys. Rev. X}\ }\textbf {\bibinfo {volume} {3}},\ \bibinfo {pages} {011012} (\bibinfo {year} {2013})}\BibitemShut {NoStop}%
	\bibitem [{\citenamefont {Xu}\ \emph {et~al.}(2024)\citenamefont {Xu}, \citenamefont {Saha}, \citenamefont {Bharti},\ and\ \citenamefont {Cabello}}]{xu2023stateindependent}%
	\BibitemOpen
	\bibfield  {author} {\bibinfo {author} {\bibfnamefont {Z.-P.}\ \bibnamefont {Xu}}, \bibinfo {author} {\bibfnamefont {D.}~\bibnamefont {Saha}}, \bibinfo {author} {\bibfnamefont {K.}~\bibnamefont {Bharti}},\ and\ \bibinfo {author} {\bibfnamefont {A.}~\bibnamefont {Cabello}},\ }\bibfield  {title} {\bibinfo {title} {Certifying sets of quantum observables with any full-rank state},\ }\href {https://doi.org/10.1103/PhysRevLett.132.140201} {\bibfield  {journal} {\bibinfo  {journal} {Phys. Rev. Lett.}\ }\textbf {\bibinfo {volume} {132}},\ \bibinfo {pages} {140201} (\bibinfo {year} {2024})}\BibitemShut {NoStop}%
	\bibitem [{\citenamefont {Saha}\ and\ \citenamefont {Cabello}(2024)}]{Saha:2025XXX}%
	\BibitemOpen
	\bibfield  {author} {\bibinfo {author} {\bibfnamefont {D.}~\bibnamefont {Saha}}\ and\ \bibinfo {author} {\bibfnamefont {A.}~\bibnamefont {Cabello}},\ }\href {https://arxiv.org/abs/2501.00409} {\bibinfo {title} {Supersinglets can be self-tested with perfect quantum strategies}} (\bibinfo {year} {2024}),\ \Eprint {https://arxiv.org/abs/2501.00409} {arXiv:2501.00409 [quant-ph]} \BibitemShut {NoStop}%
	\bibitem [{\citenamefont {Stairs}(1983)}]{Stairs:1983PS}%
	\BibitemOpen
	\bibfield  {author} {\bibinfo {author} {\bibfnamefont {A.}~\bibnamefont {Stairs}},\ }\bibfield  {title} {\bibinfo {title} {Quantum logic, realism, and value definiteness},\ }\href {https://doi.org/10.1086/289140} {\bibfield  {journal} {\bibinfo  {journal} {Philos. Sci.}\ }\textbf {\bibinfo {volume} {50}},\ \bibinfo {pages} {578} (\bibinfo {year} {1983})}\BibitemShut {NoStop}%
	\bibitem [{\citenamefont {Brown}\ and\ \citenamefont {Svetlichny}(1990)}]{Brown:1990FPH}%
	\BibitemOpen
	\bibfield  {author} {\bibinfo {author} {\bibfnamefont {H.~R.}\ \bibnamefont {Brown}}\ and\ \bibinfo {author} {\bibfnamefont {G.}~\bibnamefont {Svetlichny}},\ }\bibfield  {title} {\bibinfo {title} {Nonlocality and {G}leason's lemma. {P}art~{I}. {D}eterministic theories},\ }\href {https://doi.org/10.1007/BF01883492} {\bibfield  {journal} {\bibinfo  {journal} {Found. Phys.}\ }\textbf {\bibinfo {volume} {20}},\ \bibinfo {pages} {1379} (\bibinfo {year} {1990})}\BibitemShut {NoStop}%
	\bibitem [{\citenamefont {Heywood}\ and\ \citenamefont {Redhead}(1983)}]{HR83}%
	\BibitemOpen
	\bibfield  {author} {\bibinfo {author} {\bibfnamefont {P.}~\bibnamefont {Heywood}}\ and\ \bibinfo {author} {\bibfnamefont {M.~L.~G.}\ \bibnamefont {Redhead}},\ }\bibfield  {title} {\bibinfo {title} {Nonlocality and the {K}ochen-{S}pecker paradox},\ }\href {https://doi.org/10.1007/BF00729511} {\bibfield  {journal} {\bibinfo  {journal} {Found. Phys.}\ }\textbf {\bibinfo {volume} {13}},\ \bibinfo {pages} {481} (\bibinfo {year} {1983})}\BibitemShut {NoStop}%
	\bibitem [{\citenamefont {Redhead}(1987)}]{Redhead:1987}%
	\BibitemOpen
	\bibfield  {author} {\bibinfo {author} {\bibfnamefont {M.~L.~G.}\ \bibnamefont {Redhead}},\ }\href@noop {} {\emph {\bibinfo {title} {Incompleteness, Nonlocality, and Realism}}}\ (\bibinfo  {publisher} {Oxford University Press},\ \bibinfo {address} {New York},\ \bibinfo {year} {1987})\BibitemShut {NoStop}%
	\bibitem [{\citenamefont {Elby}\ and\ \citenamefont {Jones}(1992)}]{Elby:1992PLA}%
	\BibitemOpen
	\bibfield  {author} {\bibinfo {author} {\bibfnamefont {A.}~\bibnamefont {Elby}}\ and\ \bibinfo {author} {\bibfnamefont {M.~R.}\ \bibnamefont {Jones}},\ }\bibfield  {title} {\bibinfo {title} {Weakening the locality conditions in algebraic nonlocality proofs},\ }\href {https://doi.org/10.1016/0375-9601(92)90123-4} {\bibfield  {journal} {\bibinfo  {journal} {Phys. Lett. A}\ }\textbf {\bibinfo {volume} {171}},\ \bibinfo {pages} {11} (\bibinfo {year} {1992})}\BibitemShut {NoStop}%
	\bibitem [{\citenamefont {Brunner}\ \emph {et~al.}(2014)\citenamefont {Brunner}, \citenamefont {Cavalcanti}, \citenamefont {Pironio}, \citenamefont {Scarani},\ and\ \citenamefont {Wehner}}]{Brunner:2014RMP}%
	\BibitemOpen
	\bibfield  {author} {\bibinfo {author} {\bibfnamefont {N.}~\bibnamefont {Brunner}}, \bibinfo {author} {\bibfnamefont {D.}~\bibnamefont {Cavalcanti}}, \bibinfo {author} {\bibfnamefont {S.}~\bibnamefont {Pironio}}, \bibinfo {author} {\bibfnamefont {V.}~\bibnamefont {Scarani}},\ and\ \bibinfo {author} {\bibfnamefont {S.}~\bibnamefont {Wehner}},\ }\bibfield  {title} {\bibinfo {title} {Bell nonlocality},\ }\href {https://doi.org/10.1103/RevModPhys.86.419} {\bibfield  {journal} {\bibinfo  {journal} {Rev. Mod. Phys.}\ }\textbf {\bibinfo {volume} {86}},\ \bibinfo {pages} {419} (\bibinfo {year} {2014})}\BibitemShut {NoStop}%
	\bibitem [{\citenamefont {Paddock}(2022)}]{Paddock:2022XXX}%
	\BibitemOpen
	\bibfield  {author} {\bibinfo {author} {\bibfnamefont {C.}~\bibnamefont {Paddock}},\ }\href@noop {} {\bibinfo {title} {Rounding near-optimal quantum strategies for nonlocal games to strategies using maximally entangled states}} (\bibinfo {year} {2022}),\ \Eprint {https://arxiv.org/abs/2203.02525} {arXiv:2203.02525 [quant-ph]} \BibitemShut {NoStop}%
	\bibitem [{\citenamefont {Mančinska}(2007)}]{Mancinska:2007}%
	\BibitemOpen
	\bibfield  {author} {\bibinfo {author} {\bibfnamefont {L.}~\bibnamefont {Mančinska}},\ }\href {http://home.lu.lv/~sd20008/papers/essays/Kochen%20Specker%20[paper].pdf} {\bibinfo {title} {Kochen-{S}pecker {T}heorem and {G}ames}} (\bibinfo {year} {2007})\BibitemShut {NoStop}%
	\bibitem [{\citenamefont {Coladangelo}(2017)}]{Coladangelo2016arxiv}%
	\BibitemOpen
	\bibfield  {author} {\bibinfo {author} {\bibfnamefont {A.}~\bibnamefont {Coladangelo}},\ }\bibfield  {title} {\bibinfo {title} {Parallel self-testing of (tilted) {EPR} pairs via copies of (tilted) {CHSH} and the magic square game},\ }\href {https://doi.org/10.26421/QIC17.9-10-6} {\bibfield  {journal} {\bibinfo  {journal} {Quantum Inf. Comput.}\ }\textbf {\bibinfo {volume} {17}},\ \bibinfo {pages} {831} (\bibinfo {year} {2017})}\BibitemShut {NoStop}%
	\bibitem [{\citenamefont {Coudron}\ and\ \citenamefont {Natarajan}(2016)}]{Coudron2016arxiv}%
	\BibitemOpen
	\bibfield  {author} {\bibinfo {author} {\bibfnamefont {M.}~\bibnamefont {Coudron}}\ and\ \bibinfo {author} {\bibfnamefont {A.}~\bibnamefont {Natarajan}},\ }\href@noop {} {\bibinfo {title} {The parallel-repeated magic square game is rigid}} (\bibinfo {year} {2016}),\ \Eprint {https://arxiv.org/abs/1609.06306} {arXiv:1609.06306 [quant-ph]} \BibitemShut {NoStop}%
	\bibitem [{\citenamefont {Ara{\'{u}}jo}\ \emph {et~al.}(2020)\citenamefont {Ara{\'{u}}jo}, \citenamefont {Hirsch},\ and\ \citenamefont {Quintino}}]{Araujo:2020Quantum}%
	\BibitemOpen
	\bibfield  {author} {\bibinfo {author} {\bibfnamefont {M.}~\bibnamefont {Ara{\'{u}}jo}}, \bibinfo {author} {\bibfnamefont {F.}~\bibnamefont {Hirsch}},\ and\ \bibinfo {author} {\bibfnamefont {M.~T.}\ \bibnamefont {Quintino}},\ }\bibfield  {title} {\bibinfo {title} {{B}ell nonlocality with a single shot},\ }\href {https://doi.org/10.22331/q-2020-10-28-353} {\bibfield  {journal} {\bibinfo  {journal} {{Quantum}}\ }\textbf {\bibinfo {volume} {4}},\ \bibinfo {pages} {353} (\bibinfo {year} {2020})}\BibitemShut {NoStop}%
	\bibitem [{\citenamefont {Schrijver}(2003)}]{Schrijver:2003}%
	\BibitemOpen
	\bibfield  {author} {\bibinfo {author} {\bibfnamefont {A.}~\bibnamefont {Schrijver}},\ }\href {https://homepages.cwi.nl/~lex/files/book.pdf} {\emph {\bibinfo {title} {Combinatorial Optimization: Polyhedra and Efficiency}}},\ Vol.~\bibinfo {volume} {24}\ (\bibinfo  {publisher} {Springer},\ \bibinfo {address} {Berlin},\ \bibinfo {year} {2003})\BibitemShut {NoStop}%
	\bibitem [{\citenamefont {Kirchweger}\ \emph {et~al.}(2023)\citenamefont {Kirchweger}, \citenamefont {Peitl},\ and\ \citenamefont {Szeider}}]{Kirchweger:2023}%
	\BibitemOpen
	\bibfield  {author} {\bibinfo {author} {\bibfnamefont {M.}~\bibnamefont {Kirchweger}}, \bibinfo {author} {\bibfnamefont {T.}~\bibnamefont {Peitl}},\ and\ \bibinfo {author} {\bibfnamefont {S.}~\bibnamefont {Szeider}},\ }\bibfield  {title} {\bibinfo {title} {Co-certificate learning with {SAT} modulo symmetries},\ }in\ \href {https://doi.org/10.24963/ijcai.2023/216} {\emph {\bibinfo {booktitle} {Proceedings of the Thirty-Second International Joint Conference on Artificial Intelligence, {IJCAI-23}}}},\ \bibinfo {editor} {edited by\ \bibinfo {editor} {\bibfnamefont {E.}~\bibnamefont {Elkind}}}\ (\bibinfo  {publisher} {International Joint Conferences on Artificial Intelligence Organization},\ \bibinfo {address} {Marina del Ray, CA},\ \bibinfo {year} {2023})\ pp.\ \bibinfo {pages} {1944--1953}\BibitemShut {NoStop}%
	\bibitem [{\citenamefont {Li}\ \emph {et~al.}(2024)\citenamefont {Li}, \citenamefont {Bright},\ and\ \citenamefont {Ganesh}}]{Li:2023XXX}%
	\BibitemOpen
	\bibfield  {author} {\bibinfo {author} {\bibfnamefont {Z.}~\bibnamefont {Li}}, \bibinfo {author} {\bibfnamefont {C.}~\bibnamefont {Bright}},\ and\ \bibinfo {author} {\bibfnamefont {V.}~\bibnamefont {Ganesh}},\ }\bibfield  {title} {\bibinfo {title} {A {SAT} solver + computer algebra attack on the minimum {K}ochen–{S}pecker problem},\ }in\ \href {https://doi.org/10.24963/ijcai.2024/210} {\emph {\bibinfo {booktitle} {Proceedings of the Thirty-Third International Joint Conference on Artificial Intelligence, {IJCAI-24}}}},\ \bibinfo {editor} {edited by\ \bibinfo {editor} {\bibfnamefont {K.}~\bibnamefont {Larson}}}\ (\bibinfo  {publisher} {International Joint Conferences on Artificial Intelligence Organization},\ \bibinfo {address} {Marina del Ray, CA},\ \bibinfo {year} {2024})\ pp.\ \bibinfo {pages} {1898--1906}\BibitemShut {NoStop}%
	\bibitem [{\citenamefont {{The Sage Developers}}(2024)}]{sagemath}%
	\BibitemOpen
	\bibfield  {author} {\bibinfo {author} {\bibnamefont {{The Sage Developers}}},\ }\href {https://www.sagemath.org} {\emph {\bibinfo {title} {{S}ageMath, the {S}age {M}athematics {S}oftware {S}ystem ({V}ersion 10.3)}}} (\bibinfo {year} {2024})\BibitemShut {NoStop}%
	\bibitem [{\citenamefont {Cabello}\ \emph {et~al.}(1996)\citenamefont {Cabello}, \citenamefont {Estebaranz},\ and\ \citenamefont {Garc{\'\i}a-Alcaine}}]{Cabello:1996PLA}%
	\BibitemOpen
	\bibfield  {author} {\bibinfo {author} {\bibfnamefont {A.}~\bibnamefont {Cabello}}, \bibinfo {author} {\bibfnamefont {J.~M.}\ \bibnamefont {Estebaranz}},\ and\ \bibinfo {author} {\bibfnamefont {G.}~\bibnamefont {Garc{\'\i}a-Alcaine}},\ }\bibfield  {title} {\bibinfo {title} {Bell-{K}ochen-{S}pecker theorem: A proof with 18 vectors},\ }\href {https://doi.org/10.1016/0375-9601(96)00134-X} {\bibfield  {journal} {\bibinfo  {journal} {Phys. Lett. A}\ }\textbf {\bibinfo {volume} {212}},\ \bibinfo {pages} {183} (\bibinfo {year} {1996})}\BibitemShut {NoStop}%
	\bibitem [{\citenamefont {Pavi{\v{c}}i{\'c}}\ \emph {et~al.}(2010)\citenamefont {Pavi{\v{c}}i{\'c}}, \citenamefont {Megill},\ and\ \citenamefont {Merlet}}]{PavicicMegillMerlet2010}%
	\BibitemOpen
	\bibfield  {author} {\bibinfo {author} {\bibfnamefont {M.}~\bibnamefont {Pavi{\v{c}}i{\'c}}}, \bibinfo {author} {\bibfnamefont {N.~D.}\ \bibnamefont {Megill}},\ and\ \bibinfo {author} {\bibfnamefont {J.-P.}\ \bibnamefont {Merlet}},\ }\bibfield  {title} {\bibinfo {title} {New {K}ochen-{S}pecker sets in four dimensions},\ }\href {https://doi.org/10.1016/j.physleta.2010.03.019} {\bibfield  {journal} {\bibinfo  {journal} {Phys. Lett. A}\ }\textbf {\bibinfo {volume} {374}},\ \bibinfo {pages} {2122} (\bibinfo {year} {2010})}\BibitemShut {NoStop}%
	\bibitem [{\citenamefont {Peres}(1991)}]{Peres:1991JPA}%
	\BibitemOpen
	\bibfield  {author} {\bibinfo {author} {\bibfnamefont {A.}~\bibnamefont {Peres}},\ }\bibfield  {title} {\bibinfo {title} {Two simple proofs of the {K}ochen-{S}pecker theorem},\ }\href {https://doi.org/10.1088/0305-4470/24/4/003} {\bibfield  {journal} {\bibinfo  {journal} {J. Phys. A: Math. Gen.}\ }\textbf {\bibinfo {volume} {24}},\ \bibinfo {pages} {L175} (\bibinfo {year} {1991})}\BibitemShut {NoStop}%
	\bibitem [{\citenamefont {Kernaghan}(1994)}]{Kernaghan:1994JPA}%
	\BibitemOpen
	\bibfield  {author} {\bibinfo {author} {\bibfnamefont {M.}~\bibnamefont {Kernaghan}},\ }\bibfield  {title} {\bibinfo {title} {{B}ell-{K}ochen-{S}pecker theorem for 20 vectors},\ }\href {https://doi.org/10.1088/0305-4470/27/21/007} {\bibfield  {journal} {\bibinfo  {journal} {J. Phys. A: Math. Gen.}\ }\textbf {\bibinfo {volume} {27}},\ \bibinfo {pages} {L829} (\bibinfo {year} {1994})}\BibitemShut {NoStop}%
	\bibitem [{\citenamefont {Penrose}(2000)}]{Penrose:2000}%
	\BibitemOpen
	\bibfield  {author} {\bibinfo {author} {\bibfnamefont {R.}~\bibnamefont {Penrose}},\ }\bibfield  {title} {\bibinfo {title} {On {B}ell non-locality without probabilities: Some curious geometry},\ }in\ \href@noop {} {\emph {\bibinfo {booktitle} {Quantum Reflections}}},\ \bibinfo {editor} {edited by\ \bibinfo {editor} {\bibfnamefont {J.}~\bibnamefont {Ellis}}\ and\ \bibinfo {editor} {\bibfnamefont {D.}~\bibnamefont {Amati}}}\ (\bibinfo  {publisher} {Cambridge University Press},\ \bibinfo {address} {Cambridge, UK},\ \bibinfo {year} {2000})\ pp.\ \bibinfo {pages} {1--27}\BibitemShut {NoStop}%
	\bibitem [{\citenamefont {Zimba}\ and\ \citenamefont {Penrose}(1993)}]{Zimba:1993SHPS}%
	\BibitemOpen
	\bibfield  {author} {\bibinfo {author} {\bibfnamefont {J.~R.}\ \bibnamefont {Zimba}}\ and\ \bibinfo {author} {\bibfnamefont {R.}~\bibnamefont {Penrose}},\ }\bibfield  {title} {\bibinfo {title} {On {B}ell non-locality without probabilities: More curious geometry},\ }\href {https://doi.org/10.1016/0039-3681(93)90061-N} {\bibfield  {journal} {\bibinfo  {journal} {Stud. Hist. Phil. Sci. B}\ }\textbf {\bibinfo {volume} {24}},\ \bibinfo {pages} {697} (\bibinfo {year} {1993})}\BibitemShut {NoStop}%
	\bibitem [{\citenamefont {Cabello}(1996)}]{Cabello:1996}%
	\BibitemOpen
	\bibfield  {author} {\bibinfo {author} {\bibfnamefont {A.}~\bibnamefont {Cabello}},\ }\emph {\bibinfo {title} {Pruebas algebraicas de imposibilidad de variables ocultas en mec\'anica cu\'antica}},\ \href {https://hdl.handle.net/20.500.14352/62875} {Ph.D. thesis},\ \bibinfo  {school} {Universidad Complutense de Madrid}, \bibinfo {address} {Madrid} (\bibinfo {year} {1996})\BibitemShut {NoStop}%
	\bibitem [{\citenamefont {Peres}(1993)}]{Peres:1993}%
	\BibitemOpen
	\bibfield  {author} {\bibinfo {author} {\bibfnamefont {A.}~\bibnamefont {Peres}},\ }\href {https://doi.org/https://doi.org/10.1007/0-306-47120-5} {\emph {\bibinfo {title} {Quantum Theory: Concepts and Methods}}}\ (\bibinfo  {publisher} {Kluwer},\ \bibinfo {address} {Dordrecht},\ \bibinfo {year} {1993})\BibitemShut {NoStop}%
	\bibitem [{\citenamefont {Bub}(1996)}]{Bub:1996FP}%
	\BibitemOpen
	\bibfield  {author} {\bibinfo {author} {\bibfnamefont {J.}~\bibnamefont {Bub}},\ }\bibfield  {title} {\bibinfo {title} {Sch\"{u}tte's tautology and the {K}ochen-{S}pecker theorem},\ }\href {https://doi.org/10.1007/bf02058633} {\bibfield  {journal} {\bibinfo  {journal} {Found. Phys.}\ }\textbf {\bibinfo {volume} {26}},\ \bibinfo {pages} {787} (\bibinfo {year} {1996})}\BibitemShut {NoStop}%
	\bibitem [{\citenamefont {Gould}\ and\ \citenamefont {Aravind}(2010)}]{gould2010isomorphism}%
	\BibitemOpen
	\bibfield  {author} {\bibinfo {author} {\bibfnamefont {E.}~\bibnamefont {Gould}}\ and\ \bibinfo {author} {\bibfnamefont {P.~K.}\ \bibnamefont {Aravind}},\ }\bibfield  {title} {\bibinfo {title} {Isomorphism between the {P}eres and {P}enrose proofs of the {B}{K}{S} theorem in three dimensions},\ }\href {https://doi.org/10.1007/s10701-010-9434-2} {\bibfield  {journal} {\bibinfo  {journal} {Found. Phys.}\ }\textbf {\bibinfo {volume} {40}},\ \bibinfo {pages} {1096} (\bibinfo {year} {2010})}\BibitemShut {NoStop}%
	\bibitem [{\citenamefont {Bengtsson}(2012)}]{bengtsson2012gleason}%
	\BibitemOpen
	\bibfield  {author} {\bibinfo {author} {\bibfnamefont {I.}~\bibnamefont {Bengtsson}},\ }\bibfield  {title} {\bibinfo {title} {Gleason, {K}ochen-{S}pecker, and a competition that never was},\ }in\ \href {https://doi.org/10.1063/1.4773124} {\emph {\bibinfo {booktitle} {AIP Conf. Proc. of the Quantum Theory of Reconsideration of Foundations 6}}},\ Vol.\ \bibinfo {volume} {1508},\ \bibinfo {editor} {edited by\ \bibinfo {editor} {\bibfnamefont {A.}~\bibnamefont {Khrennikov}}, \bibinfo {editor} {\bibfnamefont {H.}~\bibnamefont {Atmanspacher}}, \bibinfo {editor} {\bibfnamefont {A.}~\bibnamefont {Migdall}},\ and\ \bibinfo {editor} {\bibfnamefont {S.}~\bibnamefont {Polyakov}}}\ (\bibinfo {organization} {American Institute of Physics},\ \bibinfo {address} {Melville, NY},\ \bibinfo {year} {2012})\ pp.\ \bibinfo {pages} {125--135}\BibitemShut {NoStop}%
	\bibitem [{\citenamefont {Kernaghan}\ and\ \citenamefont {Peres}(1995)}]{Kernaghan:1995PLA}%
	\BibitemOpen
	\bibfield  {author} {\bibinfo {author} {\bibfnamefont {M.}~\bibnamefont {Kernaghan}}\ and\ \bibinfo {author} {\bibfnamefont {A.}~\bibnamefont {Peres}},\ }\bibfield  {title} {\bibinfo {title} {{K}ochen-{S}pecker theorem for eight-dimensional space},\ }\href {https://doi.org/10.1016/0375-9601(95)00012-R} {\bibfield  {journal} {\bibinfo  {journal} {Phys. Lett. A}\ }\textbf {\bibinfo {volume} {198}},\ \bibinfo {pages} {1} (\bibinfo {year} {1995})}\BibitemShut {NoStop}%
	\bibitem [{\citenamefont {Cabello}\ \emph {et~al.}(2005)\citenamefont {Cabello}, \citenamefont {Estebaranz},\ and\ \citenamefont {Garc{\'\i}a-Alcaine}}]{Cabello:2005PLA}%
	\BibitemOpen
	\bibfield  {author} {\bibinfo {author} {\bibfnamefont {A.}~\bibnamefont {Cabello}}, \bibinfo {author} {\bibfnamefont {J.~M.}\ \bibnamefont {Estebaranz}},\ and\ \bibinfo {author} {\bibfnamefont {G.}~\bibnamefont {Garc{\'\i}a-Alcaine}},\ }\bibfield  {title} {\bibinfo {title} {Recursive proof of the {B}ell--{K}ochen--{S}pecker theorem in any dimension n\textgreater{}3},\ }\href {https://doi.org/10.1016/j.physleta.2005.03.067} {\bibfield  {journal} {\bibinfo  {journal} {Phys. Lett. A}\ }\textbf {\bibinfo {volume} {339}},\ \bibinfo {pages} {425} (\bibinfo {year} {2005})}\BibitemShut {NoStop}%
	\bibitem [{\citenamefont {Trandafir}\ \emph {et~al.}(2024)\citenamefont {Trandafir}, \citenamefont {Gonzales-Ureta},\ and\ \citenamefont {Cabello}}]{trandafir2024perfectquantumstrategiessmall}%
	\BibitemOpen
	\bibfield  {author} {\bibinfo {author} {\bibfnamefont {S.}~\bibnamefont {Trandafir}}, \bibinfo {author} {\bibfnamefont {J.~R.}\ \bibnamefont {Gonzales-Ureta}},\ and\ \bibinfo {author} {\bibfnamefont {A.}~\bibnamefont {Cabello}},\ }\href@noop {} {\bibinfo {title} {Perfect quantum strategies with small input cardinality}} (\bibinfo {year} {2024}),\ \Eprint {https://arxiv.org/abs/2407.21473} {arXiv:2407.21473 [quant-ph]} \BibitemShut {NoStop}%
	\bibitem [{\citenamefont {McKay}(1998)}]{Mckay1998}%
	\BibitemOpen
	\bibfield  {author} {\bibinfo {author} {\bibfnamefont {B.~D.}\ \bibnamefont {McKay}},\ }\bibfield  {title} {\bibinfo {title} {Isomorph-free exhaustive generation},\ }\href {https://doi.org/https://doi.org/10.1006/jagm.1997.0898} {\bibfield  {journal} {\bibinfo  {journal} {J. Algorithms}\ }\textbf {\bibinfo {volume} {26}},\ \bibinfo {pages} {306} (\bibinfo {year} {1998})}\BibitemShut {NoStop}%
\end{thebibliography}

%


\end{document}